\documentclass[ba, noinfoline]{imsart}

\usepackage{amsmath}
\usepackage{amsthm}
\usepackage{amssymb} 
\usepackage{graphicx}
\usepackage{subcaption}
\RequirePackage{natbib}
\usepackage[colorlinks,citecolor=blue,urlcolor=black,filecolor=blue,backref=page]{hyperref}
\usepackage{bm}
\usepackage{dsfont} 
\usepackage{placeins} 
\usepackage{enumitem} 
\usepackage{xcolor}
\usepackage[norefs, nocites]{refcheck}
\usepackage{todonotes}
\DeclareMathAlphabet{\mathpzc}{OT1}{pzc}{m}{it} 
\usepackage{chngcntr}
\newtheorem{theorem}{Theorem}[section]
\newtheorem{lemma}[theorem]{Lemma}
\newtheorem{proposition}[theorem]{Proposition}

\theoremstyle{remark}
\newtheorem{remark}[theorem]{Remark}
\newtheorem*{remark*}{Remark}

\newcommand{\prob}{\mathbb{P}}
\newcommand{\expectation}{\mathbb{E}}

\newcommand{\given}[1][]{\, #1| \,}
\newcommand{\ind}{\mathds{1}}
\newcommand{\eqDistribution}{\mathrel{\raisebox{-.2ex}{$\overset{\scalebox{.6}{$\, d$}}{=}$}}} 

\newcommand{\normal}{\mathcal{N}}

\newcommand{\logit}{\textrm{logit}}

\newcommand{\bernoulli}{\textrm{Bernoulli}}

\newcommand{\transpose}{\text{\raisebox{.5ex}{$\intercal$}}}

\newcommand{\eigen}{\nu}

\newcommand{\y}{\bm{y}}
\newcommand{\x}{\bm{x}}
\newcommand{\z}{\bm{z}}
\newcommand{\e}{\bm{e}} 
\newcommand{\bA}{\bm{A}}
\newcommand{\bB}{\bm{B}}
\newcommand{\X}{\bm{X}}
\newcommand{\I}{\bm{I}}
\newcommand{\bv}{\bm{v}}
\newcommand{\kummerConfluentHyperGeom}{M}
\newcommand{\betaFun}{B}

\newcommand{\lshrink}{\lambda}
\newcommand{\gshrink}{\tau}
\newcommand{\Lshrink}{\Lambda}
\newcommand{\bbeta}{\bm{\beta}}
\newcommand{\bhatbeta}{\bm{\hat{\beta}}}
\newcommand{\bmu}{\bm{\mu}}
\newcommand{\bomega}{\bm{\omega}}
\newcommand{\blshrink}{\bm{\lshrink}}
\newcommand{\bLshrink}{\bm{\Lshrink}}
\newcommand{\bOmega}{\bm{\Omega}}

\newcommand{\bPhi}{\bm{\Phi}}
\newcommand{\bSigma}{\bm{\Sigma}}
\newcommand{\btheta}{\bm{\theta}}
\newcommand{\bTheta}{\bm{\Theta}}
\newcommand{\bphi}{\bm{\phi}}

\newcommand{\localPriorInAbstract}{\pi_{\rm local}}
\newcommand{\localPrior}{\pi_{\rm loc}}

\newcommand{\globalPrior}{\pi_{\rm glo}}
\newcommand{\bridgeExponent}{a} 
\newcommand{\betaPrev}[1][]{\beta^{*#1}} 
\newcommand{\bbetaPrev}{\bbeta^*}

\newcommand{\blshrinkPrev}{\blshrink^*}
\newcommand{\gshrinkPrev}{\gshrink^*}
\newcommand{\bomegaPrev}{\bomega^*}
\newcommand{\betaNext}{\beta}
\newcommand{\bbetaNext}{\bbeta}
\newcommand{\kernel}{P}
\newcommand{\lyapunovExponent}{\alpha} 
\newcommand{\driftContraction}{\gamma}
\newcommand{\driftConst}{b}
\newcommand{\const}{C}
\newcommand{\smallSet}{S}
\newcommand{\stableDistScale}{c}
\newcommand{\stableDensity}{\pi_{\rm st}}

\newcommand{\polyagamma}{P{\'o}lya-Gamma}

\newcommand{\diff}{{\rm d}}
\newcommand{\yesnumber}{\addtocounter{equation}{1}\tag{\theequation}}

\renewcommand{\theequation}{\arabic{section}.\arabic{equation}}
\renewcommand{\thefigure}{\arabic{section}.\arabic{figure}}
\counterwithin*{equation}{section}
\counterwithin*{figure}{section}
\firstpage{0}
\lastpage{0}

\startlocaldefs
\endlocaldefs

\begin{document}

\begin{frontmatter}

\title{Shrinkage with shrunken shoulders: \\ Gibbs sampling shrinkage model posteriors with guaranteed convergence rates}
\runtitle{Regularized shrinkage and ergodicity of Gibbs sampler}


\begin{aug}

\author{
	\fnms{Akihiko} \snm{Nishimura} $^*$
	\ead[label=e1]{aki.nishimura@jhu.edu}
	\ead[label=e11, url]{https://aki-nishimura.github.io}
}
\address{
	$^*$ Department of Biostatistics,
	Johns Hopkins University.
	\printead{e1}
}
\author{
	\fnms{Marc A.} \snm{Suchard} $^\dagger$
	\ead[label=e2]{msuchard@ucla.edu}
}
\address{
	$^\dagger$ Departments of Biomathematics, Biostatistics, and Human Genetics, University of California -- Los Angeles.
	\printead{e2}
}

\runauthor{A.~Nishimura and M.A.~Suchard}

\end{aug}

\begin{abstract}
Use of continuous shrinkage priors --- with a ``spike'' near zero and heavy-tails towards infinity ---  is an increasingly popular approach to induce sparsity in parameter estimates.
When the parameters are only weakly identified by the likelihood, however, the posterior may end up with tails as heavy as the prior, jeopardizing robustness of inference.
A natural solution is to ``shrink the shoulders'' of a shrinkage prior by lightening up its tails beyond a reasonable parameter range, yielding a \textit{regularized} version of the prior.
We develop a regularization approach which, unlike previous proposals, preserves computationally attractive structures of original shrinkage priors.
We study theoretical properties of the Gibbs sampler on resulting posterior distributions, with emphasis on convergence rates of the \polyagamma{} Gibbs sampler for sparse logistic regression.
Our analysis shows that the proposed regularization leads to geometric ergodicity under a broad range of global-local shrinkage priors.
Essentially, the only requirement is for the prior $\localPriorInAbstract(\cdot)$ on the local scale $\lshrink$ to satisfy $\localPriorInAbstract(0) < \infty$.
If $\localPriorInAbstract(\cdot)$ further satisfies $\lim_{\lshrink \to 0} \localPriorInAbstract(\lshrink) / \lshrink^\bridgeExponent < \infty$ for $\bridgeExponent > 0$, as in the case of Bayesian bridge priors, we show the sampler to be uniformly ergodic.
\end{abstract}





\end{frontmatter}

\section{Introduction}
\label{sec:intro}
Bayesian modelers are increasingly adopting continuous shrinkage priors to control the effective number of parameters and model complexity in a data-driven manner.
These priors are designed to shrink most of the parameters towards zero while allowing for the likelihood to pull a small fraction of them away from zero.
To achieve such effects, a shrinkage prior has a density with a ``spike'' near zero and heavy-tails towards infinity, encoding information that parameter values are likely close to zero but otherwise could be anywhere.
Originally developed for the purpose of sparse regression \citep{carvalho2009sparsity_via_horseshoe}, shrinkage priors have found applications in trend filtering of time series data \citep{kowal2019dynamic_shrinkage_processes}, (dynamic) factor models \citep{kastner2019sparse_dynamic_covariance_estimation}, graphical models \citep{li2019graphical_horseshoe}, compression of deep neural networks \citep{louizos2017bayesian_compression_for_deep_learning}, among others. 

Shrinkage priors are often expressed as a scale mixture of Gaussians on the unknown parameter $\bbeta = (\beta_1, \ldots, \beta_p)$ \citep{polson2010global_local}:
\begin{equation}
\label{eq:global_local_prior}
\pi(\beta_j \given \gshrink, \lshrink_j)
	\sim \normal(0, \gshrink^2 \lshrink_j^2), \
	\lshrink_j \sim \localPrior(\cdot).
\end{equation}
This \textit{global-local} representation simplifies the posterior conditionals and lead to straightforward inference via Gibbs sampling.
The \textit{global scale} $\gshrink$ controls the average magnitude of $\beta_j$'s and hence overall sparsity level.
The \textit{local scale} $\lshrink_j$ is specific to individual $\beta_j$ and its density $\localPrior(\cdot)$ controls the size of the spike and tail behavior of the marginal $\beta_j \given \gshrink$.
For instance, the popular \textit{horseshoe} prior of \cite{carvalho2010horseshoe} uses $\localPrior(\lshrink) \propto (1 + \lshrink^2)^{-1}$, inducing a marginal $\pi(\beta_j \given \gshrink)$ with the spike proportional to $- \log(|\beta_j / \gshrink|)$ as $| \beta_j / \gshrink | \to 0$ and the tail proportional to $(\beta_j / \gshrink)^{-2}$ as $| \beta_j / \gshrink | \to \infty$.
Another notable example is the Bayesian bridge prior of \cite{polson2014bayes_bridge}, which generalizes the Bayesian lasso of \cite{park2008bayesian_lasso} with $\pi(\beta_j \given \gshrink)$ having a larger spike as $| \beta_j / \gshrink | \to 0$ and heavier tails as $| \beta_j / \gshrink | \to \infty$.
Most importantly from the computational efficiency perspective, 
the bridge prior possesses a closed-form expression $\pi(\beta_j \given \gshrink) \propto \exp(-|\beta_j / \gshrink |^{\bridgeExponent})$ for $\bridgeExponent \in (0, 1)$ and thus allows for a collapsed Gibbs update from $\gshrink \given \bbeta$ with $\lshrink_j$'s marginalized out.

For a simple purpose such as estimating the unknown means of independent Gaussian observations, a broad class of shrinkage priors achieve theoretically optimal performance \citep{vanDerPas2016sparse_mean_contraction, ghosh2017shrinkage_prior_optimality}.
The lack of prior information in the tail of the distribution is problematic, however, in more complex models where parameters are only weakly identified.
In such models, the posterior may have a tail as heavy as the prior, resulting in unreliable parameter estimates \citep{ghosh2018cauchy_for_logit}.

To address the above shortcoming of shrinkage priors, we build on the work of \cite{piironen2017regularized-horseshoe} and propose a computationally convenient way to \textit{regularize} shrinkage priors.
The basic idea is to modify the prior so that the marginal distribution of $|\beta_j|$ has light-tails beyond a reasonable range.
Our formulation has computational advantages over that of \cite{piironen2017regularized-horseshoe} due to a subtle yet important difference.
By preserving the global-local structure \eqref{eq:global_local_prior}, our regularized shrinkage priors can benefit from partial marginalization approaches that substantially improve mixing of Gibbs samplers (\citealt{polson2014bayes_bridge, johndrow2018scalable_mcmc}; Appendix~\ref{sec:bridge_prior_properties}).
In addition, our regularization leaves the posterior conditionals of $\lshrink_j$'s unchanged, allowing their conditional updates via existing specialized samplers (\citealt{griffin2010normal-gamma, polson2014bayes_bridge}; Appendix~\ref{sec:local_scale_rejection_sampler}).\footnote{
	Appendix~\ref{sec:local_scale_rejection_sampler} describes a simple and provably efficient rejection-sampler for the conditional distributions of local scale parameter $\lshrink_j$'s under the horseshoe prior.
	Despite the horseshoe's popularity, we find that no existing algorithm for the conditional update comes with theoretically guaranteed efficiency.
}

Our regularized shrinkage priors allow for posterior inference via Gibbs sampler whose convergence rates often are provably fast.
As an illustrative example, we consider Bayesian sparse logistic regression models, whose need for regularization motivated the work of \citet{piironen2017regularized-horseshoe}.
Gibbs sampling via the P{\'o}lya-Gamma data augmentation of \cite{polson2013polya_gamma} is a state-of-the-art approach to posterior computation under logistic model.
When combined with advanced numerical linear algebra techniques, this Gibbs sampler is highly scalable to large data sets \citep{nishimura2018cg-accelerated-gibbs}, but its theoretical convergence rate has not been investigated.
Assuming that the prior density $\localPrior(\lshrink)$ is continuous and bounded except possibly at $\lshrink = 0$, we establish that the Gibbs sampler is geometrically ergodic whenever $\localPrior(0) < \infty$.
Stronger uniform convergence is achieved when $\int \lshrink^{-1} \localPrior(\lshrink) \, \diff \lshrink < \infty$.
The integrability condition holds in particular when $\localPrior(\lshrink) = O(\lshrink^\bridgeExponent)$ for $\bridgeExponent > 0$ as $\lshrink \to 0$, which is the case for normal-gamma priors with shape parameter larger than $1/2$ \citep{griffin2010normal-gamma} and for Bayesian bridge priors (\citealt{polson2014bayes_bridge} and Appendix~\ref{sec:bridge_prior_properties}).

Previous studies of the convergence rates under shrinkage models have focused exclusively on linear regression with specific parametric families of shrinkage priors \citep{pal2014ergodicity_shrinkage_models, johndrow2018scalable_mcmc}.
In contrast, our analysis requires no parametric assumptions on the shrinkage prior, at the same time extending the convergence results to the logistic model and, in Appendix~\ref{sec:further_results_on_general_shrinkage_model_Gibbs_behavior}, to the probit model.


To summarize, this work provides two major contributions to the Bayesian shrinkage literature.
First, we propose an effective and Gibbs-friendly approach to suitably modify shrinkage priors for use in weakly-identifiable models (Section~\ref{sec:gibbs_friendly_regularization}).
Second, we develop theoretical tools to study the behavior of shrinkage model Gibbs samplers near the spike $\beta_j = 0$ without any parametric assumption on $\pi_{\rm loc}(\cdot)$, thereby unifying convergence analyses of the logistic regression Gibbs samplers under a range of shrinkage priors (Section~\ref{sec:ergodicity_under_sparse_logit}).
We conclude the article in Section~\ref{sec:simulation} by demonstrating a practical use case of regularized shrinkage models via simulation study, which emulates increasingly common situations where the sample sizes are large yet the signals are difficult to detect.


\section{Regularized shrinkage prior}
\label{sec:gibbs_friendly_regularization}
\cite{piironen2017regularized-horseshoe} proposes to control the tail behavior of a global-local shrinkage prior by defining its regularized version with \textit{slab width} $\zeta > 0$ as
\begin{equation}
\label{eq:regularized_global_local_mixture}
\beta_j \given \gshrink, \lshrink_j, \zeta
	\sim \normal \! \left(0,
		\left( \frac{1}{\zeta^2} + \frac{1}{\gshrink^2 \lshrink_j^2} \right)^{-1}
	\right),
\end{equation}
with the prior $\localPrior(\cdot)$ on the local scale $\lshrink_j$ unmodified.
This regularization ensures that the variance of $\beta_j \given \gshrink, \lshrink_j, \zeta$ is upper bounded by $\zeta^{2}$ and hence $\beta_j \given \zeta$ marginally has a density with Gaussian tails beyond $|\beta_j| > \zeta$.
The slab width $\zeta$ can be either given a prior distribution or fixed at a reasonable value.\footnote{%
	While an appropriate choice of $\zeta$ is application specific, by way of illustration, we suggest $\zeta = 2$ as a weakly informative and sensible starting point in biomedical applications with standardized predictors.
	\cite{schuemie2018high_throughput_observational_studies} surveys 59,196 published effect estimates in the observational study literature and finds only a small portion of them exceeds 2.
}

While beneficial in improving statistical properties \citep{piironen2017regularized-horseshoe}, regularization the form \eqref{eq:regularized_global_local_mixture} compromises the posterior conditional structures of shrinkage models.
Specifically, the conditional distribution of $\gshrink, \blshrink$ is altered through their dependency on $\zeta$.
This structural change is at best an inconvenience and potentially a cause of computational inefficiency, prohibiting the use of common acceleration techniques.
For instance, the global scale $\gshrink$ is known to mix slowly when updating from its full conditional, so the state-of-the-art Gibbs samplers for Bayesian sparse regression marginalize out a subset of parameters when updating $\gshrink$ \citep{johndrow2018scalable_mcmc, nishimura2018cg-accelerated-gibbs}.
The analytical tractabilities of the integrals, which these marginalization strategies rely on, is lost when using the regularization as in \eqref{eq:regularized_global_local_mixture}.

We propose a more computationally convenient formulation, which induces regularization similar to that of \eqref{eq:regularized_global_local_mixture} while keeping $\gshrink$ and $\blshrink$ conditionally independent of $\zeta$ given $\bbeta$.
Intuitively, we achieve regularization indirectly through fictitious data that makes values $|\beta_j| \gg \zeta$ unlikely.
The use of such fictitious data is technically unnecessary in defining our regularization strategy (Appendix~\ref{sec:alt_regularization}), but makes the mechanism and resulting posterior properties more transparent.

We visually illustrate in Figure~\ref{fig:comparison_of_former_and_proposed_regularization} the construction of our regularized prior as well as the corresponding posterior structure when data $\y$ and $\X$ inform $\bbeta$ through the likelihood $L(\y \given \X, \bbeta)$.
Given a global-local prior $\beta_j \given \gshrink, \lshrink_j \sim \normal(0, \gshrink^2 \lshrink_j^2)$, we introduce fictitious data $z_j$ whose realized value and underlying distribution are assumed to be
\begin{equation}
\label{eq:gibbs_friendly_regularized_prior}
z_j  =  0, \ \
	z_j \given \beta_j, \zeta \sim \normal(\beta_j, \zeta^2)
\end{equation}
for $j = 1, \ldots, p$.
We then define the regularized prior as the distribution of $\beta_j$ conditional on $z_j = 0$.
Under this model, the distribution of $\beta_j \given \gshrink, \lshrink_j, \zeta, z_j = 0$ coincides with that of \eqref{eq:regularized_global_local_mixture}.
On the other hand, the scale parameters $\gshrink, \blshrink$ are conditionally independent of the others given $\bbeta$, so that the posterior full conditional $\gshrink, \blshrink \given \bbeta, \zeta, \z, \y, \X$ ($\eqDistribution \gshrink, \blshrink \given \bbeta$) has the same density as in the unregularized version.
Our regularization thus allows the Gibbs sampler to update $\gshrink, \blshrink$ with the exact same algorithm as the one designed for the original shrinkage prior.
We summarize our discussion as Proposition~\ref{prop:gibbs_friendly_regularized_prior} below.
\begin{proposition}
\label{prop:gibbs_friendly_regularized_prior}
Consider a global-local shrinkage prior $\beta_j \given \gshrink, \lshrink_j \sim \normal(0, \gshrink^2 \lshrink_j^2)$, $\lshrink_j \sim \localPrior(\cdot)$ and $\gshrink \sim \globalPrior(\cdot)$.
Introducing the fictitious data $\z = \bm{0}$ as in \eqref{eq:gibbs_friendly_regularized_prior} is equivalent to using the regularized prior \eqref{eq:regularized_prior} on $(\beta_j, \lshrink_j)$, yielding
\begin{equation*}
\beta_j \given \gshrink, \lshrink_j, \zeta, z_j = 0
	\sim \normal \! \left(0,
		\left( \frac{1}{\zeta^2} + \frac{1}{\gshrink^2 \lshrink_j^2} \right)^{-1}
	\right).
\end{equation*}
Or, with $\lshrink_j$ marginalized out, we have
\begin{equation*}
\pi(\beta_j \given \gshrink, \zeta, z_j = 0)
	\propto \pi(\beta_j \given \gshrink) \exp \left(- \frac{\beta_j^2}{2 \zeta^2}\right).
\end{equation*}
When the likelihood depends only on $\bbeta$, the posterior full conditional of $\gshrink, \blshrink$ has density
\begin{equation}
\label{eq:scale_parameter_posterior}
\pi(\gshrink, \blshrink \given \bbeta)
	\propto \globalPrior(\gshrink) \prod_j \frac{1}{\gshrink \lshrink_j} \exp\!\left( - \frac{\beta_j^2}{2 \gshrink^2 \lshrink_j^2}\right) \localPrior(\lshrink_j).
\end{equation}
\end{proposition}

\begin{figure}
    \begin{subfigure}[b]{0.48\textwidth}
    	\centering
        \includegraphics[width=.75\textwidth]{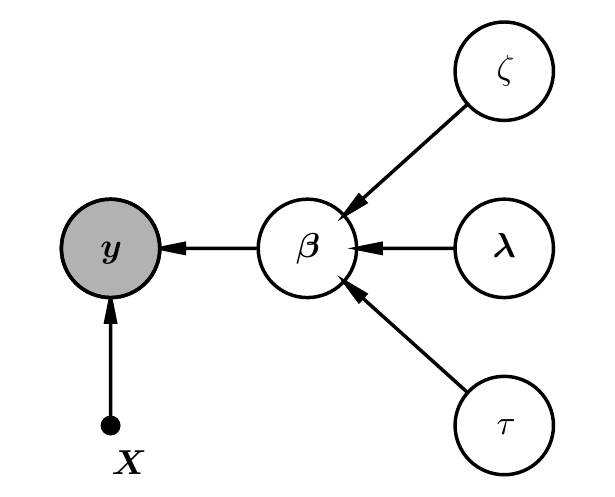}
        \caption{
        	Of the form \eqref{eq:regularized_global_local_mixture} as previously proposed.
        	The posterior conditional of $(\gshrink, \blshrink)$ is affected by their dependency on $\zeta$ through $\bbeta$.
        }
        \label{fig:former_regularization}
    \end{subfigure}
    ~~
    \begin{subfigure}[b]{0.48\textwidth}
    	\centering
        \includegraphics[width=.75\textwidth]{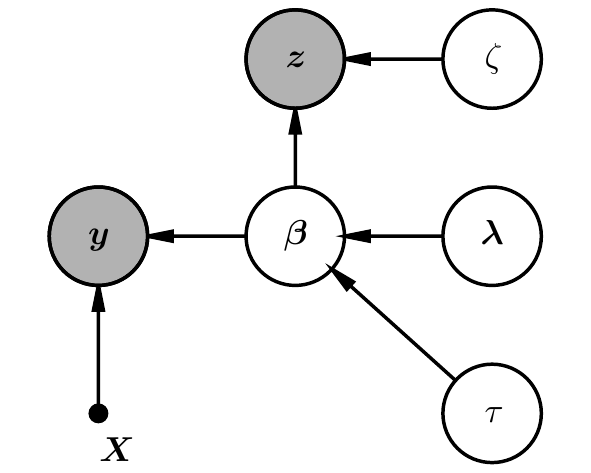}
        \caption{
        	Of the form \eqref{eq:gibbs_friendly_regularized_prior} as in Proposition~\ref{prop:gibbs_friendly_regularized_prior}.
        	Regularization does not affect the posterior conditional of $(\gshrink, \blshrink)$ as the parameters remains decoupled from $\zeta$.
        }
        \label{fig:proposed_regularization}
    \end{subfigure}
    \caption{Directed acyclic graphical model (a.k.a.\ Bayesian network) representation of regularized shrinkage priors under the two alternative formulations.}
    \label{fig:comparison_of_former_and_proposed_regularization}
\end{figure}

\section{Geometric and uniform ergodicity under regularized sparse logistic regression}
\label{sec:ergodicity_under_sparse_logit}
Shrinkage priors' popularity stems from, to a considerable extent, the ease of posterior computation via Gibbs sampling  \citep{bhadra2017lasso_horseshoe}.
As we have shown in Section~\ref{sec:gibbs_friendly_regularization}, shrinkage models can incorporate regularization without affecting its computational tractability.
We now investigate how fast such Gibbs samplers converge.

As a representative example where regularization is essential, we focus on Bayesian sparse logistic regression \citep{piironen2017regularized-horseshoe, nishimura2018cg-accelerated-gibbs}.
To be explicit, we consider the model
\begin{equation}
\label{eq:regularized_sparse_logistic_regression_model}
\begin{gathered}
y_i \given \x_i, \bbeta
	\sim \bernoulli\big(\logit^{-1}(\x_i^\transpose \bbeta) \big), \ \,
z_j = 0 \given \beta_j
	\sim \normal(0, \zeta^2), \\
\beta_j \given \gshrink, \lshrink_j
	\sim \normal(0, \gshrink^2 \lshrink_j^2), \
\gshrink \sim \globalPrior(\cdot), \
\lshrink_j \sim \localPrior(\cdot).
\end{gathered}
\end{equation}

The \polyagamma{} data-augmentation of \cite{polson2013polya_gamma} is a widely-used approach to carry out the posterior computation under the logistic model.
By introducing an auxiliary parameter $\bomega = (\omega_1, \ldots, \omega_n)$ having a \polyagamma{} distribution, the Gibbs sampler induces a transition kernel: $(\bomegaPrev, \bbetaPrev, \blshrinkPrev, \gshrinkPrev) \to (\bomega, \bbeta, \blshrink, \gshrink)$ through the following cycle of conditional updates:
\begin{enumerate}[topsep=.25\baselineskip, itemsep=.25\baselineskip]
\item Draw $\gshrink \given \bbetaPrev, \blshrinkPrev$ from the density proportional to \eqref{eq:scale_parameter_posterior}.
When using Bayesian bridge priors, draw from the collapsed distribution $\gshrink \given \bbetaPrev$ (Appendix~\ref{sec:bridge_prior_properties}).
\item Draw $\blshrink \given \bbetaPrev, \gshrink$ from the density proportional to \eqref{eq:scale_parameter_posterior}.
\item Draw
$\omega_i \given \bbetaPrev, \X
	\sim \textrm{PolyaGamma}(\textrm{shape}=1, \textrm{tilting}=\x_i^\transpose \bbetaPrev)$
for $i = 1, \ldots, n$.
\item Draw $\bbeta \given \bomega, \gshrink, \blshrink, \y, \X, \z = \bm{0}$ from the multivariate-Gaussian
\begin{equation}
\label{eq:regcoef_conditonal_for_sparse_logit}
\begin{aligned}
\bbeta \given \bomega, \gshrink, \blshrink, \y, \X, \z = \bm{0}
	&\sim \normal\!\left(
		\bPhi^{-1}\X^\transpose \left(\y - \textstyle\frac12 \right), \bPhi^{-1}
	\right) \\
	&\hspace*{3em} \text{ for } \
	\bPhi = \X^\transpose \bOmega \X + \zeta^{-2} \I + \gshrink^{-2} \bLshrink^{-2},
\end{aligned}
\end{equation}
where $\bOmega = \textrm{diag}(\bomega)$ and $\bLshrink = \textrm{diag}(\blshrink)$.
\end{enumerate}
Note that the transition kernel actually depends neither on $\bomegaPrev$ nor $\gshrinkPrev$ (nor $\blshrinkPrev$ in the Bayesian bridge case) because of conditional independence.
We refer readers to \cite{polson2013polya_gamma} for more details on this data augmentation scheme.
In our analysis, we do not use any specific properties of the \polyagamma{} distribution aside from a couple of results from \cite{choi2013ergodicity-bayes-logit} and \cite{wang2018ergodicity_polya_gamma_flat_prior}.

The P{\'o}lya-Gamma Gibbs sampler for the logistic model has previously been analyzed under a Gaussian or flat prior on $\bbeta$ \citep{choi2013ergodicity-bayes-logit, wang2018ergodicity_polya_gamma_flat_prior}, but not under shrinkage priors.
We establish geometric and uniform ergodicity --- critical properties for any practical Markov chain Monte Carlo algorithms \citep{jones2001honest_mcmc}.
These properties imply the Markov chain central limit theorem and enables consistent estimation of Monte Carlo errors, ensuring that the Gibbs sampler reliably estimates quantities of interest \citep{flegal2011mcmc_with_confidence}.
To avoid cluttering notations and obscuring the main ideas, our analysis below assumes the slab width $\zeta$ to be fixed; however, the same conclusions hold if we only assume a prior constraint of the form $\zeta \leq \zeta_{\max} < \infty$ (Remark~\ref{rmk:proof_sketch_for_variable_slab_size}).


We verify that the Gibbs sampler satisfies the \textit{minorization} and \textit{drift} condition upon on which geometric and uniform ergodicity are immediately implied by the well-known theory of Markov chains \citep{meyn2009stochastic_stability, roberts2004ergodic_theory}. 
In the statements to follow, we assume that a transition kernel $P(\btheta^*, \diff \btheta)$ has a corresponding density function which, with slight abuse of notation, we denote by $\kernel(\btheta \given \btheta^*)$; 
in other words, the two satisfy a relation $P(\btheta^*, A) = \int_A \kernel(\btheta \given \btheta^*) \, \diff \btheta$.
A chain on the space $\btheta \in \bTheta$ with transition kernel $\kernel(\btheta^*, \diff \btheta)$ is said to satisfy a minorization condition with a \textit{small set} $\smallSet$ if there are $\delta > 0$ and a probability density $\pi(\cdot)$ such that 
\begin{equation*} 
\kernel(\btheta \given \btheta^*) 
	\geq \delta \, \pi(\btheta) 
	\ \text{ for all } \btheta^* \in \smallSet. 
\end{equation*} 
The chain is uniformly ergodic when $\smallSet = \bTheta$. 
Otherwise, the chain is geometrically ergodic if it additionally satisfies a drift condition i.e.\ there is a \textit{Lyapunov} function $V(\btheta) \geq 0$ such that, for $\driftContraction < 1$ and $\driftConst < \infty$, 
\begin{equation*} 
P V(\btheta^*) 
	:= \textstyle \int V(\btheta) \kernel(\btheta \given \btheta^*) \, \diff \btheta 
	\leq \driftContraction V(\btheta^*) + \driftConst 
\end{equation*} 
and $\smallSet = \{ \btheta : V(\btheta) \leq d \}$ is a small set for some $d > 2 \driftConst / (1 - \driftContraction)$ \citep{rosenthal1995minorization_and_convergence}. 

For a two-block component-wise sampler on the space $(\btheta, \bphi)$, alternately sampling $\btheta \sim \kernel(\, \cdot \given \bphi)$ and $\bphi \sim \kernel(\, \cdot \given \btheta)$, the geometric and uniform ergodicity of the joint chain follows from that of the marginal chain with the transition kernel $\kernel(\btheta \given \btheta^*) = \int \kernel(\btheta \given \bphi) \kernel(\bphi \given \btheta^*) \, \diff \bphi$ \citep{roberts2001deinitialization}.
In establishing the uniform ergodicity under Bayesian bridge (Theorem~\ref{thm:uniform_ergodicity_of_bayesian_bridge}), we decompose the collapsed Gibbs sampler into components $\bbeta$ and $(\bomega, \gshrink, \blshrink)$ and study the marginal chain in $\bbeta$.
In the subsequent analysis establishing the geometric ergodicity under a more general class of regularized shrinkage priors (Theorem~\ref{thm:geom_ergodicity}), we decompose the Gibbs sampler into components $(\bbeta, \blshrink)$ and $(\bomega, \gshrink)$ and study the marginal chain in $(\bbeta, \blshrink)$.

Below are the main ergodicity results we will establish in this section, the uniform rate under Bayesian bridge and geometric rate under more general shrinkage priors:

\begin{theorem}[Uniform ergodicity in the Bayesian bridge case]
\label{thm:uniform_ergodicity_of_bayesian_bridge}
If the prior $\globalPrior(\cdot)$ is supported on $[\gshrink_{\min}, \infty)$ for $\gshrink_{\min} > 0$, then the \polyagamma{} Gibbs sampler for regularized Baysian bridge logistic regression is uniformly ergodic.
\end{theorem}

\begin{theorem}[Geometric ergodicity]
\label{thm:geom_ergodicity}
Suppose that the local scale prior satisfies \linebreak $\| \localPrior \|_\infty < \infty$ and that the global scale prior $\globalPrior(\cdot)$ is supported on $[\gshrink_{\min}, \gshrink_{\max}]$ for $0 < \gshrink_{\min} \leq \gshrink_{\max} < \infty$.
Then the \polyagamma{} Gibbs sampler for regularized sparse logistic regression is geometrically ergodic.
\end{theorem}

\begin{remark*}
Uniform / geometric ergodicity is an essential requirement for, yet not a guarantee of, practically efficient Markov chains \citep{roberts2004ergodic_theory}.
In fact, the simulation results of Section~\ref{sec:simulation} show that the benefit of regularization is greatest when $\zeta$ is chosen small enough to impose a reasonable prior constraint on the value of $\beta_j$'s.
\end{remark*}

\subsection{Behavior of shrinkage model Gibbs samplers near $\beta_j = 0$}
\label{sec:gibbs_sampler_behavior_near_spike}

In many models, establishing minorization and drift condition amounts to quantifying the chain's behavior in the tail of the target.
In studying convergence rates under shrinkage models, however, we are faced with an additional and distinctive challenge: the need to establish that the chain does not get ``stuck'' near the spike at $\beta_j = 0$ \citep{pal2014ergodicity_shrinkage_models, johndrow2018scalable_mcmc}.
Regularization effectively eliminates the possibility of the chain meandering to infinity, making it relatively routine to analyze its behavior as $\beta_j \to \infty$.
On the other hands, the existing results provide no general insights into the behavior near $\beta_j = 0$.
In fact, a careful examination of the proofs by \cite{pal2014ergodicity_shrinkage_models} and \cite{johndrow2018scalable_mcmc} reveals that the analyses under various shrinkage priors could have been unified if we had a more general characterization of shrinkage model Gibbs samplers' behavior near $\beta_j = 0$.

To fill in this theoretical gap, we start our analysis by abstracting key model-agnostic results from our proofs of minorization and drift condition for the sparse logistic regression Gibbs sampler. 
Our Proposition~\ref{prop:local_scale_posterior_in_the_limit} and \ref{prop:scale_negative_power_expectation_bound} below characterize properties of the distribution of $\lshrink_j \given \beta_j, \gshrink$ ---
this distribution, due to conditional independence, typically coincides with the full posterior conditional of $\lshrink_j$ and critically informs behavior of the subsequent update of $\beta_j$ in a shrinkage model Gibbs sampler.
Our proof techniques apply to a broad range of shrinkage priors, essentially requiring only that $\| \localPrior \|_\infty := \max_\lshrink \localPrior(\lshrink) < \infty$.\footnote{
	The results presented in this article, specifically those that depend on Proposition~\ref{prop:modifying_density_at_zero} and Lemma~\ref{lem:expectation_of_scale_negative_power}, implicitly assume that $\localPrior(\lshrink)$ is absolutely continuous at $\lshrink_{\min} = \inf\left\{ \lshrink: \localPrior(\lshrink) > 0 \right\}$.
	This is a purely technical assumption as any shrinkage prior in practice should satisfy $\localPrior(\lshrink) > 0$ for $\lshrink > 0$ and be a differentiable function of $\lshrink$.
}

Proposition~\ref{prop:local_scale_posterior_in_the_limit} below plays a critical role in our proof of minorization condition. 
The proposition tells us that a sample from $\lshrink_j \given \betaPrev_j, \gshrink$ has a uniformly lower-bounded probability of $\lshrink_j \geq a$ as long as $|\betaPrev_j / \gshrink|$ is bounded away from zero.
In turn, the subsequent update of $\beta_j$ conditional on $\lshrink_j$ should also have a guaranteed chance of landing away from zero.
Intuitively, we can thus interpret the proposition as suggesting that a shrinkage model Gibbs sampler should not get ``absorbed'' to the spike at $\beta_j = 0$.
The difference in the limiting behavior as $| \betaPrev_j / \gshrink | \to 0$, depending on whether $\int \lshrink^{-1} \localPrior(\lshrink) \, \diff \lshrink < \infty$, is also significant and leads to the difference between geometric and uniform convergence under the sparse logistic regression example through  Theorem~\ref{thm:minorization}.
\begin{proposition}
\label{prop:local_scale_posterior_in_the_limit}
For any $a > 0$, the tail probability $\prob(\lshrink_j \geq a \given \betaPrev_j, \gshrink)$ is a decreasing function of $|\betaPrev_j / \gshrink|$.
If $\int \lshrink^{-1} \localPrior(\lshrink) \, \diff \lshrink = \infty$, then as $| \betaPrev_j / \gshrink | \to 0$ the tail probability converges to $0$, i.e.\ the conditional $\lshrink_j \given \betaPrev_j, \gshrink$ converges in distribution to a delta measure at $0$.
If $\int \lshrink^{-1} \localPrior(\lshrink) \, \diff \lshrink < \infty$, then the conditional $\lshrink_j \given \betaPrev_j, \gshrink$ converges in distribution to $\pi(\lshrink_j) \propto \lshrink_j^{-1} \localPrior(\lshrink_j)$ as $|\betaPrev_j / \gshrink| \to 0$.
\end{proposition}


Another key property of $\lshrink_j \given \beta_j, \gshrink$, featured prominently in our proof of the drift condition (Theorem~\ref{thm:drift}), is provided by Proposition~\ref{prop:scale_negative_power_expectation_bound} below.
To briefly provide a context, a Lyapunov function of the form $V(\bbeta) = \sum_j | \beta_j |^{-\alpha}$ has proven effective in analyzing a shrinkage model Gibbs sampler (\citealt{pal2014ergodicity_shrinkage_models}, \citealt{johndrow2018scalable_mcmc}, Section~\ref{sec:drift}).
And bounding the conditional expectation of $\gshrink^{-\lyapunovExponent} \lshrink_j^{-\lyapunovExponent}$ as below often constitutes a critical step in establishing the drift condition.

\begin{proposition}
\label{prop:scale_negative_power_expectation_bound}
Let $R > 0$ and $\alpha \in [0, 1)$.
If $\| \localPrior \|_\infty < \infty$, then there is an increasing function $\gamma(r) > 0$ with $\lim_{r \to 0} \gamma(r) = 0$, for which the expectation with respect to $\lshrink_j \given \betaPrev_j, \gshrink$ satisfies
\begin{equation}
\label{eq:scale_negative_power_expectation_bound}
\expectation\!\left[
	\gshrink^{-\lyapunovExponent} \lshrink_j^{-\lyapunovExponent} \given \gshrink, \betaPrev_j
\right]
	\leq \gamma(R / \gshrink) \left( \, \left| \betaPrev_j \right|^{-\lyapunovExponent} +  \left| R \right|^{-\lyapunovExponent} \right).
\end{equation}
\end{proposition}
\noindent Proposition~\ref{prop:local_scale_posterior_in_the_limit} and \ref{prop:scale_negative_power_expectation_bound} are substantial theoretical contributions on their own, but we defer their proofs to Appendix~\ref{sec:proofs_of_gibbs_sampler_behavior_near_spike} so that we can without interruption proceed to establish ergodicity results in the regularized sparse logisitic regression case.

\begin{remark*}
The assumption $\| \localPrior \|_\infty < \infty$ is sufficient but not necessary one for the conclusion of Proposition~\ref{prop:scale_negative_power_expectation_bound} and later of Theorem~\ref{thm:drift}.
Following the analysis by \citet{pal2014ergodicity_shrinkage_models}, we can show that the conclusions also hold under normal-gamma priors with any shape parameter $\bridgeExponent > 0$.
These priors have the property $\localPrior(\lshrink) \sim O(\lshrink^{2\bridgeExponent - 1})$ as $\lshrink \to 0$ and hence $\lim_{\lshrink \to 0} \pi(\lshrink) = \infty$ for $\bridgeExponent < 1/2$.
We leave it as future work to characterize the behavior of general shrinkage priors with $\| \localPrior \|_\infty = \infty$.
\end{remark*}

\begin{remark*}
In Appendix~\ref{sec:further_results_on_general_shrinkage_model_Gibbs_behavior}, 
we show that Proposition~\ref{prop:local_scale_posterior_in_the_limit} and \ref{prop:scale_negative_power_expectation_bound} can also be applied to establish uniform/geometric ergodicity of a Gibbs sampler for Bayesian sparse probit regression, demonstrating their relevance beyond the sparse logistic regression example.
\end{remark*}

\subsection{Minorization --- with uniform ergodicity in special cases}
\label{sec:minorization_and_unif_ergodicity}

Having described the noteworthy model-agnostic results within our proofs, from now on we focus exclusively on the regularized sparse logistic regression case.
We first consider the Gibbs sampler with fixed $\gshrink$ in Lemma~\ref{lem:minorization} and Theorem~\ref{thm:minorization}.
While fixing the global scale parameter is a common assumption in the ergodicity proofs for shrinkage models \citep{pal2014ergodicity_shrinkage_models}, we subsequently show that this assumption can be replaced with much weaker ones;
we only require $\gshrink \sim \globalPrior(\cdot)$ to be supported away from $0$ in Theorem~\ref{thm:uniform_ergodicity_of_bayesian_bridge} and additionally away from $+\infty$ in Theorem~\ref{thm:non_uniform_minorization}.

Let $\kernel(\bbeta \given \bbetaPrev, \gshrink, \blshrink)$ denote the transition kernel corresponding to Step 3 and 4 of the Gibbs sampler as described in Page~\pageref{eq:regcoef_conditonal_for_sparse_logit} and $\kernel(\bbeta \given \bbetaPrev, \gshrink)$ corresponding to Step 2 -- 4.
In other words, we define
\begin{align*}
\kernel(\bbeta \given \bbetaPrev, \gshrink, \blshrink)
	&= \int \pi(\bbeta \given \bomega, \gshrink, \blshrink, \y, \X, \z = \bm{0}) \, \pi(\bomega \given \bbetaPrev, \X) \, \diff \bomega, \\
\kernel(\bbeta \given \bbetaPrev, \gshrink)
	&= \int \kernel(\bbeta \given \bbetaPrev, \gshrink, \blshrink) \, \pi(\blshrink \given \bbetaPrev) \, \diff \blshrink.
\end{align*}
The following lemma builds on a result of \cite{choi2013ergodicity-bayes-logit} and plays a prominent role, along with Proposition~\ref{prop:local_scale_posterior_in_the_limit}, in our proofs of minorization conditions.

\begin{lemma}
\label{lem:minorization}
Whenever $\min_j \gshrink \lshrink_j \geq R > 0$, there is $\delta' > 0$ --- independent of $\gshrink$ and $\blshrink$ except through $R$ --- 
such that the following minorization condition holds:
\begin{equation*}
\kernel(\bbeta \given \bbetaPrev, \gshrink, \blshrink)
	\geq \delta' \, \normal(\bbeta; \bmu_R, \bPhi_R^{-1}),
\end{equation*}
where $\bPhi_{R} = \frac12 \X^\transpose \X + \zeta^{-2} \I + R^{-2} \I$ and $\bmu_{R} = \bPhi_{R}^{-1} \X^\transpose (\y - \bm{1} / 2)$.
\end{lemma}
\noindent We defer the proof to Appendix~\ref{sec:proof_of_minorization_lemma}.

We now establish a minorization condition for the Gibbs sampler with fixed $\gshrink$.

\begin{theorem}[Minorization]
\label{thm:minorization}
Let $\epsilon, R > 0$.
On a small set $\{\bbetaPrev : \min_j | \betaPrev_j / \gshrink | \geq \epsilon \}$, the marginal transition kernel satisfies a minorization condition
\begin{equation*}
\kernel(\bbeta \given \bbetaPrev, \gshrink)
	\geq \delta(\gshrink) \, \normal(\bbeta; \bmu_{R}, \bPhi_{R}^{-1}),
\end{equation*}
where $\bmu_{R}$ and $\bPhi_{R}$ are defined as in Lemma~\ref{lem:minorization}, and $\delta(\gshrink) > 0$ is increasing in $\gshrink$ and otherwise depends only on $\epsilon$, $R$, and $\localPrior$.
Moreover, the minorization holds uniformly on $\bbetaPrev \in \mathbb{R}^p$ in case the prior satisfies
$ \int_0^\infty \lshrink^{-1} \localPrior(\lshrink)  \, \diff \lshrink < \infty $.
\end{theorem}

\begin{proof}
Using Lemma~\ref{lem:minorization}, we have
\begin{align*}
\kernel(\bbeta \given \bbetaPrev, \gshrink)
	&= \int \kernel(\bbeta \given \bbetaPrev, \gshrink, \blshrink) \pi(\blshrink \given \bbetaPrev, \gshrink) \, \diff \blshrink \\
	&\geq \int_{ \left\{ \min_j \gshrink \lshrink_j \geq R \right\} } \kernel(\bbeta \given \bbetaPrev, \gshrink, \blshrink) \pi(\blshrink \given \bbetaPrev, \gshrink) \, \diff \blshrink \\
	&\geq \delta' \, \normal(\bbeta; \bmu_{R}, \bPhi_{R}^{-1})
		\prod_j \int_{R / \gshrink}^\infty \pi(\lshrink_j \given \betaPrev_j, \gshrink) \ \diff \lshrink_j,
\end{align*}
for $\delta' > 0$ depending only on $R$.
Also, Proposition~\ref{prop:local_scale_posterior_in_the_limit} implies that whenever $|\betaPrev_j / \gshrink| \geq \epsilon$
\begin{equation*}
\int_{R / \gshrink}^\infty \pi(\lshrink_j \given \betaPrev_j, \gshrink) \ \diff \lshrink_j
	\geq
	\int_{R / \gshrink}^\infty \pi\!\left(
		\lshrink \given[\big] |\betaPrev / \gshrink| = \epsilon
	\right) \diff \lshrink
	> 0.
\end{equation*}
Hence, $ \prod_j \int_{R / \gshrink}^\infty \pi(\lshrink_j \given \betaPrev_j, \gshrink) \ \diff \lshrink_j $ is lower bounded by a positive constant depending only on $\epsilon$ and $R / \gshrink$.
In case $C = \int_0^\infty \lshrink^{-1} \localPrior(\lshrink)  \, \diff \lshrink < \infty$, we can forgo the assumption $|\betaPrev_j / \gshrink| \geq \epsilon$ and obtain a uniform lower bound since
\begin{equation*}
\int_R^\infty \pi(\lshrink_j \given \betaPrev_j, \gshrink) \ \diff \lshrink_j
	\geq \frac{1}{C} \int_R^\infty \lshrink^{-1} \localPrior(\lshrink)  \, \diff \lshrink
	> 0. \qedhere
\end{equation*}
\end{proof}

We now relax the assumption of fixed $\gshrink$.
The results of \cite{vanDerPas2017horseshoe_contraction} suggest that a constraint of the form $0 < \gshrink_{\min} \leq \gshrink \leq \gshrink_{\max} < \infty$ can improve the statistical property of shrinkage priors.
As it turns out, such a constraint also enables us to establish minorization conditions for the full Gibbs sampler under sparse logistic regression with $\gshrink$ update incorporated.
We can in fact take $\gshrink_{\max} = \infty$ in case of the Bayesian bridge prior, whose unique structure allows us to marginalize out $\lshrink_j$'s when updating $\gshrink$ (\citealt{polson2014bayes_bridge}; Appendix~\ref{sec:bridge_prior_properties}).
This collapsed update of $\gshrink$ from $\gshrink \given \bbeta$ makes it possible to deduce the uniform ergodicity result of Theorem~\ref{thm:uniform_ergodicity_of_bayesian_bridge} as an immediate consequence of Theorem~\ref{thm:minorization} by studying the marginal transition $\bbetaPrev \to \bbetaNext$ with kernel
\begin{equation}
\label{eq:bridge_marginal_transition_kernel}
\kernel(\bbetaNext \given \bbetaPrev)
	= \int_{\gshrink_{\min}}^\infty \kernel(\bbetaNext \given \bbetaPrev, \gshrink) \, \pi(\gshrink \given \bbetaPrev) \, \diff \gshrink.
\end{equation}

\begin{proof}[Proof of Theorem~\ref{thm:uniform_ergodicity_of_bayesian_bridge}]
It suffices to establish uniform minorization for the marginal transition kernel \eqref{eq:bridge_marginal_transition_kernel}.
Under the Bayesian bridge prior, we have $\localPrior(\lshrink) \propto O(\lshrink^{2\bridgeExponent})$ as $\lshrink \to 0$ (Appendix~\ref{sec:bridge_prior_properties}) and hence $\int \lshrink^{-1} \localPrior(\lshrink) < \infty$.
The minorization condition of Theorem~\ref{thm:minorization} thus holds uniformly in $\bbetaPrev$, yielding
\begin{equation}
\label{eq:bridge_kernel_lower_bd}
\int_{\gshrink_{\min}}^\infty
	\kernel(\bbeta \given \bbetaPrev, \gshrink) \,
	\pi(\gshrink \given \bbetaPrev) \,
\diff \gshrink
	\geq
	\normal(\bbeta; \bmu_{R}, \bPhi_{R}^{-1})
	\int_{\gshrink_{\min}}^\infty
		\delta(\gshrink) \,
		\pi(\gshrink \given \bbetaPrev) \,
	\diff \gshrink,
\end{equation}
for $R > 0$.
Theorem~\ref{thm:minorization} further tells us that $\delta(\gshrink) > 0$ is increasing in $\gshrink$, so we have
\begin{equation}
\label{eq:bridge_minorization_const_lower_bd}
\int_{\gshrink_{\min}}^\infty
		\delta(\gshrink) \,
		\pi(\gshrink \given \bbetaPrev) \,
	\diff \gshrink
	\geq \delta(\gshrink_{\min})
	> 0.
\end{equation}
The inequalities \eqref{eq:bridge_kernel_lower_bd} and \eqref{eq:bridge_minorization_const_lower_bd} together establish uniform minorization.
\end{proof}

For more general shrinkage priors, the global scale $\gshrink$ must be updated from the full conditional $\gshrink \given \bbeta, \blshrink$.
This makes it necessary to study the marginal transition $(\bbeta^*, \blshrink^*) \to (\bbeta, \blshrink)$, jointly in regression coefficients and local scales, with kernel
\begin{equation}
\label{eq:coef_lscale_marginal_transition_kernel}
\kernel(\bbeta, \blshrink \given \bbetaPrev, \blshrinkPrev)
	= \int_{\gshrink_{\min}}^{\gshrink_{\max}}
		\kernel(\bbeta \given \bbetaPrev, \gshrink, \blshrink)
		\textstyle \prod_j \pi(\lshrink_j \given \betaPrev_j, \gshrink) \,
		\pi(\gshrink \given \bbetaPrev, \blshrinkPrev) \,
	\diff \gshrink.
\end{equation}
We establish a minorization condition for this general case in Theorem~\ref{thm:non_uniform_minorization}.

\begin{theorem}
\label{thm:non_uniform_minorization}
If the prior $\globalPrior(\cdot)$ is supported on $[\gshrink_{\min}, \gshrink_{\max}]$ for $0 < \gshrink_{\min} \leq \gshrink_{\max} < \infty$,
then the marginal transition kernel $\kernel(\bbeta, \blshrink \given \bbetaPrev, \blshrinkPrev)$ of the \polyagamma{} Gibbs sampler for regularized sparse logistic regression satisfies a minorization condition on a small set $\left\{(\bbetaPrev, \blshrinkPrev) : 0 < \epsilon \leq | \betaPrev_j | \leq E < \infty \text{ for all } j \right\}$.
\end{theorem}

\begin{proof}
By Lemma~\ref{lem:minorization} and the fact $\gshrink \lshrink_j \geq \gshrink_{\min} \lshrink_j$, we know that for $R > 0$
\begin{equation}
\label{eq:minorization_regcoef_update}
\kernel(\bbeta \given \bbetaPrev, \gshrink, \blshrink)
	\geq
		\ind \!\left\{ \textstyle \min_j \gshrink_{\min} \lshrink_j \geq R \right\} \,
		\delta' \, \normal(\bbeta; \bmu_{R}, \bPhi_{R}^{-1}).
\end{equation}
To lower bound the term $\prod_j \pi(\lshrink_j \given \betaPrev_j, \gshrink)$ in  \eqref{eq:coef_lscale_marginal_transition_kernel}, we first recall that
\begin{equation*}
\pi(\lshrink_j \given \betaPrev_j, \gshrink)
	= \frac{
		\lshrink_j^{-1} \exp\!\left( - \betaPrev[2]_j / 2 \gshrink^2 \lshrink_j^2 \right) \localPrior(\lshrink_j)
	}{
		\int_0^\infty \lshrink^{-1}  \exp\!\left( - \betaPrev[2]_j / 2 \gshrink^2  \lshrink^2 \right) \localPrior(\lshrink) \, \diff \lshrink
	}.
\end{equation*}
When $\gshrink_{\min} \leq \gshrink \leq \gshrink_{\max}$ and $\epsilon \leq | \betaPrev_j | \leq E$, we have
\begin{equation*}
\exp\!\left( - E^2 / 2 \gshrink_{\min}^2 \lshrink^2 \right)
	\leq \exp\!\left( - \beta_j^2 / 2 \gshrink^2 \lshrink^2 \right)
	\leq \exp\!\left( - \epsilon^2 / 2 \gshrink_{\max}^2 \lshrink^2 \right).
\end{equation*}
It follows from the above inequalities that
\begin{equation}
\label{eq:minorization_local_scale_update}
\pi(\lshrink_j \given \betaPrev_j, \gshrink)
	\geq \frac{
		\lshrink_j^{-1} \exp\!\left( - E^2 / 2 \gshrink_{\min}^2 \lshrink_j^2 \right) \localPrior(\lshrink_j)
	}{
		\int_0^\infty \lshrink^{-1}  \exp\!\left( - \epsilon^2 / 2 \gshrink_{\max}^2  \lshrink^2 \right) \localPrior(\lshrink) \, \diff \lshrink
	}
	:= \eta \, \pi_{\rm lower}(\lshrink_j)
\end{equation}
for $\eta > 0$ and density $\pi_{\rm lower}(\cdot)$ independent of $\betaPrev_j$ and $\gshrink$.
Combining \eqref{eq:minorization_regcoef_update} and \eqref{eq:minorization_local_scale_update}, we can lower bound the transition kernel \eqref{eq:coef_lscale_marginal_transition_kernel} as
\begin{align*}
& \kernel(\bbeta, \blshrink \given \bbetaPrev, \blshrinkPrev) \\
& \quad \geq \delta' \eta \,
		\ind \!\left\{ \min_j \lshrink_j \geq \frac{R}{\gshrink_{\min}} \right\} \,
		\normal(\bbeta; \bmu_{R}, \bPhi_{R}^{-1})
		\prod_j \pi_{\rm lower}(\lshrink_j)
		\int_{\gshrink_{\min}}^{\gshrink_{\max}} \! \pi(\gshrink \given \bbetaPrev, \blshrink) \, \diff \gshrink \\
& \quad = \delta' \eta \
		\normal(\bbeta; \bmu_{R}, \bPhi_{R}^{-1})
		\prod_j \ind \!\left\{ \lshrink_j \geq \frac{R}{\gshrink_{\min}} \right\} \pi_{\rm lower}(\lshrink_j). \qedhere
\end{align*}
\end{proof}

\subsection{Drift condition and geometric ergodicity}
\label{sec:drift}

Here we establish a drift condition for geometric ergodicity under sparse logistic regression.
As discussed in Section~\ref{sec:gibbs_sampler_behavior_near_spike}, the regularization prevents the Markov chain from meandering to infinity, so the main question is whether the chain can get ``stuck'' for a long time near $\betaPrev_j = 0$.
The following result shows that this does not happen as long as the global scale $\gshrink$ is bounded away from zero.
\begin{theorem}
\label{thm:drift}
Suppose that the local scale prior satisfies $\| \localPrior \|_\infty < \infty$ and that the global scale prior $\globalPrior(\cdot)$ is supported on $[\gshrink_{\min}, \infty)$ for $\gshrink_{\min} > 0$.
Then the marginal transition kernel $\kernel(\bbeta, \blshrink \given \bbetaPrev, \blshrinkPrev)$ satisfies a drift condition with a Lyapunov function
$V(\bbeta) = \sum_j | \beta_j |^{-\lyapunovExponent}$ for any $0 \leq \lyapunovExponent < 1$.
\end{theorem}

\begin{proof}
Note that $PV(\bbetaNext^*)$ can be expressed as a series of iterated expectations with respect to (1)  $\bbeta \given \bomega, \gshrink, \blshrink, \y, \X, \z = \bm{0}$,  (2) $\bomega \given \bbetaPrev$, (3) $\blshrink \given \bbetaPrev, \gshrink$, and (4) $\gshrink \given \bbetaPrev, \blshrinkPrev$.
We will bound the iterated expectations of $| \beta_j |^{-\lyapunovExponent}$ one by one.

Since $\bbeta \given \bomega, \gshrink, \blshrink, \y, \X, \z = \bm{0}$ is distributed as Gaussian, denoting by $\mu_j$ and $\sigma_j^2$ the conditional mean and variance of $\beta_j$, Proposition~\ref{prop:negative_moment_bound} below tells us that
\begin{equation*}
\begin{aligned}
&\expectation\left[
	|\beta_j|^{-\lyapunovExponent} \given \bomega, \gshrink, \blshrink, \y, \X, \z = \bm{0}
\right]
	\leq
		\const_\lyapunovExponent(\mu_j / \sigma_j) \,
		\sigma_j^{- \lyapunovExponent} \\
&\hspace{2em} \text{ where } \,
	\sup_{t} \, \const_\lyapunovExponent(t)
		\leq \frac{
			\Gamma \left(\frac{1 - \lyapunovExponent}{2}\right)
		}{
			2^{\lyapunovExponent / 2} \sqrt{\pi}
		}
		\ \text{ and } \
	\const_\lyapunovExponent(t) = O(|t|^{-\lyapunovExponent})
	\text{ as } |t| \to \infty.
\end{aligned}
\end{equation*}
For the purpose of this proof, we can simply set $\const_\lyapunovExponent$ to be its global upper bound; however, a tighter bound may be obtained when the posterior concentrates away from zero and thereby resulting in $|\mu_j / \sigma_j| \to \infty$ and $\const_\lyapunovExponent(\mu_j / \sigma_j) \to 0$ as the sample size increases.
Combined with Proposition~\ref{prop:marginal_variance_bound} below, the above inequality implies
\begin{equation}
\label{eq:first_expectation_simplified}
\frac{1}{\const_\lyapunovExponent}
\expectation\left[
	|\beta_j|^{-\lyapunovExponent} \given \bomega, \gshrink, \blshrink, \y, \X, \z = \bm{0}
\right]
	\leq \gshrink^{-\lyapunovExponent} \lshrink_j^{-\lyapunovExponent}
				+ \zeta^{-\lyapunovExponent}
				+ 1 - \frac{\lyapunovExponent}{2} + \frac{\lyapunovExponent}{2} \sum_{i = 1}^n \omega_i x_{ij}^2.
\end{equation}
In taking the expectation of \eqref{eq:first_expectation_simplified} with respect to $\bomega \given \bbetaPrev$, we use the result $\expectation[\, \omega_j \given \bbetaPrev]  \leq 1 / 4$ of \citet{wang2018ergodicity_polya_gamma_flat_prior} to obtain
\begin{equation}
\label{eq:expectation_wrt_polya_gamma}
\frac{1}{\const_\lyapunovExponent}
\expectation\left[
	|\beta_j|^{-\lyapunovExponent} \given \gshrink, \blshrink
\right]
	\leq \gshrink^{-\lyapunovExponent} \lshrink_j^{-\lyapunovExponent}
				+ \zeta^{-\lyapunovExponent}
				+ 1 - \frac{\lyapunovExponent}{2} + \frac{\lyapunovExponent}{8} \sum_{i = 1}^n x_{ij}^2.
\end{equation}
Taking the expectation of \eqref{eq:expectation_wrt_polya_gamma} with respect to $\blshrink \given \gshrink, \bbetaPrev$, we have
\begin{equation}
\label{eq:expectation_wrt_local_scale}
\begin{aligned}
\frac{1}{\const_\lyapunovExponent}
\expectation\left[
	|\beta_j|^{-\lyapunovExponent} \given \gshrink, \bbetaPrev
\right]
	&\leq \expectation\!\left[
			\gshrink^{-\lyapunovExponent} \lshrink_j^{-\lyapunovExponent} \given \gshrink, \betaPrev_j
		\right]
		+ \const'(\lyapunovExponent, \X) \\
	\text{ where } \
	&
	\const'(\lyapunovExponent, \X)
		= \zeta^{-\lyapunovExponent} + 1 - \frac{\lyapunovExponent}{2} + \frac{\lyapunovExponent}{8} \sum_{i = 1}^n x_{ij}^2.
\end{aligned}
\end{equation}
Now choose $R > 0$ small enough that $\driftContraction(R / \gshrink) \leq \driftContraction(R / \gshrink_{\min}) < \const_\lyapunovExponent^{-1}$ in Proposition~\ref{prop:scale_negative_power_expectation_bound}.
Then we have the following inequality for $\driftContraction' := \const_\lyapunovExponent \driftContraction(R / \gshrink_{\min}) < 1$:
\begin{equation*}
\const_\lyapunovExponent \, \expectation\!\left[
	\gshrink^{-\lyapunovExponent} \lshrink_j^{-\lyapunovExponent} \given \gshrink, \betaPrev_j
\right]
	\leq \driftContraction'
		\left( | \betaPrev_j |^{-\lyapunovExponent} + | R |^{-\lyapunovExponent} \right) \!
\end{equation*}
for all $\gshrink \geq \gshrink_{\min}$.
Incorporating the above inequality into \eqref{eq:expectation_wrt_local_scale}, we obtain
\begin{equation*}
\expectation\left[
	|\beta_j|^{-\lyapunovExponent} \given \gshrink, \bbetaPrev
\right]
	\leq \driftContraction' \, | \betaPrev_j |^{-\lyapunovExponent}
		+ \driftContraction' \, | R |^{-\lyapunovExponent}
		+ \const_\lyapunovExponent \const'(\lyapunovExponent, \X).
\end{equation*}
Since $\pi(\gshrink \given \bbetaPrev, \blshrinkPrev)$ is supported on $\gshrink \geq \gshrink_{\min}$ by our assumption, taking the expectation with respect to $\gshrink \given \bbetaPrev, \blshrinkPrev$ yield
\begin{equation*}
\expectation\left[
	|\beta_j|^{-\lyapunovExponent} \given \bbetaPrev, \blshrinkPrev
\right]
		\leq \driftContraction' \, | \betaPrev_j |^{-\lyapunovExponent}
			+ \driftContraction' \, | R |^{-\lyapunovExponent}
			+ \const_\lyapunovExponent \const'(\lyapunovExponent, \X). \qedhere
\end{equation*}
\end{proof}

Theorem~\ref{thm:non_uniform_minorization} and \ref{thm:drift} together imply the geometric ergodicity result of Theorem~\ref{thm:geom_ergodicity}:

\begin{proof}[Proof of Theorem~\ref{thm:geom_ergodicity}]
We show that
$V(\bbeta) = \sum_j | \beta_j |^{-\lyapunovExponent} + \| \bbeta \|^2$
is a Lyapunov function for the marginal transition kernel $\kernel(\bbeta, \blshrink \given \bbetaPrev, \blshrinkPrev)$.
Note that
\begin{equation*}
\begin{aligned}
&\expectation \! \left[
	\| \bbeta \|^2 \given \bomega, \gshrink, \blshrink, \y, \X, \z = \bm{0}
\right] \\
	&\hspace*{2.5em}= \big\| \expectation \! \left[
		\bbeta \given \bomega, \gshrink, \blshrink, \y, \X, \z = \bm{0}
	\right] \big\|^2
	+ \textstyle \sum_j \textrm{var}\!\left(
		\beta_j^2 \given \bomega, \gshrink, \blshrink, \y, \X, \z = \bm{0}
	\right) \\
	&\hspace*{2.5em}= \big\| \bSigma \X^\transpose \! \left( \y - \textstyle \frac{1}{2} \right)  \big\|^2
	+ \textstyle \sum_j \e_j^\transpose \bSigma \e_j
\end{aligned}
\end{equation*}
for $\bSigma = \left(
		\X^\transpose \bOmega \X + \zeta^{-2} \I + \gshrink^{-2} \bLshrink^{-2}
	\right)^{-1}$.
Since $\bSigma \prec  \zeta^2 \I$, we have
$ \e_j^\transpose \bSigma \e_j \leq \zeta^2 $
and
$\| \bSigma \X^\transpose ( \y - \textstyle \frac{1}{2}) \|^2 \leq \zeta^2  \| \X^\transpose ( \y - \textstyle \frac{1}{2}) \|^2$.
Thus we have
\begin{equation}
\label{eq:bound_on_expectation_of_sqnorm}
\expectation \! \left[
	\| \bbeta \|^2 \given \bomega, \gshrink, \blshrink, \y, \X, \z = \bm{0}
\right]
	\leq \zeta^2 \big\| \bSigma \X^\transpose \! \left( \y - \textstyle \frac{1}{2} \right)  \big\|^2 +  n \zeta^2.
\end{equation}
Since the right-hand side does not depend on $\bomega, \gshrink, \blshrink$, the expectation with respect to $\kernel(\bbeta, \blshrink \given \bbetaPrev, \blshrinkPrev)$ satisfies the same bound:%
\begin{equation*}
\expectation \! \left[
	\| \bbeta \|^2 \given \bbetaPrev, \blshrinkPrev
\right]
	\leq
	\zeta^2 \big\| \bSigma \X^\transpose \! \left( \y - \textstyle \frac{1}{2} \right)  \big\|^2 +  n \zeta^2.
\end{equation*}
In addition to the above bound, we know that $\sum_j | \beta_j |^{-\lyapunovExponent}$ is a Lypunov function by Theorem~\ref{thm:drift}.
Hence, $V(\bbeta) = \sum_j | \beta_j |^{-\lyapunovExponent} + \| \bbeta \|^2$ is again a Lyapunov function.
Moreover, by Theorem~\ref{thm:non_uniform_minorization}, we know that the Gibbs sampler satisfies a minorization condition on the set $\left\{\bbetaPrev : 0 < \epsilon \leq | \betaPrev_j | \leq E < \infty \text{ for all } j \right\}$ for $\epsilon > 0$ and $E < \infty$.
Thus the sampler is geometrically ergodic.
\end{proof}

\begin{remark}
\label{rmk:proof_sketch_for_variable_slab_size}
As mentioned earlier, the geometric and uniform ergodicity as well as analogues of the intermediate results continue to hold when we relax the assumption of fixed $\zeta$ to a prior constraint of the form $\zeta \leq \zeta_{\max} < \infty$.
The proof goes as follows.
Due to the conditional independence, the Gibbs sampler on the joint space draws alternately from $\zeta \given \bbeta, \z = \bm{0}$ and $\bbeta, \bomega, \gshrink, \blshrink \given \y, \X, \z = \bm{0}, \zeta$.
By repeating all the previous arguments with $\zeta_{\max}$ in place of $\zeta$, we obtain essentially the identical minorization and drift bounds that hold for all $\zeta \leq \zeta_{\max}$.
Since the bounds hold uniformly on the support $\zeta \leq \zeta_{\max}$, the identical bounds again hold when taking the expectation over $\zeta \given \bbeta, \z = \bm{0}$.
\end{remark}

\subsubsection{Auxiliary results for proof of geometric ergodicity}
Proposition~\ref{prop:negative_moment_bound} and \ref{prop:marginal_variance_bound} below are used in the proof of Theorem~\ref{thm:drift} and are proved in Appendix~\ref{sec:proof_of_neg_moment_and_marginal_var_bound}.
Proposition~\ref{prop:negative_moment_bound} is a refinement of Proposition A1 in \citet{pal2014ergodicity_shrinkage_models} and of Equation~(41) in \citet{johndrow2018scalable_mcmc}, neither of which have the $D(\mu / \sigma)$ term.
\begin{proposition}
\label{prop:negative_moment_bound}
For $\lyapunovExponent \in (0, 1)$ and $\beta \sim \normal(\mu, \sigma^2)$, we have
\begin{equation*}
\expectation | \beta |^{-\lyapunovExponent}
	\leq \frac{
				\Gamma \left(\frac{1 - \lyapunovExponent}{2}\right)
			}{
				2^{\lyapunovExponent / 2} \sqrt{\pi}
			}
		\, \sigma^{- \lyapunovExponent}
		\min \!\left\{ 1, D(\mu / \sigma) \right\},
\end{equation*}
where $D(t) = O(|t|^{-\lyapunovExponent}) \to 0$ as $|t| \to \infty$ and can be chosen as
\begin{equation}
\label{eq:bound_on_kummers_confluent_hypergeom}
D(t) =
	\frac{1}{\betaFun\!\left(\frac{\lyapunovExponent}{2}, \frac{1 - \lyapunovExponent}{2} \right)}
	\left[
		\frac{ 2^{\frac{5}{2} - \lyapunovExponent} }{ 1 - \lyapunovExponent }
			\exp\!\left( - \frac{t^2}{4} \right)
		+ 2^{\frac{1}{2} + \lyapunovExponent}
			\Gamma\!\left( \frac{\lyapunovExponent}{2} \right)
			\left| t\right|^{- \lyapunovExponent}
	\right].
\end{equation}
\end{proposition}

\begin{proposition}
\label{prop:marginal_variance_bound}
The diagonals $\sigma_j$ of $\bSigma = \left( \X^\transpose \bOmega \X + \zeta^{-2} \I + \gshrink^{-2} \bLshrink^{-2} \right)^{-1}$ satisfy the following inequality for $0 \leq \lyapunovExponent < 1$:
\begin{equation*}
\sigma_j^{-\lyapunovExponent}
	\leq \gshrink^{-\lyapunovExponent} \lshrink_j^{-\lyapunovExponent}
			+ \zeta^{-\lyapunovExponent}
			+ 1 - \frac{\lyapunovExponent}{2} + \frac{\lyapunovExponent}{2} \sum_{i = 1}^n \omega_i x_{ij}^2.
\end{equation*}
\end{proposition}

\newcommand{\predFreq}{w}
\section{Simulation}
\label{sec:simulation}

We run a simulation study to assess the computational and statistical properties of the regularized sparse logistic regression model.
We use the Bayesian bridge prior $\pi(\beta_j \given \gshrink) \propto \gshrink^{-1} \exp( - | \beta_j / \gshrink |^{\bridgeExponent})$ to take advantage of the efficient global scale parameter update scheme.
This prior also allows us to experiment with a range of spike and tail behavior by varying the exponent $\bridgeExponent$, inducing larger spikes and heavier tails as $\bridgeExponent \to 0$.
For the global scale parameter, we chose the objective prior $\globalPrior(\gshrink) \propto \gshrink^{-1}$ (\citealp{berger2015overall_obayes}, Appendix~\ref{sec:bridge_prior_properties}) with the range restriction $10^{-6} \leq \expectation[ \, |\beta_j| \given \gshrink] \leq 1$
to ensure posterior propriety, though in practice we never observe a posterior draw of $\gshrink$ outside this range.
For the posterior computations, we use the \polyagamma{} Gibbs sampler provided by the \textit{bayesbridge} package available from Python Package Index (\url{pypi.org});
the source code is available at the GitHub
repository \url{https://github.com/aki-nishimura/bayes-bridge}.

\subsection{Data generating process: ``large $n$, but weak signal''  problems}
\label{sec:data_generating_process}
\cite{piironen2017regularized-horseshoe} demonstrate the benefits of regularizing shrinkage priors in the ``$p > n$'' case, when the number of predictors $p$ exceeds the sample size $n$.
To complement their study, we consider the case of rare outcomes and infrequently observed features, another common situation in which regularizing shrinkage priors becomes essential.
For example in healthcare data, many outcomes of interests have low incidence rates and many treatments are prescribed to only a small fraction of patients \citep{tian2018large_scale_pscore}.
This results in binary outcomes $\y$ and features $\x_j$ filled mostly with $0$'s, making the amount of information much less than otherwise expected \citep{greenland2016sparse_data_bias}.

To simulate under these ``large $n$, but weak signal'' settings, we generate synthetic data with $n = 2{,}500$ and $p = 500$ as follows.
We construct binary features with a range of observed frequencies by first drawing $2 \predFreq_j \sim \textrm{Beta}(1 / 2, 2)$ for $j = 1, \ldots, 500$; this in particular means $0 \leq \predFreq_j \leq 0.5$ and $\expectation[\predFreq_j] = 0.1$.
For each $j$, we then generate $x_{ij} \sim \bernoulli(\predFreq_j)$ for $i = 1, \ldots, n$.
We choose the true signal to be $\beta_j = 1$ for $j = 1, \ldots, 10$ and $\beta_j = 0$ for $j = 11, \ldots, 500$.
To simulate an outcome with low incidence rate, we choose the intercept to be $\beta_0 = 1.5$ and draw $y_i \sim \bernoulli(\logit(-\x_i^\transpose \bbeta))$, resulting in $y_i = 1$ for approximately 5\% of its entries.

\subsection{Convergence and mixing: with and without regularization}
\label{sec:convergence_and_mixing_with_and_without_regularization}
With the above data generating process, outcome $\y$ and design matrix $\X$ barely contain enough information to estimate all the coefficients $\beta_j$'s.
In particular, sparse logistic model without regularization can lead to a heavy-tailed posterior, for which uniform and geometric ergodicity of the \polyagamma{} Gibbs sampler becomes questionable.

These potential convergence and mixing issues are evidenced by the traceplot (Figure~\ref{fig:unregularized_bridge_traceplot}) of the posterior samples based on bridge exponent $\bridgeExponent = 1 / 16$.
As we are particularly concerned with the Markov chain wandering off to the tail of the target, we examine the estimated credible intervals to identify the coefficients with potential convergence and mixing issues.
Plotted in Figure~\ref{fig:bridge_traceplot} are the coefficients with the widest 95\% credible intervals; these coefficients also have some of the smallest estimated effective sample sizes, though the accuracy of such estimates is not guaranteed without geometric ergodicity.
When regularizing the shrinkage prior with a slab width $\zeta = 1$, the posterior samples indicate no such convergence or mixing issues (Figure~\ref{fig:regularized_bridge_traceplot}) and yield more sensible posterior credible intervals (Figure~\ref{fig:credible_interval_with_and_wo_regularization_comparison}).

We emphasize that there is no fundamental change in the Gibbs sampler itself when incorporating the regularization, the only change being the addition of the $\zeta^{-2} \I$ term in the conditional precision matrix \eqref{eq:regcoef_conditonal_for_sparse_logit} when updating $\bbeta$.
It is the change in the posterior --- more specifically the guaranteed light tails of the $\bbeta$ marginal --- that induces faster convergence and mixing.

We also assess sensitivity of convergence and mixing rates on the slab width $\zeta$.
The regularized prior recovers the unregularized one as $\zeta \to \infty$.
This means that, as seen from the problematic computational behavior of the unregularized model, $\zeta$ cannot be taken too large in this limited data setting.
In other words, the choice of $\zeta$ has to reflect some degree of prior information on $\beta_j$'s.
We need not assume strong prior information, however;
Figure~\ref{fig:bridge_traceplot_with_varying_slab_size} demonstrates that even small amount of regularization (e.g. $\zeta = 2 \text{ or } 4$) can noticeably improve the computational behavior over the unregularized case.

\begin{figure}
\begin{subfigure}{\linewidth}
\centering
\includegraphics[width=.9\linewidth]{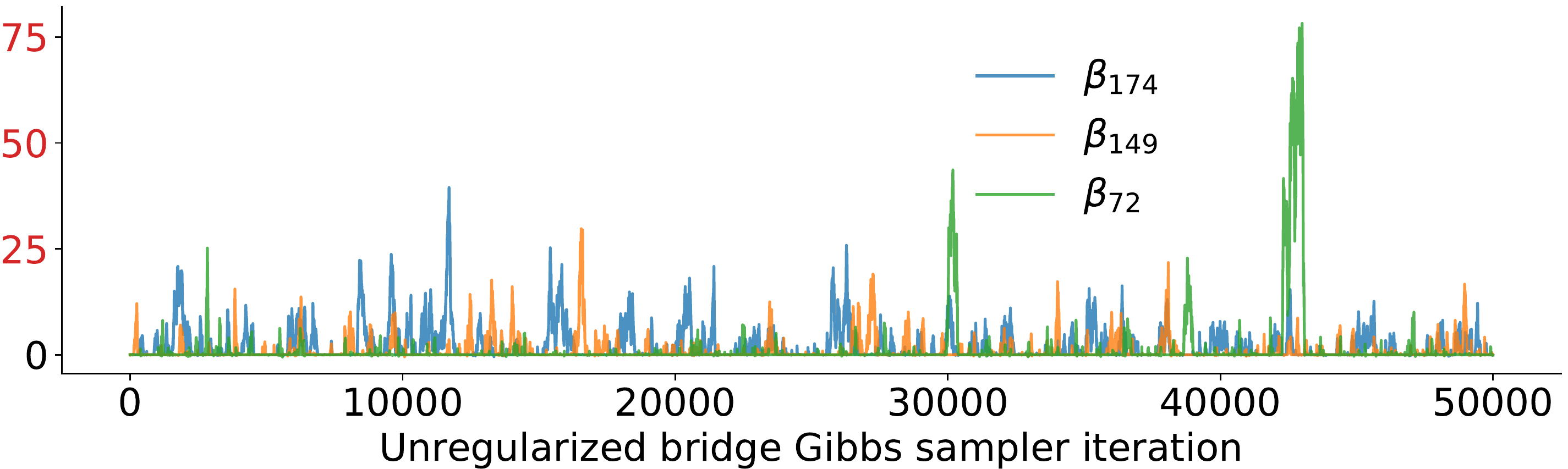}
\caption{%
	Without regularization, the Markov chain takes multiple ``excursions'' --- each lasting over hundreds of iterations --- into the unreasonable value range of the coefficients.
	The deviation in $\beta_{172}$ is particularly prominent around the 42,000th iteration. 
	More severe deviations may occur if the chain is run longer.
}
\label{fig:unregularized_bridge_traceplot}
\end{subfigure}
\begin{subfigure}{\linewidth}
\centering
\includegraphics[width=.9\linewidth]{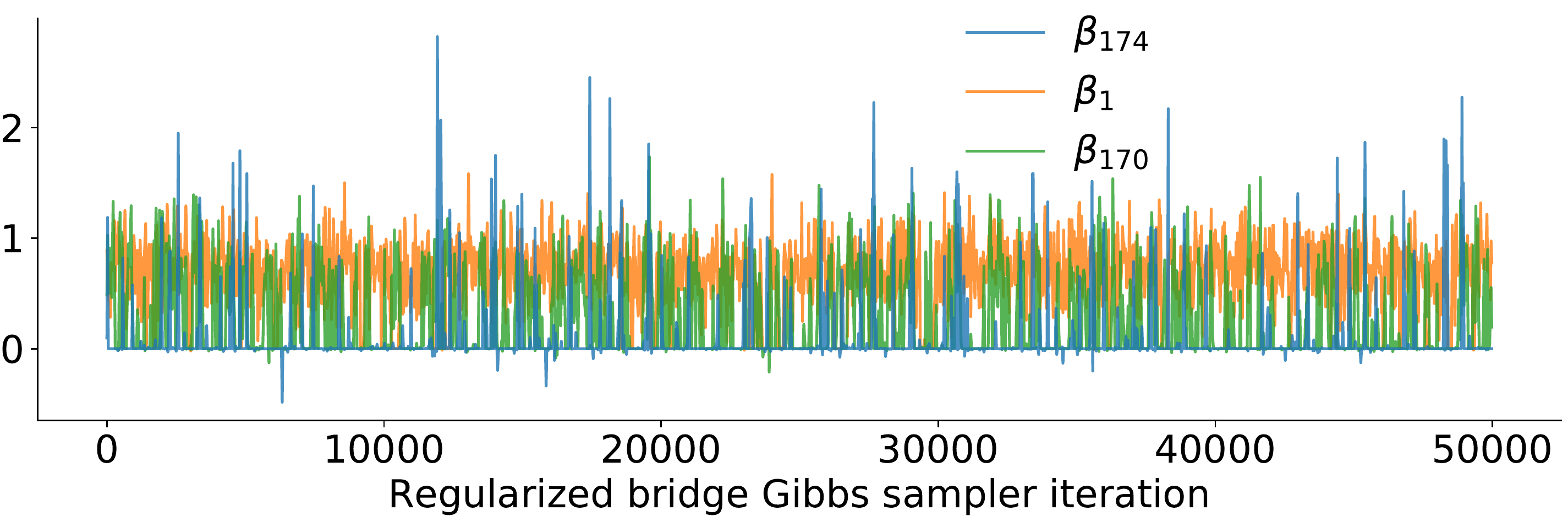}
\caption{%
	With regularization, the Markov chain does not display any serious mixing issues.
	The noticeable auto-correlation is due to the multi-modality of the posterior, an unavoidable feature of shrinkage models.
	Note that the coefficients with widest credible intervals do not coincide with the unregularized setting.%
}
\label{fig:regularized_bridge_traceplot}
\end{subfigure}
\caption{%
	Traceplot under the Bayesian bridge logistic regression with exponent $1 / 16$.
	Shown are the three coefficients with most potentially problematic mixing behaviors; see the main text for the details on our criteria.
}
\label{fig:bridge_traceplot}
\end{figure}

\begin{figure}
\includegraphics[width=\linewidth]{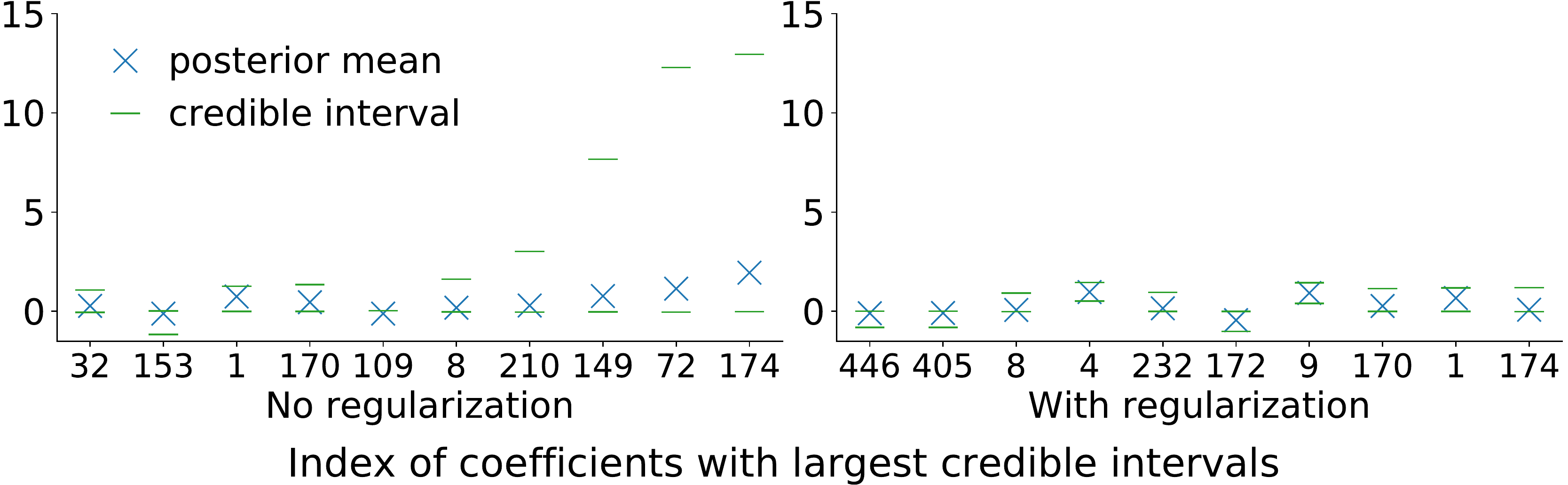}
\caption{%
	Ten widest 95\% posterior credible intervals under the Bayesian bridge logistic regression with (right) and without (left) regularization.
	Without regularization, the intervals are unrealistically large compared to the signal size of $\beta_j = 1$ for $j = 1, \ldots, 10$.
}
\label{fig:credible_interval_with_and_wo_regularization_comparison}
\end{figure}

\begin{figure}

	\begin{subfigure}{\linewidth}
	\centering
	\includegraphics[width=.875\linewidth]{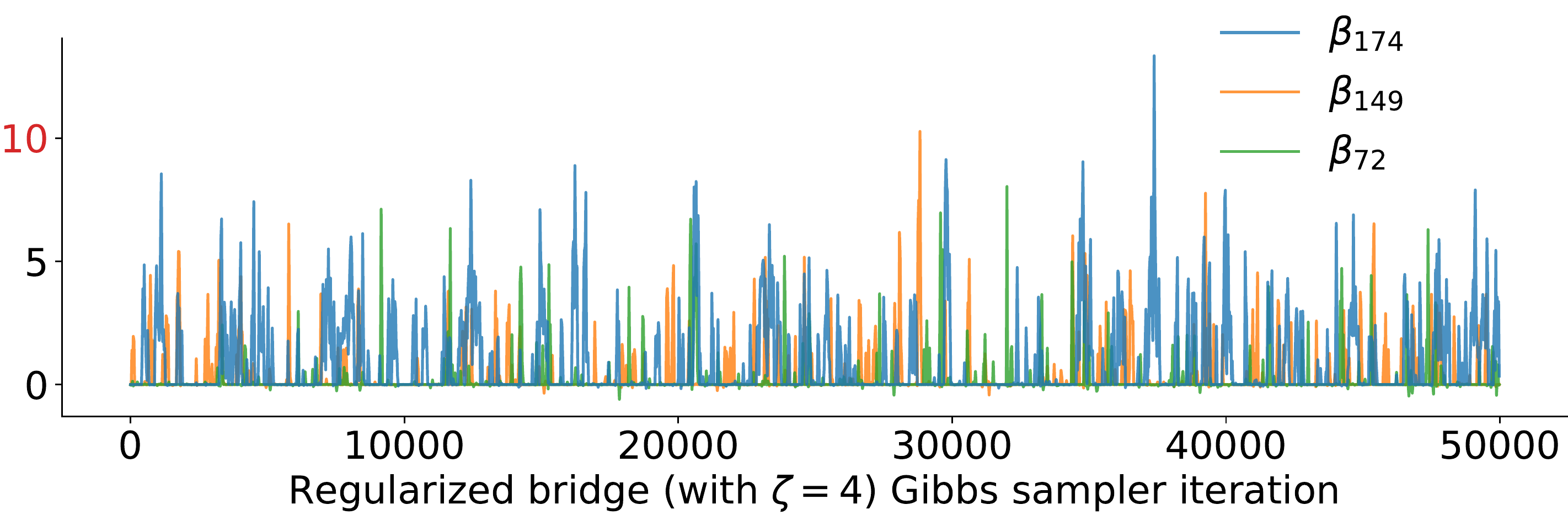}
	\end{subfigure}
	\vspace{.25\baselineskip}

	\begin{subfigure}{\linewidth}
	\centering
	\includegraphics[width=.875\linewidth]{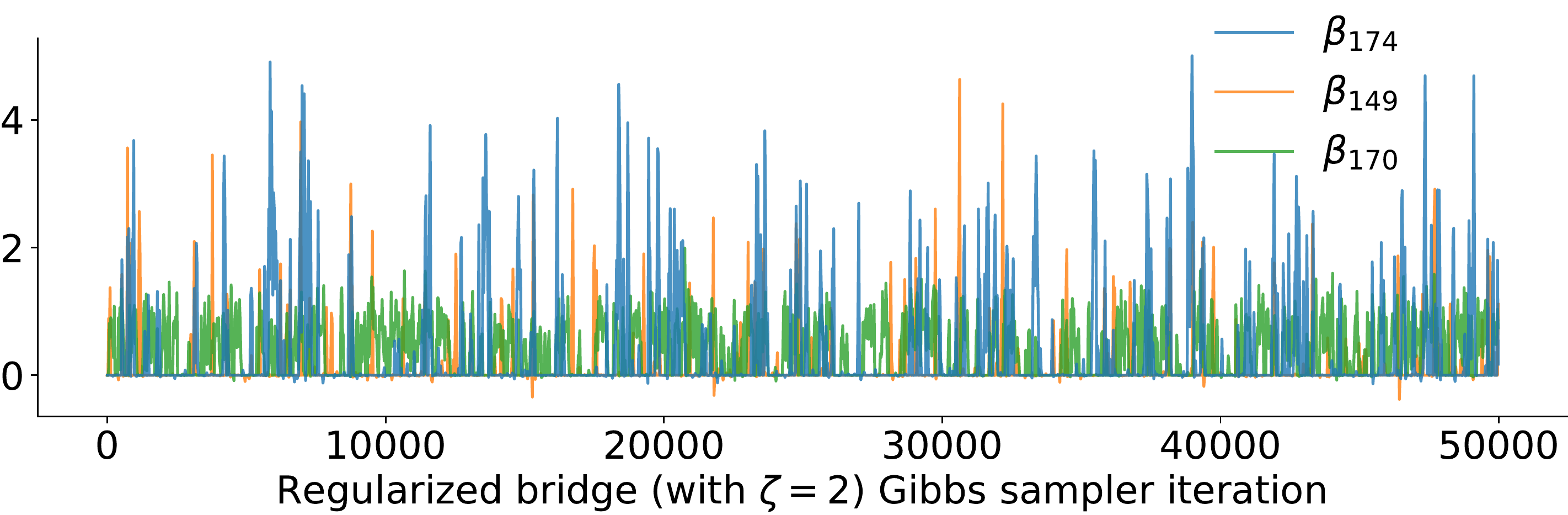}
	\end{subfigure}

\caption{%
	Traceplots under different slab widths: $\zeta = 2$ (bottom) and $\zeta = 4$ (top).
	The settings are otherwise identical to those of Figure~\ref{fig:bridge_traceplot}.
	As before, the three coefficients with most problematic mixing behaviors do not always coincide across different slab widths.
}
\label{fig:bridge_traceplot_with_varying_slab_size}
\end{figure}

\FloatBarrier
\subsection{Statistical properties of shrinkage model for weak signals}
To study the shrinkage model's ability to separate out the non-zero $\beta_j$ from the $\beta_j = 0$, we simulate 10 replicate data sets and estimate the posterior for each of them.
In total, we obtain 5,000 marginal posterior distributions --- 10 independent replication for each of the $p = 500$ regression coefficients --- with 100 for the signal $\beta_j = 1$ and 4,900 for the non-signal $\beta_j = 0$.
As all the predictors $\x_j$'s are simulated in an exchangeable manner, the 100 (and 4,900) posterior marginals for the signal (and non-signal) are also exchangeable.

Figure~\ref{fig:credible_intervals_under_bridge_1_16th_prior} show the posterior credible intervals.
Due to the low incidence rate and infrequent binary features,
many of the signals are too weak to be detected.
We also find that the credible intervals seemingly do not achieve their nominal frequentist coverage for signals below detection strength.
This finding is consistent with the existing theoretical results on shrinkage priors and is unsurprising in light of the impossibility theorem by \cite{li1989confidence_interval} --- confidence intervals cannot be optimally tight and have nominal coverage at the same time.
Credible intervals produced by Bayesian shrinkage models tend to be optimally tight and thus require appropriate manual adjustments to achieve the nominal coverage \citep{vanDerPas2017horseshoe_contraction}.
No statistical procedure is immune to this tightness-coverage trade-off; therefore, the apparent under-coverage should be seen not as a flaw but more as a feature of Bayesian shrinkage models.

We benchmark the signal detection capability of the posterior against the frequentist lasso, arguably the most widely-used approach to feature selection.
Obtaining the lasso point estimates requires a selection of the hyper-parameter commonly referred to as the \textit{penalty}  parameter.
For its choice, we first follow the standard practice of minimizing the ten-fold cross-validation errors \citep{hastie2009statistical-learning}.
However, this approach yields inconsistent and poor overall performance, detecting only 13 out of the 100 signals (Figure~\ref{fig:credible_intervals_under_bridge_1_16th_prior}).
Cross-validation likely fails here because each fold does not capture the characteristics of the whole data when the signals are so weak.
As a more stable alternative for calibrating the penalty parameter, we try an empirical Bayes procedure based on the Bayesian interpretation of the lasso \citep{park2008bayesian_lasso}.
We first estimate the posterior marginal mean of the penalty parameter from the Bayesian lasso Gibbs sampler.
Conditionally on this value,  we then find the posterior mode of $\bbeta$.
This procedure seems to yield more consistent performance, detecting 39 out of the 100 signals albeit with the estimates more shrunk towards null than the Bayesian posterior means.
The empirical Bayes procedure demonstrates more consistent behavior for the non-signals as well (Figure~\ref{fig:non_signal_lasso_bayes_estimate_comparison_under_bridge_1_16th_prior}).

\begin{figure}
	\begin{subfigure}{\linewidth}
	\includegraphics[width=\linewidth]{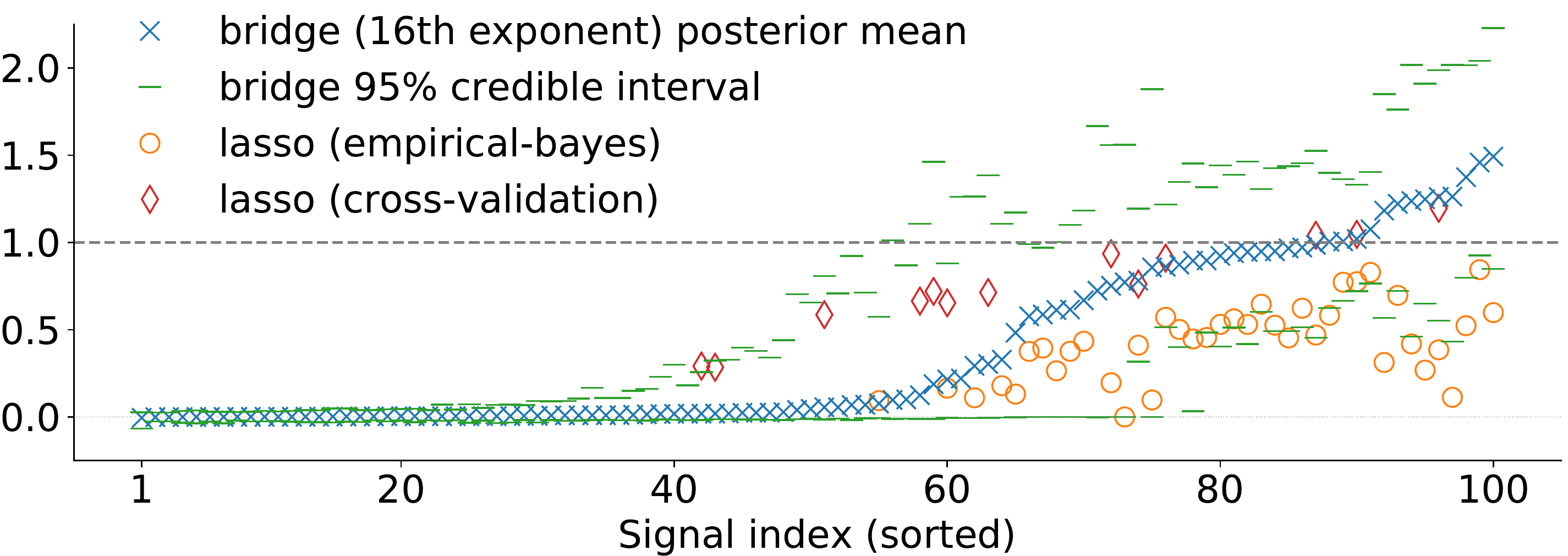}
	\end{subfigure}
	\\ \vspace*{.5\baselineskip}
	\begin{subfigure}{\linewidth}
	\includegraphics[width=\linewidth]{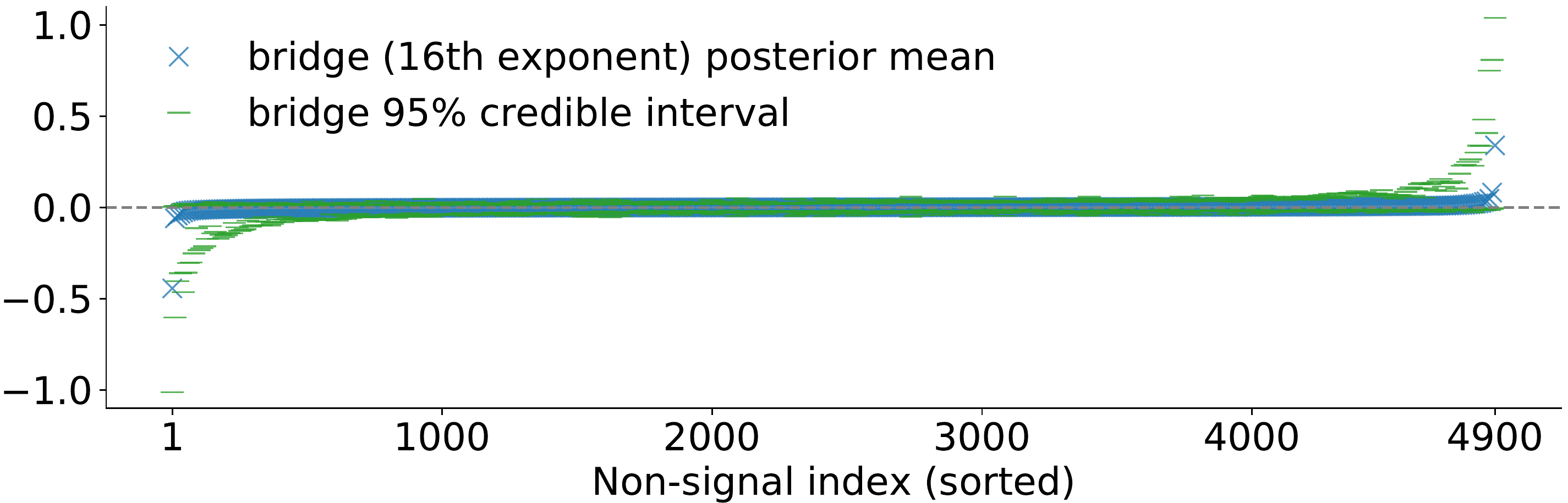}
	\end{subfigure}
	\caption{%
		The 95\% posterior credible intervals for the signals $\beta_j = 1$ (top) and non-signals $\beta_j = 0$ (bottom) under the Bayesian bridge logistic regression with the bridge exponent $1 / 16$.
		The intervals are sorted by the posterior means.
		To avoid clutter, the top plot shows only the non-zero values of the lasso estimates.
		The lasso estimates for the non-signals are summarized in Figure~\ref{fig:non_signal_lasso_bayes_estimate_comparison_under_bridge_1_16th_prior} and are not shown in the bottom plot.
	}
	\label{fig:credible_intervals_under_bridge_1_16th_prior}
\end{figure}

\begin{figure}
\begin{minipage}{.57\linewidth}
\includegraphics[width=\linewidth]{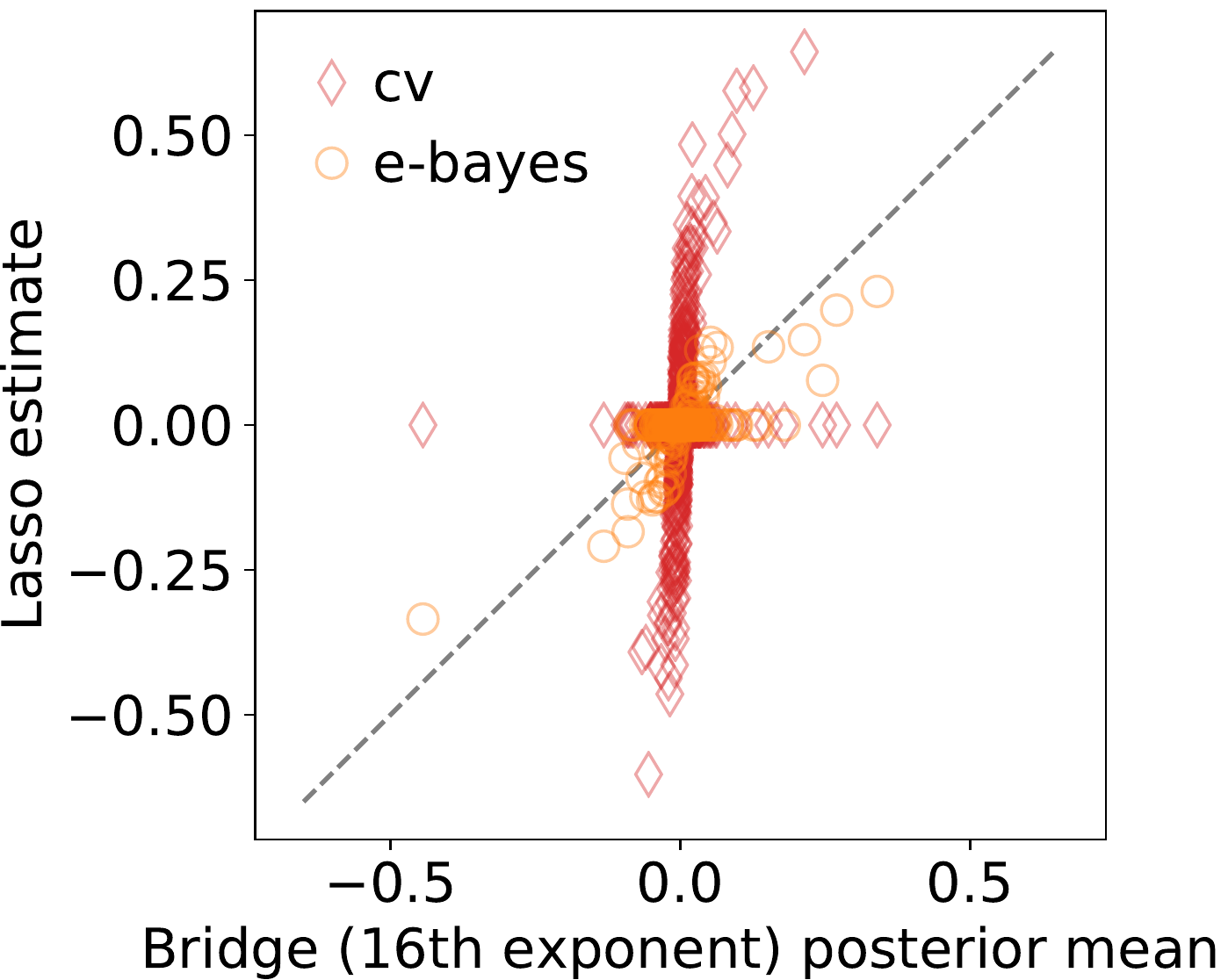}
\end{minipage}
~
\begin{minipage}{.4\linewidth}
\caption{%
	Comparison of the 4,900 Bayesian bridge posterior means and lasso estimates for the non-signals $\beta_j = 0$.
	Lasso with cross-validation produces a larger number of false positives.
	Lasso with the empirical Bayes calibration yields the estimates more in line with the bridge posterior.
}
\label{fig:non_signal_lasso_bayes_estimate_comparison_under_bridge_1_16th_prior}
\end{minipage}
\end{figure}

We also assess how the spike size and (pre-regularization) tail behavior of the prior influence statistical properties of the resulting posterior.
For this purpose, we fit the regularized bridge model with the exponent $\bridgeExponent^{-1} \in \{2, 4, 8, 16\}$ to the same data sets.
Figure~\ref{fig:credible_intervals_under_bridge_1_4th_prior} summarizes the credible intervals under the $\bridgeExponent = 1 / 4$ case.
The credible intervals are centered around the values similar to the $\bridgeExponent = 1 / 16$ case (Figure~\ref{fig:credible_intervals_under_bridge_1_16th_prior}), but are much wider overall.
We observe the same pattern throughout the range of the exponent values: similar median values, but tighter intervals for the smaller exponents.
In particular, as can be seen in Figure~\ref{fig:credible_interval_width_vs_coverage}, more ``extreme'' shrinkage priors with larger spikes and heavier-tails seem to yield tighter credible intervals for the same coverage.

\begin{figure}
	\begin{subfigure}{\linewidth}
	\includegraphics[width=\linewidth]{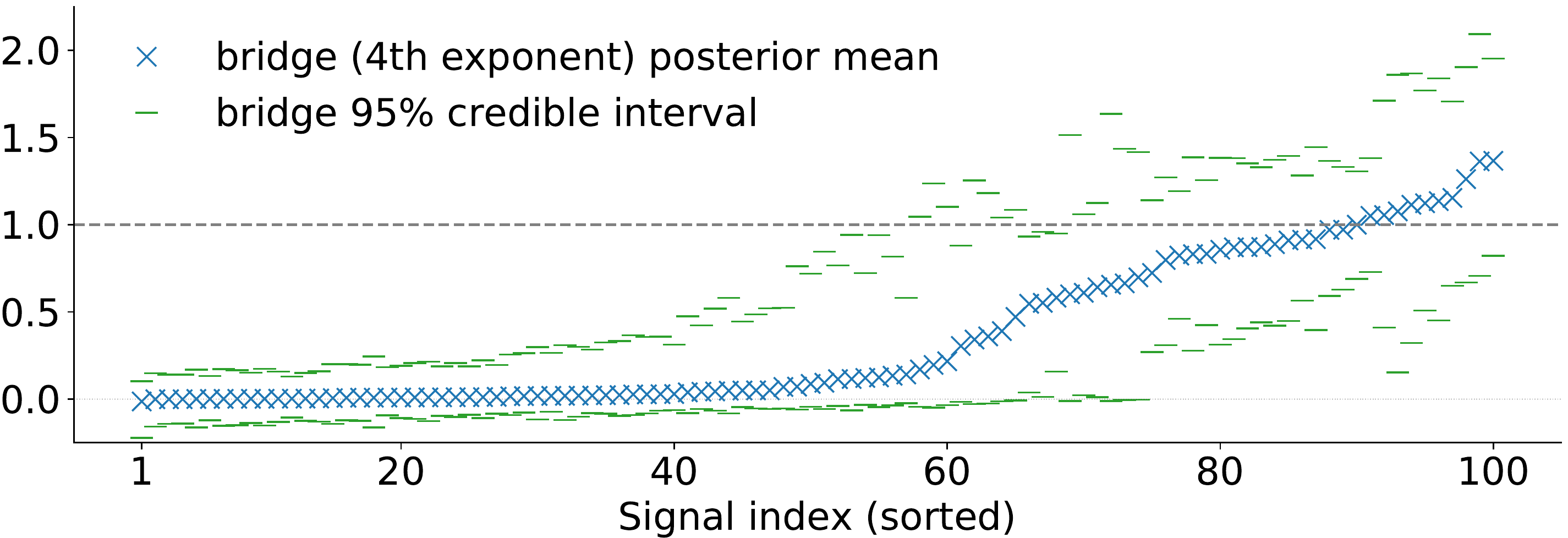}
	\end{subfigure}
	\begin{subfigure}{\linewidth}
	\includegraphics[width=\linewidth]{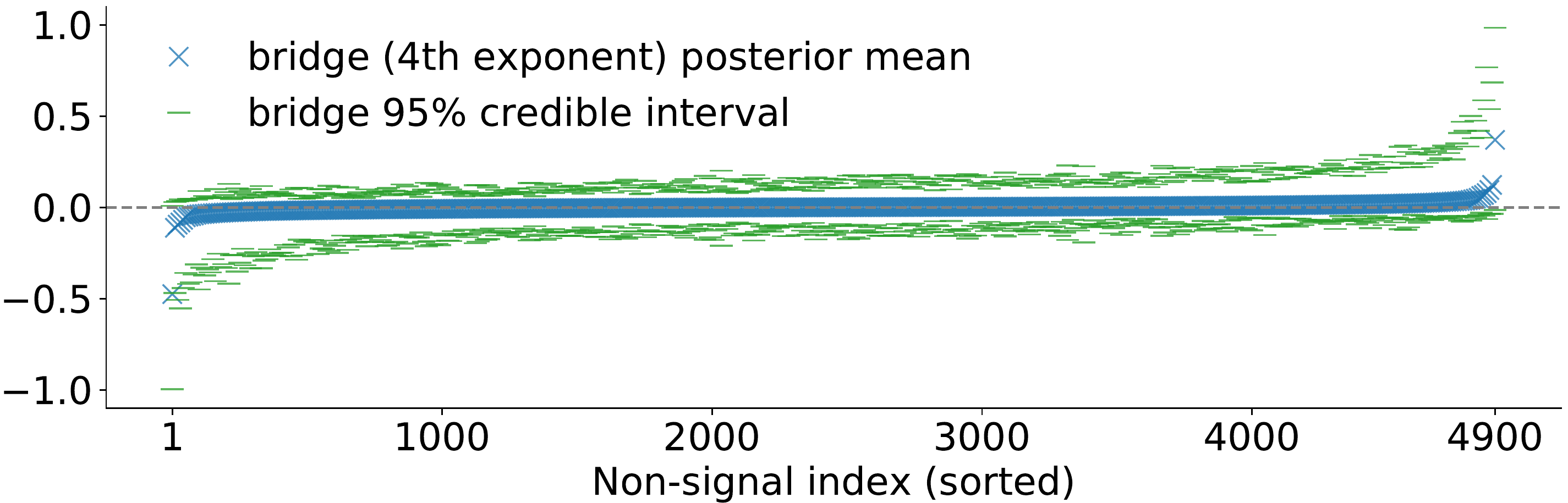}
	\end{subfigure}
	\caption{%
		The 95\% posterior credible intervals under the Bayesian bridge logistic regression with the bridge exponent $1 / 4$.
		Compared with the $1 / 16$ exponent case (Figure~\ref{fig:credible_intervals_under_bridge_1_16th_prior}), the posterior distributions have similar means but much wider credible intervals.
	}
	\label{fig:credible_intervals_under_bridge_1_4th_prior}
\end{figure}

\begin{figure}
\includegraphics[width=.95\linewidth]{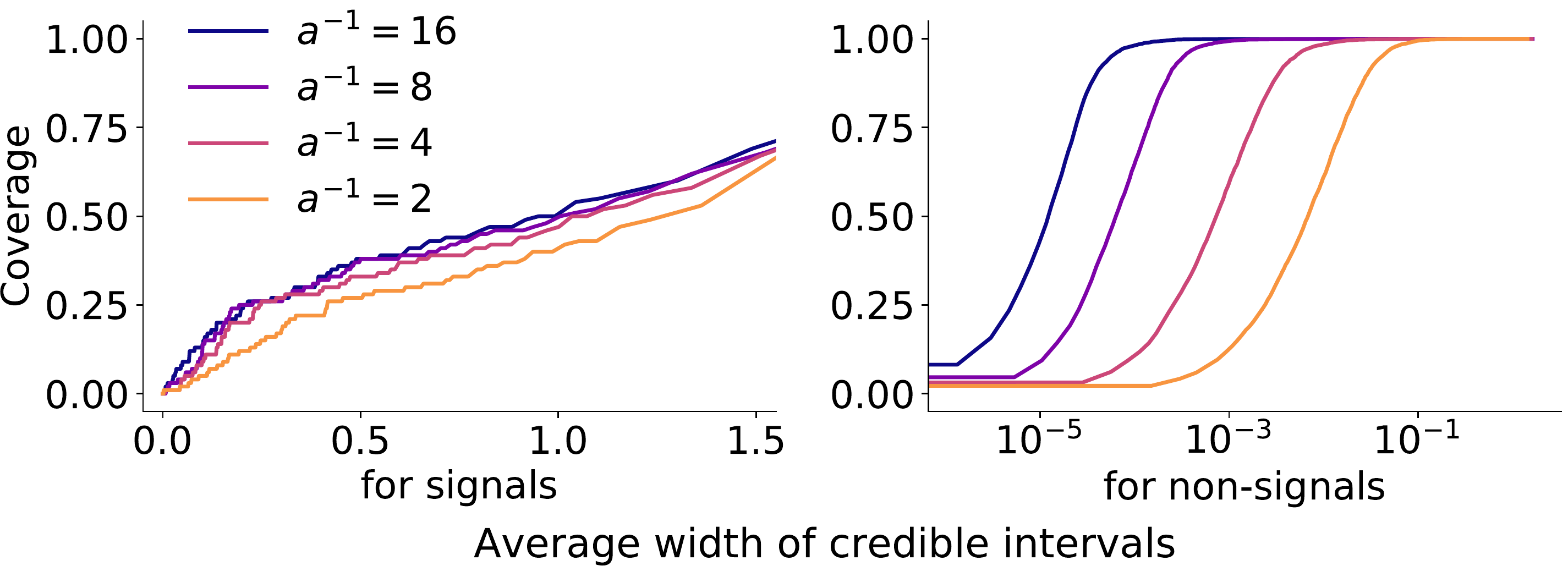}
\caption{%
	Average width v.s.\ coverage of the credible intervals.
	The plots are produced by computing the equal-tailed credible intervals at a range of credible levels.
	The $x$-axis is in the $\log_{10}$ scale for the non-signals.
}
\label{fig:credible_interval_width_vs_coverage}
\end{figure}

\section{Discussion}
Shrinkage priors have been adopted in a variety of Bayesian models, but the potential issues arising from their heavy-tails are often overlooked.
Our method provides a simple and convenient way to regularize shrinkage priors, making the posterior inference more robust.
Both the theoretical and empirical results demonstrate the benefits of regularization in improving the statistical and computational properties when parameters are only weakly identified.
Much of the systematic investigations into the shrinkage prior properties has so far focused on rather simple models and situations in which signals are reasonable strong.
Our work adds to the emerging efforts to better understand the behavior of shrinkage models in more complex settings.

\section*{Acknowledgments}
We are indebted to Andrew Holbrook for the alliteration in the article title.
This work was partially supported through National Institutes of Health grants R01 AI107034 and U19 AI135995 and through Food and Drug Administration grant HHS 75F40120D00039.

\FloatBarrier
\bibliographystyle{ba}
\bibliography{ergodicity_under_regularized_shrinkage}{}

\newpage
\appendix

\renewcommand{\thetheorem}{\Alph{section}.\arabic{theorem}}
\renewcommand{\theequation}{\Alph{section}.\arabic{equation}}
\renewcommand{\thefigure}{\Alph{section}.\arabic{figure}}

\section{Alternative definition of proposed regularization}
\label{sec:alt_regularization}
In Section~\ref{sec:gibbs_friendly_regularization}, we described our regularization approach as effectively modifying the prior on $\beta_j$ through the likelihood of fictitious data $z_j$.
While many properties of the resulting posterior are most apparent from this formulation, we can forgo the use of fictitious data and achieve the identical effect via direct modification of a shrinkage prior as follows.
We define the regularized prior $\pi_{\rm reg}(\cdot)$ by setting the distribution of $\beta_j, \lshrink_j \given \gshrink, \zeta$ as
\begin{align*}
\pi_{\rm reg}(\beta_j, \lshrink_j \given \gshrink, \zeta)
	&\propto  
			\exp \! \left( 
				- \frac{\beta_j^2}{2 \zeta^2} 
			\right) \,  
			\frac{1}{\gshrink \lshrink_j} 
			\exp \! \left( 
			- \frac{\beta_j^2}{2 \gshrink^2 \lshrink_j^2} 
		\right) \,
	\localPrior(\lshrink_j) \nonumber \\
	&\propto \normal \! \left(\beta_j \, \Bigg| \
				0, \left( \frac{1}{\zeta^2} + \frac{1}{\gshrink^2 \lshrink_j^2} \right)^{-1}
			\right)
		\left( 1 + \frac{\gshrink^2 \lshrink_j^2}{\zeta^2} \right)^{-1/2}
		\localPrior(\lshrink_j)
\end{align*}
where $\normal(\, \cdot \given 0, \sigma^2)$ denotes the centered Gaussian density with variance $\sigma^2$. 
In other words, in addition to defining $\pi(\beta_j \given \gshrink, \lshrink_j, \zeta)$ as in \eqref{eq:regularized_global_local_mixture}, we alter the prior on $\lshrink_j$ as 
$\pi(\lshrink_j \given \gshrink, \zeta) \propto \localPrior(\lshrink_j) / \sqrt{1 + \gshrink^2 \lshrink_j^2 / \zeta^2}$.
Incidentally, we see that our regularized prior is very similar to that of \cite{piironen2017regularized-horseshoe}, but has a slightly lighter tail due to the factor $1 / \sqrt{1 + \gshrink^2 \lshrink_j^2 / \zeta^2}$ which, as $\lshrink_j \to \infty$, behaves like $\zeta / \gshrink \lshrink_j$.

\newcommand{\probitLik}{L_{\textrm{probit}}}
\newcommand{\gaussianCdf}{\Phi}
\newcommand{\cdfArg}{t}
\section{%
	Further results on behavior of shrinkage\\ model Gibbs samplers: probit regression as example
}
\label{sec:further_results_on_general_shrinkage_model_Gibbs_behavior}

As we discussed in Section~\ref{sec:gibbs_sampler_behavior_near_spike}, Proposition~\ref{prop:local_scale_posterior_in_the_limit} and \ref{prop:scale_negative_power_expectation_bound} are quite general in scope and can provide insight into behavior of shrinkage model Gibbs samplers more broadly. 

Here we demonstrate the broader relevance of these results, as well as of a few additional results, by applying them to establish uniform/geometric ergodicity of a Gibbs sampler for regularized Bayesian sparse probit regression.
More explicitly, we consider the model
\begin{equation*}
\begin{gathered}
y_i \given \x_i, \bbeta
	\sim \bernoulli\big(\gaussianCdf(\x_i^\transpose \bbeta) \big), \ \,
z_j = 0 \given \beta_j
	\sim \normal(0, \zeta^2), \\
\beta_j \given \gshrink, \lshrink_j
	\sim \normal(0, \gshrink^2 \lshrink_j^2), \
\gshrink \sim \globalPrior(\cdot), \
\lshrink_j \sim \localPrior(\cdot),
\end{gathered}
\end{equation*}
where $\gaussianCdf(\cdfArg)$ denotes the cumulative distribution function of the standard Gaussian.
The corresponding Gibbs sampler induces a transition kernel $(\bbetaPrev, \blshrinkPrev, \gshrinkPrev) \to (\bbeta, \blshrink, \gshrink)$ through the following cycle of conditional updates:
\begin{enumerate}[topsep=.25\baselineskip, itemsep=.25\baselineskip]
\item Draw $\gshrink \given \bbetaPrev, \blshrinkPrev$ from the density proportional to \eqref{eq:scale_parameter_posterior}.
When using Bayesian bridge priors, draw from the collapsed distribution $\gshrink \given \bbetaPrev$ (Appendix~\ref{sec:bridge_prior_properties}).
\item Draw $\blshrink \given \bbetaPrev, \gshrink$ from the density proportional to \eqref{eq:scale_parameter_posterior}.
\item Draw $\bbeta \given \gshrink, \blshrink, \y, \X, \z = \bm{0}$ from the density proportional to 
\begin{equation}
\label{eq:reg_coef_unnormalized_cond_density_for_probit}
\begin{aligned}
\pi\left(
	\bbeta \given \gshrink, \blshrink, \y, \X, \z = \bm{0}
\right)
	&\propto \probitLik(\y \given \X, \bbeta) \, L(\z = \bm{0} \given \bbeta) \, \pi(\bbeta \given \gshrink, \blshrink) \\
	&\propto \probitLik(\y \given \X, \bbeta) \, \pi(\bbeta \given \gshrink, \blshrink, \z = \bm{0})
\end{aligned}
\end{equation}
where $\probitLik(\y \given \X, \bbeta) = {\textstyle \prod_i} \, \gaussianCdf(\x_i^\transpose \bbeta)^{y_i} \left( 1 - \gaussianCdf(\x_i^\transpose \bbeta) \right)^{1 - y_i}$ is the probit likelihood.
The density \eqref{eq:reg_coef_unnormalized_cond_density_for_probit} belongs to a unified skew-normal family, from which we can draw independent samples by the algorithm of \cite{durante2019conjugate_probit}.
\end{enumerate}
\noindent Borrowing a terminology from \cite{durante2019conjugate_probit}, we refer to the above Gibbs sampler as the \textit{conjugate Gibbs sampler} for probit model to distinguish it from the more traditional one based on the data augmentation scheme of \cite{albert1993bayesian_probit}.

Theorem~\ref{thm:uniform_ergodicity_of_bayesian_bridge_for_probit} and \ref{thm:geom_ergodicity_for_probit} below provide uniform and geometric ergodicity results for the conjugate Gibbs sampler and are exact analogues of the corresponding results Theorem~\ref{thm:uniform_ergodicity_of_bayesian_bridge} and \ref{thm:geom_ergodicity} for the logistic case.

\begin{theorem}[Uniform ergodicity for probit model]
\label{thm:uniform_ergodicity_of_bayesian_bridge_for_probit}
If the prior $\globalPrior(\cdot)$ is supported on $[\gshrink_{\min}, \infty)$ for $\gshrink_{\min} > 0$, then the conjugate Gibbs sampler for regularized Baysian bridge probit regression is uniformly ergodic.
\end{theorem}

\begin{theorem}[Geometric ergodicity for probit model]
\label{thm:geom_ergodicity_for_probit}
Suppose that the local scale prior satisfies $\| \localPrior \|_\infty < \infty$ and that the global scale prior $\globalPrior(\cdot)$ is supported on $[\gshrink_{\min}, \gshrink_{\max}]$ for $0 < \gshrink_{\min} \leq \gshrink_{\max} < \infty$.
Then the conjugate Gibbs sampler for regularized sparse probit regression is geometrically ergodic.
\end{theorem}

\subsection{Proofs of Theorem~\ref{thm:uniform_ergodicity_of_bayesian_bridge_for_probit} and \ref{thm:geom_ergodicity_for_probit}}
\label{sec:ergodicity_proof_for_probit_model}

The proof of Theorem~\ref{thm:uniform_ergodicity_of_bayesian_bridge_for_probit} (and \ref{thm:geom_ergodicity_for_probit}) above follows a path essentially identical to the proof of Theorem~\ref{thm:uniform_ergodicity_of_bayesian_bridge} (and \ref{thm:geom_ergodicity}) with most arguments carrying through verbatim or with trivial modifications;
we only need to replace a few model-specific inequalities with the corresponding ones for the probit model. 
For establishing minorization conditions, Lemma~\ref{lem:minorization_for_probit} below replaces Lemma~\ref{lem:minorization}.
For establishing drift conditions, 
the bound on the conditional expectation of $| \beta_j |^{-\lyapunovExponent}$ in Lemma~\ref{lem:negative_power_expectation_bound} replaces Eq.~\eqref{eq:expectation_wrt_polya_gamma}, and 
the bound on the conditional expectation of $\| \beta \|^2$ in Lemma~\ref{lem:sqnorm_expectation_bound_for_probit} replaces Eq.~\eqref{eq:bound_on_expectation_of_sqnorm}.
Remarkably, Lemma~\ref{lem:negative_power_expectation_bound} and \ref{lem:sqnorm_expectation_bound_for_probit} only requires a likelihood $L(\y \given \X, \bbeta)$ to be a bounded function of $\bbeta$ and thus may be applicable beyond the probit case. 

We sketch out the proofs of Theorem~\ref{thm:uniform_ergodicity_of_bayesian_bridge_for_probit} and \ref{thm:geom_ergodicity_for_probit} below.
Again, the omitted details are essentially identical to the logistic case or, in fact, simpler because the probit case does not involve the additional \polyagamma{} parameter.

\begin{proof}[Proof of Theorem~\ref{thm:uniform_ergodicity_of_bayesian_bridge_for_probit}]
A minorization result analogous to Theorem~\ref{thm:minorization} follows from Proposition~\ref{prop:local_scale_posterior_in_the_limit} and Lemma~\ref{lem:minorization_for_probit}.
This minorization result straightforwardly implies a uniform minorization under Bayesian bridge priors as in Theorem~\ref{thm:uniform_ergodicity_of_bayesian_bridge}.
See the proofs of Theorem~\ref{thm:minorization} and \ref{thm:uniform_ergodicity_of_bayesian_bridge} for details.
\end{proof}

\begin{proof}[Proof of Theorem~\ref{thm:geom_ergodicity_for_probit}]
A minorization result analogous to Theorem~\ref{thm:non_uniform_minorization} follows from Lemma~\ref{lem:minorization_for_probit}.
Proposition~\ref{prop:scale_negative_power_expectation_bound}, Lemma~\ref{lem:negative_power_expectation_bound}, and Lemma~\ref{lem:sqnorm_expectation_bound_for_probit} together imply that $V(\bbeta) = \sum_j |\beta_j|^{-\alpha} + \| \bbeta \|^2$ is a Lyapunov function as in the proofs of Theorem~\ref{thm:drift} and \ref{thm:geom_ergodicity}.
The geometric ergodicity then follows from the minorization and drift condition.
See the proofs of Theorem~\ref{thm:non_uniform_minorization}, \ref{thm:drift}, and \ref{thm:geom_ergodicity} for details.
\end{proof}

\subsection{Minorization lemma for probit model}

\begin{lemma}
\label{lem:minorization_for_probit}
Whenever $\min_j \gshrink \lshrink_j \geq R > 0$, there are $\tilde{\delta}, \tilde{\delta}' > 0$ --- independent of $\gshrink$ and $\blshrink$ except through $R$ --- 
such that the following minorization condition holds:
\begin{equation}
\label{eq:minorization_for_probit}
\begin{aligned}
&\pi(\bbeta \given \gshrink, \blshrink,  \y, \X, \z = \bm{0}) \\
	&\hspace{2.5em}\geq \tilde{\delta}
	\, \probitLik(\y \given \X, \bbeta) 
	\, \normal \left(\bbeta; \, \bm{0}, (\zeta^{-2} + R^{-2})^{-1} \I \right) \\
	&\hspace{2.5em}\geq \tilde{\delta}' 
	\, \normal \! \left(\bbeta; \, \bm{0}, \left[ \X^\transpose \X + (\zeta^{-2} + R^{-2}) \I \right]^{-1} \right).
\end{aligned}
\end{equation}
\end{lemma}

\begin{proof}
The conditional distribution of $\bbeta \given \gshrink, \blshrink, \y, \X$ is given by
\begin{equation}
\label{eq:reg_coef_conditional_density_for_probit}
\begin{aligned}
&\pi(\bbeta \given \gshrink, \blshrink, \y, \X, \z = \bm{0})
	= \frac{
		\probitLik(\y \given \X, \bbeta) \, \pi(\bbeta \given \gshrink, \blshrink, \z = \bm{0})
	}{
		\int \probitLik(\y \given \X, \bbeta') \, \pi(\bbeta' \given \gshrink, \blshrink, \z = \bm{0}) \, \diff \bbeta'
	}.
\end{aligned}
\end{equation}
Since $\gaussianCdf(\cdfArg) = 1 - \gaussianCdf(-\cdfArg) \leq 1$ for all $\cdfArg$, we have $\| \probitLik \|_\infty \leq 1$ and
\begin{equation}
\label{eq:denom_upper_bound_for_probit_minorization}
\int \probitLik(\y \given \X, \bbeta') \, \pi(\bbeta' \given \gshrink, \blshrink, \z = \bm{0}) \, \diff \bbeta'
	\leq \int \pi(\bbeta' \given \gshrink, \blshrink, \z = \bm{0}) \, \diff \bbeta'
	= 1.
\end{equation}
Also, we can easily verify that the following inequality holds whenever $\min_j \gshrink \lshrink_j \geq R$: 
\begin{equation}
\label{eq:num_lower_bound_for_probit_minorization}
\begin{aligned}
\pi(\bbeta \given \gshrink, \blshrink, \z = \bm{0})
	&= \prod_j \frac{1}{\sqrt{2 \pi}} \left( \zeta^{-2} + \gshrink^{-2} \lshrink_j^{-2} \right)^{1/2} \exp\!\left( - \frac{1}{2} \left( \zeta^{-2} + \gshrink^{-2} \lshrink_j^{-2} \right) \beta_j^2 \right) \\
	&\geq \prod_j \frac{1}{\sqrt{2 \pi} \zeta} \exp\left( - \frac{1}{2} \left( \zeta^{-2} + R^{-2} \right) \beta_j^2 \right).
\end{aligned}
\end{equation}
Combining \eqref{eq:denom_upper_bound_for_probit_minorization} and \eqref{eq:num_lower_bound_for_probit_minorization}, we can lower bound \eqref{eq:reg_coef_conditional_density_for_probit} with $\tilde{\delta} > 0$ as 
\begin{equation}
\pi(\bbeta \given \gshrink, \blshrink, \y, \X, \z = \bm{0})
	\geq \tilde{\delta} \, \probitLik(\y \given \X, \bbeta) \ \normal \left(\bbeta; \, \bm{0}, (\zeta^{-2} + R^{-2})^{-1} \I \right),
\end{equation}
establishing the first inequality in \eqref{eq:minorization_for_probit}.

To establish the second inequality in \eqref{eq:minorization_for_probit}, we will show that
\begin{equation}
\label{eq:lower_bound_for_probit}
\min\{\gaussianCdf(\cdfArg), 1 - \gaussianCdf(\cdfArg) \}
	\geq \min\left\{ 
		1 - \gaussianCdf(1), \frac{1}{2 \sqrt{2 \pi}}
	\right\} \exp\left(- \cdfArg^2 \right);
\end{equation}
this will imply $\probitLik(\y \given \X, \bbeta) \geq \min\left\{ 1 - \gaussianCdf(1), (2 \sqrt{2 \pi})^{-1} \right\} \exp(- \| \X \bbeta \|^2) $ and complete the proof.
Eq 7.1.13 of \cite{abramowitz1965math_func_handbook} tells us that
\begin{equation}
1 - \gaussianCdf(\cdfArg)
	\geq \frac{1}{\sqrt{2 \pi}} \frac{\cdfArg}{\cdfArg^2 + 1} \exp\left(- \frac{\cdfArg^2}{2} \right).
\end{equation}
We therefore have
\begin{equation}
1 - \gaussianCdf(\cdfArg)
	\geq \frac{1}{2 \sqrt{2 \pi}} \frac{1}{\cdfArg} \exp\left(- \frac{\cdfArg^2}{2} \right)
	\geq \frac{1}{2 \sqrt{2 \pi}} \exp\left(- \cdfArg^2 \right)
	\ \text{ for } \ \cdfArg \geq 1;
\end{equation}
the latter inequality follows from the fact that $\cdfArg^{-1} \geq \exp(- \cdfArg^2 / 2)$ for $\cdfArg \geq 1$, which can be proven, for example, by noting that $\frac{\diff}{\diff \cdfArg} \left( \cdfArg \exp(- \cdfArg^2 / 2) \right) \leq 0$ for $\cdfArg \geq 1$.
For $\cdfArg \leq 1$, we have $1 - \gaussianCdf(\cdfArg) \geq 1 - \gaussianCdf(1)$ since $\gaussianCdf(\cdfArg)$ is increasing in $\cdfArg$. 
Combining the lower bounds for $t \geq 1$ and $t \leq 1$, we obtain
\begin{equation*}
1 - \gaussianCdf(\cdfArg)
	\geq \min\left\{ 
		1 - \gaussianCdf(1), \frac{1}{2 \sqrt{2 \pi}} \exp\left(- \cdfArg^2 \right) 
	\right\}
	\geq \min\left\{ 
		1 - \gaussianCdf(1), \frac{1}{2 \sqrt{2 \pi}}
	\right\} \exp\left(- \cdfArg^2 \right).
\end{equation*}
Since $\gaussianCdf(\cdfArg) = 1 - \gaussianCdf(-\cdfArg)$, the same lower bound also holds for $\gaussianCdf(\cdfArg)$, yielding \eqref{eq:lower_bound_for_probit}.
\end{proof}

\subsection{Drift condition lemmas for bounded likelihood models}

As we mentioned in Section~\ref{sec:ergodicity_proof_for_probit_model}, Lemma~\ref{lem:negative_power_expectation_bound} and \ref{lem:sqnorm_expectation_bound_for_probit} here apply not only to the probit case but also to any model whose likelihood is a bounded function of $\bbeta$.
Lemma~\ref{lem:negative_power_expectation_bound} in particular holds with or without the fictitious likelihood $L(\z = \bm{0} \given \bbeta)$ for regularization.
While stated in terms of a generic bounded likelihood $L(\y \given \X, \bbeta)$, Lemma~\ref{lem:negative_power_expectation_bound} can be applied to regularized models simply by replacing the likelihood $\bbeta \to L(\y \given \X, \bbeta)$ in its statement with the regularized one $\bbeta \to L(\y \given \X, \bbeta) L(\z = \bm{0} \given \bbeta)$.

\begin{lemma}
\label{lem:negative_power_expectation_bound}
Let $\lyapunovExponent \in [0, 1)$.
Suppose the likelihood satisfies $\| L \|_\infty := \sup_{\bbeta} L(\y \given \X, \bbeta) < \infty$ and is strictly positive and continuous at $\bbeta = \bm{0}$.
Then the following inequality holds for the conditional expectation under $\bbeta \given \gshrink, \blshrink, \y, \X$ with constants $\const, \const' < \infty$ depending only on $\lyapunovExponent$ and functionals of the likelihood $\bbeta \to L(\y \given \X, \bbeta)$:
\begin{equation}
\label{eq:negative_power_expectation_bound}
\expectation \! \left[
	| \beta_j |^{-\lyapunovExponent} \given \gshrink, \blshrink, \y, \X
\right]
	\leq \const | \gshrink \lshrink_j |^{-\lyapunovExponent} + \const'.
\end{equation}
\end{lemma}

\begin{proof}
The conditional distribution of $\bbeta \given \gshrink, \blshrink, \y, \X$ is given by
\begin{equation}
\label{eq:reg_coef_conditional_density}
\pi(\bbeta \given \gshrink, \blshrink, \y, \X)
	= \frac{
		L(\y \given \X, \bbeta) \, \pi(\bbeta \given \gshrink, \blshrink)
	}{
		\int L(\y \given \X, \bbeta') \, \pi(\bbeta' \given \gshrink, \blshrink) \, \diff \bbeta'
	}.
\end{equation}
We consider the conditional expectation \eqref{eq:negative_power_expectation_bound} under two separate cases: $\max_j \gshrink \lshrink_j \leq \epsilon$ and $\min_j \gshrink \lshrink_j \geq \epsilon$, where $\epsilon > 0$ is any value small enough to guarantee the likelihood to be positive on the set $\| \bbeta' \|_\infty = \max_j |\beta_j'| \leq \epsilon$.

When $\max_j \gshrink \lshrink_j \leq \epsilon$, we have
\begin{align*}
\int L(\y \given \X, \bbeta') \, \pi(\bbeta' \given \gshrink, \blshrink) \, \diff \bbeta'
	&\geq \int_{ \| \bbeta' \|_\infty \, \leq \, \epsilon} L(\y \given \X, \bbeta') \, \pi(\bbeta' \given \gshrink, \blshrink) \, \diff \bbeta' \\
	&\geq \left( \min_{ \| \bbeta' \|_\infty\, \leq \, \epsilon} L(\y \given \X, \bbeta') \right)
		\prod_j \int_{-\epsilon}^{\epsilon} \pi(\beta_j '\given \gshrink, \lshrink_j) \, \diff \beta_j' \\
	&\geq \left( \min_{ \| \bbeta' \|_\infty \, \leq \, \epsilon} L(\y \given \X, \bbeta') \right)
		\big( \gaussianCdf(1) - \gaussianCdf(-1) \big)^p,
\end{align*}
where $\gaussianCdf(\cdot)$ is the cumulative distribution function of the standard Gaussian.
Using the above lower bound on the numerator, we can bound \eqref{eq:reg_coef_conditional_density} as
\begin{equation}
\label{eq:reg_coef_cond_density_bound_for_small_prior_sd}
\pi(\bbeta \given \gshrink, \blshrink, \y, \X)
	\leq \const_\epsilon \pi(\bbeta \given \gshrink, \blshrink)
\end{equation}
for $\const_\epsilon = \| L \|_\infty \big/ \left( \min_{ \| \bbeta' \|_\infty \, \leq \, \epsilon} L(\y \given \X, \bbeta') \right)\big( \gaussianCdf(1) - \gaussianCdf(-1) \big)^p$.
It now follows that
\begin{equation}
\label{eq:negative_power_cond_expectation_bound_inner_case}
\expectation \! \left[
	| \beta_j |^{-\lyapunovExponent} \given \gshrink, \blshrink, \y, \X
\right]
	\leq \const_\epsilon \, \expectation \! \left[
		| \beta_j |^{-\lyapunovExponent} \given \gshrink, \blshrink
	\right]
	= \const_\lyapunovExponent \const_\epsilon \left| \gshrink \lshrink_j \right|^{- \lyapunovExponent},
\end{equation}
where the latter equality with
$ \const_\lyapunovExponent
	= \Gamma \left(\frac{1 - \lyapunovExponent}{2}\right) \big/ \, 2^{\lyapunovExponent / 2} \sqrt{\pi}$
derives from the formula for negative moments of Gaussians \citep{winkelbauer2012gaussian_moments}.

Turning to the case $\min_j \gshrink \lshrink_j \geq \epsilon$, we have
\begin{equation}
\label{eq:reg_coef_cond_density_bound_for_large_prior_sd}
\begin{aligned}
\pi(\bbeta \given \gshrink, \blshrink, \y, \X)
	&= \frac{
		L(\y \given \X, \bbeta)
		\prod_j \exp\left( - \frac{\beta_j^2}{2 \gshrink^2 \lshrink_j^2} \right)
	}{
		{\displaystyle \int} \, L(\y \given \X, \bbeta')
		\prod_j \exp\left( - \frac{\beta_j'^2}{2 \gshrink^2 \lshrink_j^2} \right)\, \diff \bbeta'
	} \\
	&\leq \frac{
		\| L \|_\infty
	}{
		\int \, L(\y \given \X, \bbeta')
		\prod_j \exp\left( - \beta_j'^2 / 2 \epsilon^2 \right)\, \diff \bbeta'
	}
	:= \const_\epsilon'.
\end{aligned}
\end{equation}
Using the above bound on the conditional density, we obtain
\begin{equation}
\label{eq:negative_power_cond_expectation_bound_outer_case}
\begin{aligned}
\expectation \! \left[
	| \betaNext_j |^{-\lyapunovExponent} \given \gshrink, \blshrink, \y, \X
\right]
	&\leq 1 + \expectation \! \left[
		| \betaNext_j |^{-\lyapunovExponent} \ind\big\{ |\beta_j| \leq 1 \big\}
		\given \gshrink, \blshrink, \y, \X
	\right] \\
	&\leq 1 + \const_\epsilon' \textstyle \int_{-1}^1 | \beta_j |^{-\lyapunovExponent} \diff \beta_j \\
	&= 1 + 2 \const_\epsilon' / (1 - \lyapunovExponent).
\end{aligned}
\end{equation}
The bounds \eqref{eq:negative_power_cond_expectation_bound_inner_case} and \eqref{eq:negative_power_cond_expectation_bound_outer_case} together show that an inequality of the form \eqref{eq:negative_power_expectation_bound} holds for any value of $\gshrink$ and $\blshrink$, whether in $\{ \max_j \gshrink \lshrink_j \leq \epsilon \}$ or $\{ \min_j \gshrink \lshrink_j \geq \epsilon \}$.
\end{proof}

\begin{lemma}
\label{lem:sqnorm_expectation_bound_for_probit}
Suppose the likelihood satisfies the assumptions as in Lemma~\ref{lem:negative_power_expectation_bound}.
Then the conditional expectation of $\beta_j^2$ under $\bbeta \given \gshrink, \blshrink, \y, \X, \z = \bm{0}$ is bounded by a constant which depends only on $\zeta$ and functionals of the likelihood $\bbeta \to L(\y \given \X, \bbeta)$.
\end{lemma}

\begin{proof}
We will derive the following bound on the conditional density
\begin{equation}
\label{eq:reg_coef_regularized_conditional_density_bound}
\begin{aligned}
&\pi(\bbeta \given \gshrink, \blshrink, \y, \X, \z = \bm{0}) \\
&\hspace*{2em}\leq 
	\widetilde{\const} \, 
	\normal(\bbeta; \, \bm{0}, \zeta^2 \I) \,
	\left( 1 + \normal\!\left(\bbeta; \, \bm{0}, \gshrink^{2} \bLshrink^{2} \right) \right) \\
&\hspace*{2em}= 
	\widetilde{\const} \, 
	\normal(\bbeta; \, \bm{0}, \zeta^2 \I)  
	+ \widetilde{\const}
		\left( \gshrink^2 \lshrink_j^2 + \zeta^2 \right)^{-1/2} 
		\normal\!\left(
			\bbeta; \, \bm{0}, \left(\gshrink^{-2} \bLshrink^{-2} + \zeta^{-2} \I \right)^{-1}
		\right),
\end{aligned}
\end{equation}
which will imply the desired bound on the conditional expectation:
\begin{align*}
\expectation\left[
	\beta_j^2 \given \gshrink, \blshrink, \y, \X, \z = \bm{0}
\right] 
	&\leq \widetilde{\const} \zeta^2 
	+ \widetilde{\const} \left( \gshrink^2 \lshrink_j^2 + \zeta^2 \right)^{-1/2} 
	\left( \gshrink^{-2} \lshrink_j^{-2} + \zeta^{-2} \right)^{-1} \\
	&= \widetilde{\const} \zeta^2 
		+ \widetilde{\const} \, \zeta^2 \gshrink^2 \lshrink_j^2 \left( \gshrink^2 \lshrink_j^2 + \zeta^2 \right)^{-3/2} \\
	&\leq \widetilde{\const} \zeta^2 
		+ \widetilde{\const} \, \zeta^2 \left( \gshrink^2 \lshrink_j^2 + \zeta^2 \right)^{-1/2} \\
	&\leq  \widetilde{\const} \zeta^2 +  \widetilde{\const} \zeta.
\end{align*}
To complete the proof, therefore, it remains to establish \eqref{eq:reg_coef_regularized_conditional_density_bound}. 
Our argument here closely follows those we use in deriving the bounds \eqref{eq:reg_coef_cond_density_bound_for_small_prior_sd} and \eqref{eq:reg_coef_cond_density_bound_for_large_prior_sd}  in the proof of Lemma~\ref{lem:negative_power_expectation_bound}.
The conditional distribution of $\bbeta \given \gshrink, \blshrink, \y, \X, \z = \bm{0}$ is given by
\begin{equation}
\label{eq:reg_coef_regularized_conditional_density}
\pi(\bbeta \given \gshrink, \blshrink, \y, \X, \z = \bm{0})
	= \frac{
		L(\y \given \X, \bbeta)  L(\z = \bm{0} \given \bbeta) \, \pi(\bbeta \given \gshrink, \blshrink)
	}{
		\int L(\y \given \X, \bbeta') L(\z = \bm{0} \given \bbeta) \, \pi(\bbeta' \given \gshrink, \blshrink) \, \diff \bbeta'
	}.
\end{equation}
As before, we choose $\epsilon > 0$ to be any value small enough to guarantee the likelihood to be positive on the set $\| \bbeta' \|_\infty = \max_j |\beta_j'| \leq \epsilon$.
We can repeat an argument analogous to the derivation of the bound \eqref{eq:reg_coef_cond_density_bound_for_small_prior_sd} to conclude that, when $\max_j \gshrink \lshrink_j \leq \epsilon$, 
\begin{equation}
\label{eq:reg_coef_regularized_cond_density_bound_for_small_prior_sd}
\pi(\bbeta \given \gshrink, \blshrink, \y, \X, \z = \bm{0})
	\leq \widetilde{\const}_\epsilon \, L(\z = \bm{0} \given \bbeta) \, \pi(\bbeta \given \gshrink, \blshrink)
\end{equation}
for $\widetilde{\const}_\epsilon = \| L(\y \given \X, \bbeta) \|_\infty \big/ \left( \min_{ \| \bbeta' \|_\infty \, \leq \, \epsilon} L(\y \given \X, \bbeta')  L(\z = \bm{0} \given \bbeta) \right)\big( \gaussianCdf(1) - \gaussianCdf(-1) \big)^p$ with the $\| \cdot \|_\infty$ norm taken with respect to $\bbeta$.
For the case $\min_j \gshrink \lshrink_j \geq \epsilon$, we follow the derivation of the bound \eqref{eq:reg_coef_cond_density_bound_for_large_prior_sd} to conclude that
\begin{equation}
\label{eq:reg_coef_regularized_cond_density_bound_for_large_prior_sd}
\begin{aligned}
&\pi(\bbeta \given \gshrink, \blshrink, \y, \X)
	= \widetilde{\const}_\epsilon' \,
		L(\z = \bm{0} \given \bbeta)
	\\
&\hspace{2.5em}\text{where } \, \widetilde{\const}_\epsilon'
	= \frac{
		\| L(\y \given \X, \bbeta) \|_\infty
	}{
		\int \, L(\y \given \X, \bbeta') L(\z = \bm{0} \given \bbeta)
		\prod_j \exp\left( - \beta_j'^2 / 2 \epsilon^2 \right)\, \diff \bbeta'
	}.
\end{aligned}
\end{equation}
Combining \eqref{eq:reg_coef_regularized_cond_density_bound_for_small_prior_sd} and \eqref{eq:reg_coef_regularized_cond_density_bound_for_large_prior_sd} yields the desired bound \eqref{eq:reg_coef_regularized_conditional_density_bound}.
\end{proof}

\section{Proofs for Section~\ref{sec:gibbs_sampler_behavior_near_spike}}
\label{sec:proofs_of_gibbs_sampler_behavior_near_spike}

\subsection{Proof of Proposition~\ref{prop:local_scale_posterior_in_the_limit}}
\label{sec:proofs_of_sampler_behaviors_near_spike_minorization}

The key ingredient in our proof of Proposition~\ref{prop:local_scale_posterior_in_the_limit} is the following general result on the stochastic ordering of tilted densities.
The result allows us to study the behavior of $\pi(\lshrink \given \betaPrev, \gshrink)$ viewed as a product of $f(\lshrink) = \lshrink^{-1} \localPrior(\lshrink)$ and $G(\lshrink) = \exp(- \betaPrev[2] / 2 \gshrink^2 \lshrink^2)$.
\begin{proposition}
\label{prop:stochastic_ordering_of_tilted_density}
Consider probability densities $\pi_G(\lshrink) \propto G(\lshrink) f(\lshrink)$ and $\pi_H(\lshrink)$ $\propto H(\lshrink) f(\lshrink)$ on $\lshrink \in [0, \infty)$ for $f, G, H \geq 0$.
Suppose that $f$ satisfies $\int_u^\infty f(\lshrink) \diff \lshrink < \infty$ for $u > 0$.
Suppose also that $G$ and $H$ are absolutely continuous and increasing, $G \leq H$, and $\lim_{\lshrink \to \infty} G(\lshrink) = \lim_{\lshrink \to \infty} H(\lshrink)$.
Then $\pi_G$ is stochastically dominated by $\pi_H$ i.e.\
\begin{equation}
\label{eq:stochastic_ordering_of_product_density}
\int_a^\infty \pi_G(\lshrink) \diff \lshrink
	\leq \int_a^\infty \pi_H(\lshrink) \diff \lshrink
	\ \text{ for any } \,
	a \in \mathbb{R}.
\end{equation}
\end{proposition}

\begin{proof} 
Multiplying $G$ and $H$ with an appropriate constant if necessary, without loss of generality we can assume $\lim_{\lshrink \to \infty} G(\lshrink) = \lim_{\lshrink \to \infty} H(\lshrink) = 1$ so that $G$ and $H$ can be interpreted as cumulative distribution functions.

We first deal with the case $G(0) = H(0) = 0$; when $\int f(\lshrink) \diff \lshrink = \infty$, this assumption is in fact implied by the integrability of $G(\lshrink) f(\lshrink)$ and $H(\lshrink) f(\lshrink)$.
In this case, we have $G(\lshrink) = \int_0^\lshrink g(u) \diff u$ and $H(\lshrink) = \int_0^\lshrink h(u) \diff u$ for density functions $g,h \geq 0$.
As can be verified using Fubini's theorem for positive functions, we can express $\pi_G$ and $\pi_H$ as
\begin{equation*}
\pi_G(\cdot) = \int f(\, \cdot \given u) g(u) \diff u
\ \text{ and } \
\pi_H(\cdot) = \int f(\, \cdot \given u) h(u) \diff u,
\end{equation*}
where $f(\, \cdot \given u)$ for $u > 0$ denote a probability density
\begin{equation*}
f(\, \cdot \given u)
	= \frac{
		f(\lshrink) \mathds{1} \{ \lshrink > u \}
	}{
		\int_u^\infty f(\lshrink) \diff \lshrink
	}.
\end{equation*}
Again by Fubini's theorem for positive functions, we have
\begin{equation}
\label{eq:cdf_as_expectation_of_increasing_function}
\int_a^\infty \pi_G(\lshrink) \diff \lshrink
	= \int F_a(u) g(u) \diff u
	\ \text{ and  } \
\int_a^\infty \pi_H(\lshrink) \diff \lshrink
	= \int F_a(u) h(u) \diff u
\end{equation}
where
\begin{equation*}
F_a(u)
	= \int_a^\infty f(\lshrink \given u) \diff \lshrink
	= \frac{
		\int_{\max\{a, u\}}^\infty f(\lshrink) \diff \lshrink
	}{
		\int_{u}^\infty f(\lshrink) \diff \lshrink
	}.
\end{equation*}
Note that the integrals in \eqref{eq:cdf_as_expectation_of_increasing_function} can be represented as expectations with respect to distributions $G$ and $H$:
\begin{equation}
\label{eq:cdf_as_expectation_of_increasing_function_explicit}
\int_a^\infty \pi_G(\lshrink) \, \diff \lshrink = \expectation_{U \sim G}\!\left[ F_a(U) \right]
	\ \text{ and  } \
\int_a^\infty \pi_H(\lshrink) \, \diff \lshrink = \expectation_{U \sim H}\!\left[ F_a(U) \right].
\end{equation}
Since $F_a$ is an increasing function and $G$ is stochastically dominated by $H$ by our assumption, the representation \eqref{eq:cdf_as_expectation_of_increasing_function_explicit} implies the desired inequality \eqref{eq:stochastic_ordering_of_product_density}.

Earlier, we made a simplifying assumption $G(0) = H(0) = 0$.
More generally, we have the relation $G(\lshrink) - G(0) = \int_0^\lshrink g(u) \diff u$ and $H(\lshrink) - H(0) = \int_0^\lshrink h(u) \diff u$ for integrable functions $g, h \geq 0$.
Essentially the identical arguments as before show that the identity \eqref{eq:cdf_as_expectation_of_increasing_function_explicit} and hence the conclusion \eqref{eq:stochastic_ordering_of_product_density} still hold in this case.
\end{proof}

\begin{proof}[Proof of Proposition~\ref{prop:local_scale_posterior_in_the_limit}]
Note that
\begin{equation*}
\pi(\lshrink_j \given \betaPrev_j, \gshrink)
	\propto \exp\left( - {c^2} / {\lshrink_j^2} \right) \lshrink_j^{-1} \localPrior(\lshrink_j)
	\ \text{ for } \
	c = c(\betaPrev_j / \gshrink) = \frac{\betaPrev_j}{\sqrt{2} \gshrink}.
\end{equation*}
Applying Proposition~\ref{prop:stochastic_ordering_of_tilted_density} with $f(\lshrink) = \lshrink^{-1} \localPrior(\lshrink)$, we see that
\begin{equation*}
\prob\!\left(\lshrink_j > a \given \betaPrev_j, \gshrink \right)
	\leq \prob\!\left(\lshrink_j > a \given \betaPrev[\prime]_j, \gshrink \right)
\end{equation*}
whenever $|\betaPrev_j / \gshrink | \geq |\betaPrev[\prime]_j / \gshrink |$.

Suppose now that $\int \lshrink^{-1} \localPrior(\lshrink) \, \diff \lshrink = \infty$.
For any $\betaPrev_j / \gshrink$, we have
\begin{equation}
\label{eq:numerator_inequality_for_convergence_to_delta_measure}
\int_a^\infty
	\exp\!\left( - \frac{\betaPrev[2]_j}{2 \gshrink^2 \lshrink_j^2} \right) \lshrink_j^{-1} \localPrior(\lshrink_j) \,
\diff \lshrink_j
	\leq \int_a^\infty \lshrink_j^{-1} \localPrior(\lshrink_j) \, \diff \lshrink_j
	\leq 1 / a.
\end{equation}
On the other hand, by Fatou's lemma,
\begin{equation}
\label{eq:denominator_inequality_for_convergence_to_delta_measure}
\liminf_{| \betaPrev_j / \gshrink | \to 0} \int
	\exp\!\left( - \frac{\betaPrev[2]_j}{2 \gshrink^2 \lshrink^2} \right) \lshrink^{-1} \localPrior(\lshrink) \,
\diff \lshrink
	\geq \int \lshrink^{-1} \localPrior(\lshrink) \, \diff \lshrink
	= \infty.
\end{equation}
From \eqref{eq:numerator_inequality_for_convergence_to_delta_measure} and \eqref{eq:denominator_inequality_for_convergence_to_delta_measure}, we conclude that for any $a > 0$
\begin{equation*}
\prob(\lshrink_j > a \given \betaPrev_j, \gshrink)
	=
	\frac{
		\int_a^\infty
			\exp\!\left( - \frac{\betaPrev[2]_j}{2 \gshrink^2 \lshrink_j^2} \right) \lshrink_j^{-1} \localPrior(\lshrink_j) \,
		\diff \lshrink_j
	}{
		\int
			\exp\!\left( - \frac{\betaPrev[2]_j}{2 \gshrink^2 \lshrink^2} \right) \lshrink^{-1} \localPrior(\lshrink) \,
		\diff \lshrink
	}
	\to 0
	\ \text{ as } \
	| \betaPrev_j / \gshrink | \to 0,
\end{equation*}
i.e.\ $\pi(\lshrink_j \given \betaPrev_j, \gshrink)$ converges in distribution to a delta measure at 0.

We now turn to quantifying the limiting behavior when $\int \lshrink^{-1} \localPrior(\lshrink) \, \diff \lshrink < \infty$.
For any $a \in [0, \infty]$, the dominated convergence theorem yields
\begin{equation*}
\lim_{|\betaPrev_j / \gshrink| \to 0} \int_0^a
	\exp\!\left( - \frac{\betaPrev[2]_j}{2 \gshrink^2 \lshrink_j^2} \right) \lshrink_j^{-1} \localPrior(\lshrink_j) \,
\diff \lshrink_j
	= \int_0^a \lshrink^{-1} \localPrior(\lshrink) \, \diff \lshrink.
\end{equation*}
The above convergence result implies the point-wise convergence of the cumulative distribution function:
\begin{equation*}
\lim_{|\betaPrev_j / \gshrink| \to 0} \prob(\lshrink_j \leq a \given \betaPrev_j, \gshrink)
	=
	\frac{
		\int_0^a \lshrink_j^{-1} \localPrior(\lshrink_j) \, \diff \lshrink_j
	}{
		\int \lshrink^{-1} \localPrior(\lshrink) \, \diff \lshrink
	}. \qedhere
\end{equation*}
\end{proof}

\subsection{Proof of Proposition~\ref{prop:scale_negative_power_expectation_bound}}
\label{sec:proofs_of_sampler_behaviors_near_spike_drift}

\begin{proof}
In upper-bounding $\expectation\!\left[ \lshrink_j^{-\lyapunovExponent} \given \gshrink, \bbetaPrev \right]$, we can without loss of generality assume that $\pi(0) > 0$ by virtue of Proposition~\ref{prop:modifying_density_at_zero} below.
In terms of the constants $\epsilon$ and $\const''(\lyapunovExponent, \localPrior)$ as defined in Lemma~\ref{lem:expectation_of_scale_negative_power} below, let
\begin{equation}
\label{eq:scale_negative_power_expectation_bound_for_small_regcoef}
\gamma(r)
	= \const''(\lyapunovExponent, \localPrior)
		\bigg/ \log\bigg( 1 + \frac{4 \epsilon^2}{r^2} \bigg).
\end{equation}
By Lemma~\ref{lem:expectation_of_scale_negative_power} and the monotonicity of $\gamma(r)$, we then have
\begin{equation*}
\expectation\!\left[
	\gshrink^{-\lyapunovExponent} \lshrink_j^{-\lyapunovExponent} \given \gshrink, \betaPrev_j
\right]
	\leq \gamma(R / \gshrink) \left| \betaPrev_j \right|^{-\lyapunovExponent}
	\ \text{ whenever } \, |\betaPrev_j| \leq R.
\end{equation*}
On the other hand, since the distribution $\lshrink_j \given \gshrink, \betaPrev_j$ stochastically dominates $\lshrink_j \given \gshrink, \betaPrev[\prime]_j$ whenever $\betaPrev_j \geq \betaPrev[\prime]_j$ (Proposition~\ref{prop:local_scale_posterior_in_the_limit}), we have
\begin{equation}
\label{eq:scale_negative_power_expectation_bound_for_nonsmall_regcoef}
\expectation\!\left[
	\gshrink^{-\lyapunovExponent} \lshrink_j^{-\lyapunovExponent} \given \gshrink, \betaPrev_j
\right]
	\leq \expectation\!\left[
			\gshrink^{-\lyapunovExponent} \lshrink_j^{-\lyapunovExponent} \given \gshrink, |\betaPrev[\prime]_j| = R
		\right]
	\ \text{ whenever } \,
	|\betaPrev_j| \geq R.
\end{equation}
Combining \eqref{eq:scale_negative_power_expectation_bound_for_small_regcoef} and \eqref{eq:scale_negative_power_expectation_bound_for_nonsmall_regcoef} yields the inequality \eqref{eq:scale_negative_power_expectation_bound}.
\end{proof}

\begin{proposition}
\label{prop:modifying_density_at_zero}
Given a prior $\localPrior(\cdot)$ such that $\localPrior(0) = 0$ and $\| \localPrior \|_\infty < \infty$, there is a density $\localPrior'(\cdot)$ such that $\localPrior'(\lshrink)$ is continuous at $\lshrink = 0$, $\localPrior'(0) > 0$, $\| \localPrior' \|_\infty < \infty$, and $\localPrior(\lshrink) \propto G(\lshrink) \localPrior'(\lshrink)$ for a bounded increasing function $G \geq 0$.
Consequently, a density $\pi(\cdot)$ stochastically dominates $\pi'(\cdot)$ when $\pi(\lshrink) \propto f(\lshrink) \localPrior(\lshrink)$ and $\pi'(\lshrink) \propto f(\lshrink) \localPrior'(\lshrink)$ for $f \geq 0$.
By taking $f(\lshrink) = \lshrink^{-1} \exp(- \betaPrev[2]_j / 2 \gshrink^2 \lshrink_j^2)$ in particular, we have the following inequality between the expectations with respect to $\pi(\cdot)$ and $\pi'(\cdot)$:
\begin{equation}
\label{eq:bound_on_expectation_by_modified_density}
\expectation\!\left[
	\lshrink_j^{-\lyapunovExponent} \given \gshrink, \betaPrev_j
\right] \leq
	\expectation'\!\left[
		\lshrink_j^{-\lyapunovExponent} \given \gshrink, \betaPrev_j
	\right]
	\ \text{ for } \lyapunovExponent \geq 0.
\end{equation}
\end{proposition}

\begin{proof}
Redefining $\localPrior(\lshrink)$ as $\localPrior(\lshrink - \lshrink_{\min})$ for $\lshrink_{\min} = \inf\left\{ \lshrink: \localPrior(\lshrink) > 0 \right\}$ if necessary, we can without loss of generality assume that $\localPrior(\lshrink) > 0$ for all sufficiently small $\lshrink > 0$.
Define
\begin{equation}
\label{eq:modifier_function_definition}
G(\lshrink) = \min\!\left\{
	\| \localPrior \|_\infty,
	\int_0^\lshrink \max\!\left\{0, \frac{\diff \localPrior}{\diff \lshrink}(u) \right\} \diff u
\right\}.
\end{equation}
Then $G$ is clearly increasing and bounded.
The definition \eqref{eq:modifier_function_definition} further guarantees that $\lim_{\lshrink \to 0} \localPrior(\lshrink) / G(\lshrink) = 1$, $\localPrior \leq G$, and $\lim_{\lshrink \to \infty} G(\lshrink) = \| \localPrior \|_\infty$.
Define $\localPrior'(\cdot)$ via the relation $\localPrior'(\lshrink) \propto \localPrior(\lshrink) / G(\lshrink)$ for $\lshrink > 0$ and $\localPrior'(0) := \lim_{\lshrink \to 0} \localPrior'(\lshrink)$.
Then $\localPrior'(\cdot)$ satisfy $\| \localPrior' \|_\infty = \localPrior'(0) = \left( \int \pi(\lshrink) / G(\lshrink) \, \diff \lshrink \right)^{-1} > 0$, as well as all the other desired properties.

When $\pi(\lshrink) \propto f(\lshrink) \localPrior(\lshrink)$ and $\pi'(\lshrink) \propto f(\lshrink) \localPrior'(\lshrink)$, the densities satisfies the relation $\pi'(\lshrink) \propto G(\lshrink) \pi(\lshrink)$.
By applying Proposition~\ref{prop:stochastic_ordering_of_tilted_density} with $H = \| G \|_\infty$, we conclude that $\pi(\cdot)$ stochastically dominates $\pi'(\cdot)$.
The inequality \eqref{eq:bound_on_expectation_by_modified_density} is an immediate consequence of this stochastic ordering.
\end{proof}

\begin{lemma}
\label{lem:expectation_of_scale_negative_power}
Suppose that $\localPrior(\lshrink)$ is continuous at $\lshrink = 0$ and $\localPrior(0) > 0$.
For $\lyapunovExponent \in [0, 1)$ and $\epsilon > 0$ small enough that $\min_{\lshrink \in [0, \epsilon]} \localPrior(\lshrink) \geq  \localPrior(0) / 2$, we have the following inequality:
\begin{equation*}
\expectation\!\left[
	\gshrink^{-\lyapunovExponent} \lshrink_j^{-\lyapunovExponent} \given \gshrink, \bbetaPrev
\right]
	\leq \const''(\lyapunovExponent, \localPrior) \, | \betaPrev_j |^{-\lyapunovExponent}
		\bigg/ \log\bigg( 1 + \frac{4 \gshrink^2 \epsilon^2}{|\betaPrev_j|^2} \bigg),
\end{equation*}
where $\const''(\lyapunovExponent, \localPrior) > 0$ is a constant depending only on $\lyapunovExponent$ and $\localPrior(\cdot)$ given by
\begin{equation*}
\const''(\lyapunovExponent, \localPrior)
	= 2^{2 + \lyapunovExponent / 2}
	\frac{ \| \localPrior \|_\infty }{ \localPrior(0) }
	\int_0^\infty
		\frac{1}{ \lshrink^{1 + \lyapunovExponent} } \exp\!\left( - \frac{1}{\lshrink^2} \right) \diff \lshrink.
\end{equation*}
\end{lemma}

\begin{proof}
Observe that
\begin{equation}
\label{eq:last_expectation_as_fraction}
\begin{aligned}
&\expectation\!\left[ \lshrink_j^{-\lyapunovExponent} \, \Big| \, \gshrink, \bbetaPrev \right] \\
	&\hspace{1.5em} = \int_0^\infty  \frac{1}{\lshrink^{1 + \lyapunovExponent}} \exp\!\left( - \frac{c_j^2}{\lshrink^2} \right) \localPrior(\lshrink) \diff \lshrink \bigg/ \!\! \int_0^\infty \frac{1}{\lshrink} \exp\!\left( - \frac{c_j^2}{\lshrink^2} \right) \localPrior(\lshrink) \diff \lshrink,
\end{aligned}
\end{equation}
where $c_j = c(\gshrink, \beta_j) = |\beta_j| / \sqrt{2} \gshrink$.
With the change of variable $\lshrink \to \lshrink / c_j$, we can write the right-hand side of \eqref{eq:last_expectation_as_fraction} as
\begin{equation}
\label{eq:fraction_of_integrals}
\frac{1}{c_j^\lyapunovExponent} \int_0^\infty \frac{1}{\lshrink^{1 + \lyapunovExponent}} \exp\!\left( - \frac{1}{\lshrink^2} \right) \localPrior(c_j \lshrink) \, \diff \lshrink
		\bigg/ \!\!
		\int_0^\infty \frac{1}{\lshrink} \exp\!\left( - \frac{1}{\lshrink^2} \right) \localPrior(c_j \lshrink) \diff \lshrink.
\end{equation}
We can upper bound the numerator as
\begin{equation}
\label{eq:numerator_upper_bd}
\frac{1}{c_j^\lyapunovExponent} \int_0^\infty \frac{1}{\lshrink^{1 + \lyapunovExponent}} \exp\!\left( - \frac{1}{\lshrink^2} \right) \localPrior(c_j \lshrink) \, \diff \lshrink
	\leq \frac{1}{c_j^\lyapunovExponent} \| \localPrior \|_\infty \int_0^\infty \frac{1}{\lshrink^{1 + \lyapunovExponent}} \exp\!\left( - \frac{1}{\lshrink^2} \right) \diff \lshrink.
\end{equation}
To lower bound the denominator, we restrict the range of integration to $[0, \epsilon / c_j]$ for $\epsilon > 0$ and apply the change of variable $\phi = \lshrink^{-2}$:
\begin{align*}
\int_0^\infty \frac{1}{\lshrink} \exp\!\left( - \frac{1}{\lshrink^2} \right) \localPrior(c_j \lshrink) \diff \lshrink
	&\geq \left( \min_{[0, \epsilon]} \localPrior \right) \int_{0}^{\epsilon / c_j} \frac{1}{\lshrink} \exp\!\left( - \frac{1}{\lshrink^2} \right) \diff \lshrink \\
	&= \left( \min_{[0, \epsilon]} \localPrior \right) \int_{c_j^2 / \epsilon^2}^\infty \phi^{-1} \exp\!\left( - \phi \right) \diff \phi.
\end{align*}
The inequality of \cite{gautschi1959incomplete_gamma_inequalities} tells us that $\int_{a}^\infty \phi^{-1} \exp(-\phi) \diff \phi \geq \log(1 + 2 a ^{-1}) / 2$, so we obtain
\begin{equation}
\label{eq:denominator_lower_bd}
\int_0^\infty \frac{1}{\lshrink} \exp\!\left( - \frac{1}{\lshrink^2} \right) \localPrior(c_j \lshrink) \diff \lshrink
	\geq \left( \min_{[0, \epsilon]} \localPrior \right) \frac12 \log\!\left( 1 + 2 \frac{\epsilon^2}{c_j^2} \right).
\end{equation}
From the upper bound \eqref{eq:numerator_upper_bd} of the numerator and lower bound \eqref{eq:denominator_lower_bd} of the denominator, it follows that the ratio \eqref{eq:fraction_of_integrals} is upper bounded by
\begin{equation*}
c_j^{-\lyapunovExponent}
		\frac{ 2 \| \localPrior \|_\infty }{
			\left( \min_{[0, \epsilon]} \localPrior \right) \log\!\left(1 + 2 \epsilon^2 c_j^{-2} \right)
		}
		\int_0^\infty
			\frac{1}{ \lshrink^{1 + \lyapunovExponent} } \exp\!\left( - \frac{1}{\lshrink^2} \right) \diff \lshrink.
\end{equation*}
Substituting $c_j = |\beta_j| / \sqrt{2} \gshrink$ into the above expression completes the proof.
\end{proof}

\section{Proof of Lemma~\ref{lem:minorization}}
\label{sec:proof_of_minorization_lemma}

Our proof of Lemma~\ref{lem:minorization} builds on the known fact below.
\begin{proposition}[\citealp{choi2013ergodicity-bayes-logit}]
\label{prop:choi_hobert_minorization}
For fixed $\gshrink$ and $\blshrink$, the marginal transition kernel satisfies the minorization condition
\begin{equation*}
\kernel(\bbeta \given \bbetaPrev, \gshrink, \blshrink)
	\geq \delta_{\gshrink \blshrink} \, \normal(\bbeta; \bmu_{\gshrink \blshrink}, \bPhi_{\gshrink \blshrink}^{-1})
\end{equation*}
where $\bPhi_{\gshrink \blshrink} = \frac12 \X^\transpose \X + \zeta^{-2} \I + \gshrink^{-2} \bLshrink^{-2}$, $\bmu_{\gshrink \blshrink} = \bPhi_{\gshrink \blshrink}^{-1} \X^\transpose (\y - \bm{1} / 2)$, and
\begin{equation}
\label{eq:choi_hobert_const}
\delta_{\gshrink \blshrink}
	=  C_n
		\frac{
			| \zeta^{-2} \I + \gshrink^{-2} \bLshrink^{-2} |^{1/2}
		}{
			 | \bPhi_{\gshrink \blshrink} |^{1/2}
		}
		\exp\!\left\{
				\frac12 \bm{w}^\transpose \left[
				\bPhi_{\gshrink \blshrink}^{-1} - \left(\zeta^{-2} \I + \gshrink^{-2} \bLshrink^{-2}\right)^{-1}
				\right]
			\bm{w} \right\}
\end{equation}
for $\bm{w} = \X^\transpose (\y - \bm{1} / 2)$ and $C_n > 0$ depending only on $n$.
\end{proposition}

Proposition~\ref{prop:lower_bound_on_choi_hobert_const} and \ref{prop:minorization_const_of_two_gaussians} below are the main workhorses for our proof of Lemma~\ref{lem:minorization} along with Proposition~\ref{prop:choi_hobert_minorization}.
We first state the results and use them to prove Lemma~\ref{lem:minorization}, before proceeding to prove the results themselves.
\begin{proposition}
\label{prop:lower_bound_on_choi_hobert_const}
As a function of $\gshrink \blshrink$, the minorization constant \eqref{eq:choi_hobert_const} is uniformly bounded below by a positive constant on the set $\min_j \gshrink \lshrink_j \geq R > 0$.
\end{proposition}

\newcommand{\smPrec}{\bPhi}
\newcommand{\smMean}{\bmu}
\newcommand{\lgPrec}{\bPhi'}
\newcommand{\lgMean}{\bmu'}
\begin{proposition}
\label{prop:minorization_const_of_two_gaussians}
If two precision matrices $\smPrec$ and $\lgPrec$ satisfy $\smPrec \prec \lgPrec$, then a minorization $\normal(\bbeta; \smMean, \smPrec^{-1}) \geq \delta \, \normal(\bbeta; \lgMean, \lgPrec^{-1})$ holds for $\delta > 0$ given by
\begin{equation}
\label{eq:minorization_const_of_two_gaussians}
\begin{aligned}
\delta
	&= \inf_{\bbeta} \, \frac{
			\normal(\bbeta; \smMean, \smPrec^{-1})
		}{
			\normal(\bbeta; \lgMean, \lgPrec^{-1})
		} \\
	&= \frac{|\smPrec|^{1/2}}{|\lgPrec|^{1/2}} 
		\exp\!\left\{
			- \frac12 (\lgMean - \smMean)^\transpose \smPrec
					\left[ ( \lgPrec - \smPrec )^{-1} (\lgPrec \lgMean - \smPrec \smMean) - \smMean \right]
		\right\}.
		\end{aligned}
\end{equation}
When the means take the form $\smMean = \smPrec^{-1} \bm{w}$ and $\lgMean = \lgPrec^{-1} \bm{w}$,  \eqref{eq:minorization_const_of_two_gaussians} simplifies to
\begin{equation*}
\delta
	= \frac{|\smPrec|^{1/2}}{|\lgPrec|^{1/2}} 
		\exp\!\left\{
			\frac{1}{2} \bm{w}^\transpose (\smPrec^{-1} - \lgPrec^{-1}) \bm{w}
		\right\}
	\geq \frac{|\smPrec|^{1/2}}{|\lgPrec|^{1/2}}.
\end{equation*}
\end{proposition}

\begin{proof}[Proof of Lemma~\ref{lem:minorization}]
On the set $\{\blshrink : \min_j \gshrink \lshrink_j \geq R\}$, Proposition~\ref{prop:choi_hobert_minorization} implies that
\begin{equation*}
\kernel(\bbeta \given \bbetaPrev, \gshrink, \blshrink)
	\geq
		\left( \min_{\gshrink \lshrink_j \geq R} \delta_{\gshrink \blshrink} \right)
		\normal(\bbeta; \bmu_{\gshrink \blshrink}, \bPhi_{\gshrink \blshrink}^{-1}),
\end{equation*}
where $\min_{\gshrink \lshrink_j \geq R} \delta_{\gshrink \blshrink}$ is guaranteed to be strictly positive by Proposition~\ref{prop:lower_bound_on_choi_hobert_const}.

We complete the proof by showing that the following inequality holds whenever $\min_j \gshrink \lshrink_j \geq R$:
\begin{equation}
\label{eq:comparison_of_two_gaussian_posteriors}
\normal(\bbeta; \bmu_{\gshrink \blshrink}, \bPhi_{\gshrink \blshrink}^{-1})
	\geq \frac{
		| \bPhi_\infty |^{1/2}
	}{
		| \bPhi_R |^{1/2}
	} \normal(\bbeta; \bmu_R, \bPhi_R^{-1}).
\end{equation}
When $\min_j \gshrink \lshrink_j > R$, we have $R^{-2} - \gshrink^{-2} \lshrink_j^{-2} > 0$ and hence
\begin{equation*}
\bPhi_R - \bPhi_{\gshrink \blshrink} = (R^{-2} \I - \gshrink^{-2} \bLshrink^{-2}) \succ 0.
\end{equation*}
By Proposition~\ref{prop:minorization_const_of_two_gaussians}, it follows that
\begin{equation}
\label{eq:comparison_of_two_gaussian_posteriors_before_taking_limit}
\normal(\bbeta; \bmu_{\gshrink \blshrink}, \bPhi_{\gshrink \blshrink}^{-1})
	\geq \frac{
		| \bPhi_{\gshrink \blshrink} |^{1/2}
	}{
		| \bPhi_R |^{1/2}
	} \, \normal(\bbeta; \bmu_R, \bPhi_R^{-1}).
\end{equation}
The above inequality in fact holds not only on the set $\{\blshrink : \gshrink \lshrink_j > R\}$ but also on the closure $\{\blshrink : \min_j \gshrink \lshrink_j \geq R\}$ since all the quantities depend continuously on $\gshrink \lshrink_j$.
The inequality \eqref{eq:comparison_of_two_gaussian_posteriors} follows from \eqref{eq:comparison_of_two_gaussian_posteriors_before_taking_limit} by observing that $\bPhi_{\gshrink \blshrink} \succ \bPhi_\infty$ and hence $|\bPhi_{\gshrink \blshrink}| \geq |\bPhi_\infty|$.
\end{proof}

\subsubsection{Proof of Proposition~\ref{prop:lower_bound_on_choi_hobert_const} and \ref{prop:minorization_const_of_two_gaussians}}
\label{sec:proof_of_intermediate_results_for_minorization_lemma}
In the proofs to follow, we will make use of the following elementary linear algebra facts about positive definite matrices.
We will denote the largest, $i$th largest, and smallest eigenvalue of a matrix $\bA$ as $\eigen_{\max}(\bA)$, $\eigen_i(\bA)$, and $\eigen_{\min}(\bA)$.
The determinant of $\bA$ is denoted by $|\bA|$ and the trace by $\textrm{tr}(\bA)$.
The notation $\bA \prec \bB$ means that $\bB - \bA$ is positive definite or, equivalently, $\bv^\transpose \bA \bv < \bv^\transpose \bB \bv$ for any vector $\bv \neq \bm{0}$.
\begin{proposition}
\label{prop:linalg_facts}
Given positive definite matrices $\bA$ and $\bB$, we have
\begin{enumerate}
\item $(\bA + \bB)^{-1} \prec \bA^{-1}$.
\item $(\bA + \bB)^{-1} \succ \bA^{-1} - \bA^{-1} \bB \bA^{-1}$
\item $\eigen_i(\bA) + \eigen_{\min}(\bB) \leq \eigen_i(\bA + \bB) \leq \eigen_i(\bA) + \eigen_{\max}(\bB)$ for all $i$.
\item $|\bA| < |\bA + \bB|$.
\item $|\bA + \bB| \leq |\bA| \exp\!\left\{ \eigen_{\max}(\bB) \, \textrm{tr}(\bA^{-1}) \right\}$.
\end{enumerate}
When $\bA \prec \bm{C}$ for another positive definite matrix $\bm{C}$, we can apply above results with $\bB = \bm{C} - \bA \succ 0$ to obtain analogous inequalities.
\end{proposition}

\begin{proof} 
The eigenvalues of $\I + \bB$ are given by $1 + \eigen_i(\bB)$ and those of $(\I + \bB)^{-1}$ by $1 / (1 + \eigen_i(\bB)) < 1$, so we have $(\I + \bB)^{-1} \prec \I$.
This result holds when $\bB$ is replaced by $\bA^{-1/2} \bB \bA^{-1/2}$ and thus implies that
\begin{align*}
\bv^\transpose (\bA + \bB)^{-1} \bv
	&= \bv^\transpose \bA^{-1/2} \left( \I + \bA^{-1/2} \bB \bA^{-1/2} \right)^{-1} \bA^{-1/2} \bv^\transpose \\
	&< \bv^\transpose \bA^{-1/2} \bA^{-1/2} \bv^\transpose
\end{align*}
for $\bv \neq \bm{0}$.
Hence we have $(\bA + \bB)^{-1}  < \bA^{-1}$.

To prove Property~2, we first show $(\I + \bB)^{-1} \succ \I - \bB$.
By applying a change of basis if necessary, we can assume that $\bB$ is diagonal.
Since $(1 + B_{ii})^{-1} > 1 - B_{ii}$, we have
\begin{equation*}
\bv^\transpose (\I + \bB)^{-1} \bv
	= \sum_i (1 + B_{ii})^{-1} v_i^2
	> \sum_i (1 - B_{ii}) v_i^2
	= \bv^\transpose (\I - \bB) \bv.
\end{equation*}
Since the result $(\I + \bB)^{-1} \succ \I - \bB$  holds when $\bB$ is replaced by $\bA^{-1/2} \bB \bA^{-1/2}$, we obtain
\begin{equation*}
\begin{aligned}
(\bA + \bB)^{-1}
	&= \bA^{-1/2} \left( \I + \bA^{-1/2} \bB \bA^{-1/2} \right)^{-1} \bA^{-1/2} \\
	&\succ \bA^{-1/2} \left( \I - \bA^{-1/2} \bB \bA^{-1/2} \right) \bA^{-1/2} \\
	&= \bA^{-1} - \bA^{-1} \bB \bA^{-1}.
\end{aligned}
\end{equation*}

Property~3 is Theorem~8.1.5 of \cite{golub2012matrix} and immediately implies Property~4.

For Property~5, observe that
\begin{equation*}
|\bA + \bB|
	= \prod_i \eigen_i(\bA + \bB)
	\leq \prod_i \left\{ \eigen_i(\bA) + \eigen_{\max}(\bB) \right\}.
\end{equation*}
Taking the logarithm and applying the inequality $\log(1 + x) \leq x$, we have
\begin{align*}
\log |\bA + \bB| - \log |\bA|
	&\leq \sum_i \log\left( 1 + \frac{ \eigen_{\max}(\bB) }{ \eigen_i(\bA) } \right) \\
	&\leq \sum_i \frac{ \eigen_{\max}(\bB) }{ \eigen_i(\bA) } \\
	&= \eigen_{\max}(\bB) \, \textrm{tr}(\bA^{-1}). \qedhere
\end{align*}
\end{proof}

\begin{proof}[Proof of Proposition~\ref{prop:lower_bound_on_choi_hobert_const}]
Throughout the proof, we use the notation $\bPhi_\infty = \frac12 \X^\transpose \X + \zeta^{-2} \I$ so that $\bPhi_{\gshrink \blshrink} = \bPhi_\infty + \gshrink^{-2} \bLshrink^{-2}$.
By Proposition~\ref{prop:linalg_facts}, we have
\begin{equation*}
\begin{aligned}
\left| \zeta^{-2} \I + \gshrink^{-2} \bLshrink^{-2} \right|
	&\geq | \zeta^{-2} \I | \\
| \bPhi_\infty + \gshrink^{-2} \bLshrink^{-2} |
	&\leq \left| \bPhi_\infty \right|
		\exp\big\{ \!
			\left( \textstyle \max_j \gshrink^{-2} \lshrink_j^{-2} \right)
			\textrm{tr}\!\left( \bPhi_\infty^{-1} \right)
		\! \big\}.
\end{aligned}
\end{equation*}
The above inequalities imply that
\begin{equation}
\label{eq:choi_hobert_proof_inequality_1}
\frac{
	| \zeta^{-2} \I + \gshrink^{-2} \bLshrink^{-2} |^{1/2}
}{
	| \bPhi |^{1/2}
}
	\geq \frac{
			| \zeta^{-2} \I |
		}{
			\left| \bPhi_\infty \right|
		} \exp \left\{
			- \frac{1}{ \textstyle \min_j \gshrink^{2} \lshrink_j^{2} } \,
			\textrm{tr}\!\left( \bPhi_\infty^{-1} \right)
		\right\}.
\end{equation}
Also by Proposition~\ref{prop:linalg_facts}, we have
\begin{equation*}
\begin{aligned}
\left( \zeta^{-2} \I + \gshrink^{-2} \bLshrink^{-2} \right)^{-1}
	&\prec \zeta^2 \I  \\
\left( \bPhi_\infty + \gshrink^{-2} \bLshrink^{-2} \right)^{-1}
	&\succ \bPhi_\infty^{-1} - \bPhi_\infty^{-1} \gshrink^{-2} \bLshrink^{-2} \bPhi_\infty^{-1}.
\end{aligned}
\end{equation*}
We therefore have
\begin{equation}
\label{eq:choi_hobert_proof_inequality_2}
\begin{aligned}
&\bm{w}^\transpose \left[
				\bPhi_{\gshrink \blshrink}^{-1} - \left(\zeta^{-2} \I + \gshrink^{-2} \bLshrink^{-2}\right)^{-1}
				\right]
			\bm{w} \\
	&\hspace{4em} \geq \bm{w}^\transpose\bPhi_\infty^{-1} \bm{w}
		- \bm{w}^\transpose \bPhi_\infty^{-1} \gshrink^{-2} \bLshrink^{-2} \bPhi_\infty^{-1} \bm{w}
		- \zeta^{-2} \| \bm{w} \|^2 \\
	&\hspace{4em} \geq \bm{w}^\transpose\bPhi_\infty^{-1} \bm{w}
			- \frac{1}{\textstyle \min_j \gshrink^{2} \lshrink_j^2} \| \bPhi_\infty^{-1} \bm{w} \|^2
			- \zeta^{-2} \| \bm{w} \|^2.
\end{aligned}
\end{equation}
From \eqref{eq:choi_hobert_proof_inequality_1} and \eqref{eq:choi_hobert_proof_inequality_2}, we see that for all $\min_j \gshrink \lshrink_j \geq R$
\begin{equation*}
\delta_{\gshrink \blshrink}
	\geq C_n \frac{
				| \zeta^{-2} \I |^{1/2}
			}{
				\left| \bPhi_\infty \right|^{1/2}
			} \exp \left\{
				\bm{w}^\transpose\bPhi_\infty^{-1} \bm{w} - \zeta^{-2} \| \bm{w} \|^2
				- \frac{
					\textrm{tr}\!\left( \bPhi_\infty^{-1} \right) + \| \bPhi_\infty^{-1} \bm{w} \|^2
				}{
					R^{2}
				}
			\right\}. \qedhere
\end{equation*}
\end{proof}

\begin{proof}[Proof of Proposition~\ref{prop:minorization_const_of_two_gaussians}]
Note that
\begin{equation*}
\inf_{\bbeta} \, \frac{
	\normal(\bbeta; \smMean, \smPrec^{-1})
}{
	\normal(\bbeta; \lgMean, \lgPrec^{-1})
} = \frac{|\smPrec|^{1/2}}{|\lgPrec|^{1/2}} \exp\!\left\{ \frac12 \inf_{\bbeta} \Delta(\bbeta) \right\},
\end{equation*}
where
\begin{equation*}
\Delta(\bbeta)
	= (\bbeta - \lgMean)^\transpose \lgPrec (\bbeta - \lgMean) - (\bbeta - \smMean)^\transpose \smPrec (\bbeta - \smMean).
\end{equation*}
The quadratic function $\Delta(\bbeta)$ has a unique global minimum since the Hessian $\partial_{\bbeta}^2 \Delta = \lgPrec - \smPrec$ is positive definite by our assumption.
Differentiating $\Delta(\bbeta)$, we see that the minimum occurs at $\bhatbeta$ such that
\begin{equation*}
\lgPrec (\bhatbeta - \lgMean) - \smPrec (\bhatbeta - \smMean) = 0,
	\ \text{ or equivalently } \
	\bhatbeta = \left( \lgPrec - \smPrec \right)^{-1} \left( \lgPrec \lgMean - \smPrec \smMean \right).
\end{equation*}
The minimum $\widehat{\Delta} = {\Delta}(\bhatbeta)$ can be expressed as
\begin{equation*}
\begin{aligned}
\widehat{\Delta}
	&= - (\lgMean - \smMean)^\transpose \smPrec (\bhatbeta - \smMean) \\
	&= - (\lgMean - \smMean)^\transpose \smPrec
		\left[ ( \lgPrec - \smPrec )^{-1} (\lgPrec \lgMean - \smPrec \smMean) - \smMean \right].
\end{aligned}
\end{equation*}
In the special case $\smMean = \smPrec^{-1} \bm{w}$ and $\lgMean = \lgPrec^{-1} \bm{w}$, we have
\begin{equation*}
\widehat{\Delta}
	= - (\lgMean - \smMean)^\transpose \smPrec \smMean
	= - \left( \lgPrec^{-1} \bm{w} - \smPrec^{-1} \bm{w} \right)^\transpose \bm{w}
	= \bm{w}^\transpose \left( \smPrec^{-1} - \lgPrec^{-1} \right) \bm{w}
	\geq 0,
\end{equation*}
where the last inequality follows from $\smPrec^{-1} \succ \lgPrec^{-1}$.
\end{proof}

\section{Proof of Proposition~\ref{prop:negative_moment_bound} and \ref{prop:marginal_variance_bound}}
\label{sec:proof_of_neg_moment_and_marginal_var_bound}

\begin{proof}[Proof of Proposition~\ref{prop:negative_moment_bound}]
\cite{winkelbauer2012gaussian_moments} tells us that a negative moment of Gaussian is given by
\begin{equation*}
\expectation | \beta |^{-\lyapunovExponent}
	= \frac{
			\Gamma \left(\frac{1 - \lyapunovExponent}{2}\right)
		}{
			2^{\lyapunovExponent / 2} \sqrt{\pi}
		} \,
		\sigma^{- \lyapunovExponent}
		\kummerConfluentHyperGeom\!\left(
			\frac{\lyapunovExponent}{2}, \frac12, - \frac{\mu^2}{2 \sigma^2}
		\right),
\end{equation*}
where $\kummerConfluentHyperGeom(\cdot, \cdot, \cdot)$ is Kummer's confluent hypergeometric function (see Proposition~\ref{prop:kummer_function_properties}).
To complete the proof, therefore, it suffices to show that
$\kummerConfluentHyperGeom\!\left(
	\frac{\lyapunovExponent}{2}, \frac12, - \frac{\mu^2}{2 \sigma^2}
\right)$
is bounded by the smaller of 1 and the function $D(\mu / \sigma)$ as given in \eqref{eq:bound_on_kummers_confluent_hypergeom}.

Since $\lyapunovExponent / 2 < 1 / 2$, Proposition~\ref{prop:kummer_function_properties} tells us that
$\kummerConfluentHyperGeom\!\left(
	\frac{\lyapunovExponent}{2}, \frac12, - \frac{\mu^2}{2 \sigma^2}
\right)$
is bounded by 1 and admits the integral representation
\begin{equation}
\label{eq:kummer_integral_representation}
\kummerConfluentHyperGeom\!\left(
	\frac{\lyapunovExponent}{2}, \frac12, - \frac{\mu^2}{2 \sigma^2}
\right)
	= \frac{1}{\betaFun\!\left(\frac{\lyapunovExponent}{2}, \frac{1 - \lyapunovExponent}{2} \right)}
		\int_0^1 (1 - u)^{\frac{1 - \lyapunovExponent}{2} - 1} u^{ \frac{\lyapunovExponent}{2} - 1} \exp\!\left( - \frac{\mu^2}{2 \sigma^2} u \right) \diff u.
\end{equation}
To bound the integral, we break up the domain of integration into $[0, 1/2]$ and $[1/2, 1]$ and observe that
\begin{align*}
\int_{1/2}^1 (1 - u)^{\frac{1 - \lyapunovExponent}{2} - 1} u^{ \frac{\lyapunovExponent}{2} - 1} \exp\!\left( - \frac{\mu^2}{2 \sigma^2} u \right) \diff u
	&\leq 2^{1 - \frac{\lyapunovExponent}{2} }
		\exp\!\left( - \frac{\mu^2}{4 \sigma^2} \right)
		\int_{1/2}^1 (1 - u)^{\frac{1 - \lyapunovExponent}{2} - 1} \diff u \\
	&= \frac{ 2^{\frac{5}{2} - \lyapunovExponent} }{ 1 - \lyapunovExponent }
		\exp\!\left( - \frac{\mu^2}{4 \sigma^2} \right),
		\yesnumber \label{eq:kummer_integral_right_half_bound}
\end{align*}
and that
\begin{align*}
\int_0^{1/2} (1 - u)^{\frac{1 - \lyapunovExponent}{2} - 1} u^{ \frac{\lyapunovExponent}{2} - 1} \exp\!\left( - \frac{\mu^2}{2 \sigma^2} u \right) \diff u
	&\leq 2^{1 - \frac{1 - \lyapunovExponent}{2}}
		\int_0^{1/2} u^{ \frac{\lyapunovExponent}{2} - 1} \exp\!\left( - \frac{\mu^2}{2 \sigma^2} u \right) \diff u \\
	&= 2^{\frac{1 + \lyapunovExponent}{2}}
		\left( \frac{\mu^2}{2 \sigma^2} \right)^{- \frac{\lyapunovExponent}{2}}
		\int_0^{\frac{\mu^2}{4 \sigma^2}} v^{ \frac{\lyapunovExponent}{2} - 1} \exp( - v) \diff v \\
	&\leq 2^{\frac{1 + \lyapunovExponent}{2}}
		\left( \frac{\mu^2}{2 \sigma^2} \right)^{- \frac{\lyapunovExponent}{2}}
		\int_0^{\infty} v^{ \frac{\lyapunovExponent}{2} - 1} \exp( - v) \diff v \\
	&= 2^{\frac{1}{2} + \lyapunovExponent}
		\left| \frac{\mu}{\sigma} \right|^{- \lyapunovExponent}
		\Gamma\!\left( \frac{\lyapunovExponent}{2} \right).
		\yesnumber \label{eq:kummer_integral_left_half_bound}
\end{align*}
By \eqref{eq:kummer_integral_representation}, \eqref{eq:kummer_integral_right_half_bound}, and \eqref{eq:kummer_integral_left_half_bound}, we obtain
\begin{equation*}
\kummerConfluentHyperGeom\!\left(
	\frac{\lyapunovExponent}{2}, \frac12, - \frac{\mu^2}{2 \sigma^2}
\right)
	\leq \frac{1}{\betaFun\!\left(\frac{\lyapunovExponent}{2}, \frac{1 - \lyapunovExponent}{2} \right)}
		\left[
			\frac{ 2^{\frac{5}{2} - \lyapunovExponent} }{ 1 - \lyapunovExponent }
				\exp\!\left( - \frac{\mu^2}{4 \sigma^2} \right)
			+ 2^{\frac{1}{2} + \lyapunovExponent}
				\Gamma\!\left( \frac{\lyapunovExponent}{2} \right)
				\left| \frac{\mu}{\sigma} \right|^{- \lyapunovExponent}
		\right] \qedhere
\end{equation*}
\end{proof}

\begin{proposition}
\label{prop:kummer_function_properties}
For $b > a > 0$, Kummer's confluent hypergeometric function 1) satisfies the inequality $\kummerConfluentHyperGeom(a, b, z) \leq \max\{1, \exp(z)\}$ and 2) admits the integral representations
\begin{align}
\kummerConfluentHyperGeom(a, b, z)
	&= \frac{ 2^{1 - b} e^{z / 2} }{ \betaFun(a, b - a) }
		\int_{-1}^1 (1 - u)^{b - a - 1} (1 + u)^{a - 1} e^{z u / 2} \diff u
	\label{eq:integ_representation_in_book} \\
	&= \frac{ 1 }{ \betaFun(a, b - a) }
		\int_0^1 (1 - u)^{b - a - 1} u^{a - 1} e^{z u} \diff u.
	\label{eq:integ_representation_in_wikipedia}
\end{align}
\end{proposition}

\begin{proof}
Kummer's function can be represented as the following infinite series (\citealt{gradshteyn2014table_of_integrals}, Section~9.210):
\begin{equation*}
\kummerConfluentHyperGeom(a, b, z)
	= 1 + \frac{a}{b} \frac{z}{1!}
		+ \frac{a (a + 1)}{b (b + 1)} \frac{z^2}{2!}
		+ \frac{a (a + 1) (a + 2)}{b (b + 1) (b + 2)} \frac{z^3}{3!}
		+ \ldots.
\end{equation*}
Since $b > a > 0$,  the series representation immediately implies
\begin{equation}
\label{eq:kummer_bound_for_positive_z}
\kummerConfluentHyperGeom(a, b, z)
	\leq 1 + \frac{z}{1!}
		+ \frac{z^2}{2!}
		+ \frac{z^3}{3!}
		+ \ldots
	= \exp(z).
\end{equation}
for $z \geq 0$. For $z \leq 0$, we first note that
\begin{equation}
\label{eq:kummer_identity}
M(a, b, z)
	= \exp(z) M(b - a, a, - z)
\end{equation}
by the identity (9.212.1) in \cite{gradshteyn2014table_of_integrals}.
Since $b > b - a > 0$ and $- z \geq 0$, we can apply our previous bound \eqref{eq:kummer_bound_for_positive_z} to conclude that $M(b - a, a, - z) \leq \exp(-z)$.
Combined with \eqref{eq:kummer_identity}, this yields $M(a, b, z) \leq 1$ for $z \leq 0$.

The integral representation \eqref{eq:integ_representation_in_book} is given in Section~9.211 of \cite{gradshteyn2014table_of_integrals}.
To obtain \eqref{eq:integ_representation_in_wikipedia}, we apply the change of variable $v = (1 + u) / 2$:
\begin{align*}
\kummerConfluentHyperGeom(a, b, z)
	&= \frac{ 2^{1 - b} e^{z / 2} }{ \betaFun(a, b - a) }
		\int_{0}^1 \left[ 2 (1 - v) \right]^{b - a - 1} (2 v)^{a - 1} e^{z (2 v - 1) / 2}  2 \, \diff v \\
	&= \frac{ 1 }{ \betaFun(a, b - a) }
		\int_0^1 (1 - v)^{b - a - 1} v^{a - 1} e^{z v} \diff v \qedhere
\end{align*}
\end{proof}

\begin{proof}[Proof of Proposition~\ref{prop:marginal_variance_bound}]
A conditional precision (in expectation) is always larger than the marginal one, so we have
\begin{equation*}
\sigma_j^{-2}
	\leq \left( \bSigma^{-1} \right)_{jj}
	= \zeta^{-2} + \gshrink^{-2} \lshrink_j^{-2} + \sum_{i = 1}^n \omega_i x_{ij}^2.
\end{equation*}
Exponentiating both sides of the inequality, we obtain
\begin{align}
\sigma_{j}^{-\lyapunovExponent}
	&\leq \left(
		\zeta^{-2} + \gshrink^{-2} \lshrink_j^{-2} + \sum_{i = 1}^n \omega_i x_{ij}^2
	\right)^{\lyapunovExponent / 2} \nonumber \\
	&\leq  \zeta^{-\lyapunovExponent}
		+ \gshrink^{-\lyapunovExponent} \lshrink_j^{-\lyapunovExponent}
		+ \left( \sum_{i = 1}^n \omega_i x_{ij}^2 \right)^{\lyapunovExponent / 2}
		\label{eq:intermediate_ineq_by_lp_norm_property} \\
	&\leq \zeta^{-\lyapunovExponent}
		+ \gshrink^{-\lyapunovExponent} \lshrink_j^{-\lyapunovExponent}
		+ 1 + \frac{\lyapunovExponent}{2} \left( \sum_{i = 1}^n \omega_i x_{ij}^2- 1 \right),
		\label{eq:intermediate_ineq_by_taylor_expansion}
\end{align}
where \eqref{eq:intermediate_ineq_by_lp_norm_property} follows from the property of $L^\lyapunovExponent$-norm $(|a| + |b|)^{\lyapunovExponent} \leq |a|^\lyapunovExponent + |b|^\lyapunovExponent$ and \eqref{eq:intermediate_ineq_by_taylor_expansion} from the Taylor expansion of the concave function $x \to x^\lyapunovExponent$ at $x = 1$.
\end{proof}

\newcommand{\shape}{s}
\newcommand{\rate}{r}
\section{Properties of Bayesian bridge prior}
\label{sec:bridge_prior_properties}

Bayesian bridge is characterized by the density of $\beta_j \given \gshrink$ given as
\begin{equation}
\label{eq:bridge_marginal}
\pi(\beta \given \gshrink)
	\propto \gshrink^{-1} \exp(-| \beta / \gshrink |^\bridgeExponent ).
\end{equation}
We obtain the global-local representation of \eqref{eq:bridge_marginal} with the conditional
$\beta \given \gshrink, \lshrink \sim \normal(0, \gshrink^2 \lshrink^2)$ when
\begin{equation*}
\localPrior(\lshrink) \propto \lshrink^{-2} \stableDensity(\lshrink^{-2} / 2),
\end{equation*}
where $\stableDensity(\cdot)$ denote the density of the one-sided stable distribution, characterized by location $\mu = 0$, skewness $\beta = 1$, characteristic exponent $\bridgeExponent / 2$, and scale $\stableDistScale = \cos\!\left( \bridgeExponent \pi / 4 \right)^{2 / \bridgeExponent}$ \citep{hofert2011tilted_stable}.
This follows from the Laplace transform identity for the stable distribution:
\begin{equation*}
\begin{aligned}
\exp(- | \beta / \gshrink |^{\bridgeExponent})
	&= \frac12 \int_{0}^\infty \exp(- \phi \beta^2 / 2 \gshrink^2) \, \stableDensity(\phi / 2) \, \diff \phi \\
	&\propto \int_{0}^\infty \normal(\beta; 0, \gshrink^2 \phi^{-1}) \, \pi(\phi) \, \diff \phi,
\end{aligned}
\end{equation*}
for $\pi(\phi) \propto \phi^{-1/2} \pi_{\text{st}}(\phi / 2)$, the density of  $\phi = \lshrink^{-2}$.

We can characterize the behavior of $\localPrior(\lshrink)$ at $\lshrink \approx 0$ from the following asymptotic behavior of the stable distribution as $x \to 0$ \citep{nolan2018stable_distributions}.
\begin{equation*}
\pi_{\text{st}}(x)
 	\sim \frac{1}{x^{(1 + \bridgeExponent)} } \sin\!\left( \varpi \bridgeExponent \right) \frac{\Gamma(\bridgeExponent + 1)}{\varpi}
\end{equation*}
where $\varpi \approx 3.14159$ is Archimedes' constant.
In particular, we have
\begin{equation*}
\localPrior(\lshrink) = O( \lshrink^{2\bridgeExponent} )
\ \text{ as } \lshrink \to 0.
\end{equation*}

The availability of the marginal $\pi(\beta_j \given \gshrink) = \int \normal(\beta_j; 0, \gshrink^2 \lshrink_j^2) \, \localPrior(\lshrink_j) \, \diff \lshrink_j$ allows for a Gibbs update of $\gshrink$ from the posterior with the local scale parameters $\lshrink_j$'s marginalized out.
More precisely, instead of drawing from $\gshrink \given \bbeta, \blshrink$, the Bayesian bridge Gibbs sampler can directly target the conditional
\begin{equation*}
\pi(\gshrink \given \bbeta)
	\propto \Bigg(
		\gshrink^{-p} \prod_{j = 1}^p \exp(-| \beta_j / \gshrink |^\bridgeExponent )
	\Bigg) \globalPrior(\gshrink).
\end{equation*}
Since $\bbeta \given \gshrink$ belongs to the location-scale family, the reference prior is $\globalPrior(\gshrink) \propto \gshrink^{-1}$ \citep{berger2015overall_obayes}, which also happens to be a conjugate prior.
More generally, in terms of the parametrization $\phi = \gshrink^{-\alpha}$, a prior $\phi \sim \textrm{Gamma}(\textrm{shape} = \shape, \textrm{rate} = \rate)$ belongs to a conjugate family, yielding the posterior conditional
\begin{equation*}
\pi(\phi \given \bbeta)
	\sim \textrm{Gamma}\!\left(
		\textrm{shape} = \shape + p, \,
		\textrm{rate} = \rate + \textstyle \sum_{j = 1}^p |\beta_j|
	\right).
\end{equation*}
In the limit $\shape, \rate \to 0$, the gamma prior on $\phi$ recovers the reference prior $\globalPrior(\gshrink) \propto \gshrink^{-1}$ which is invariant under reparametrization,

\newcommand{\boundingDensity}{g}
\newcommand{\unnormDensity}{f}
\newcommand{\multiplicativeConst}{\kappa}
\newcommand{\acceptRate}{A} 
\section{Sampler for local scale posterior under horseshoe prior}
\label{sec:local_scale_rejection_sampler}

Our theoretical results on convergence rate assume the ability to sample independently from the conditionals $\lshrink_j \given \beta_j, \gshrink$ for $j = 1, \ldots, p$.
While not necessarily trivial, this task is typically quite manageable given the wide range of algorithms available to deal with  univariate distributions \citep{devroye2006random_variate_generation, ripley2009stochastic_simulation}.

As an illustration, we present a simple rejection sampler for the conditional $\lshrink_j \given \beta_j, \gshrink$ under the prior $\localPrior(\lshrink_j) \propto 1 / (1 + \lshrink_j^2)$ --- corresponding to the horseshoe prior, arguably the most popular of the existing shrinkage priors \citep{bhadra2017lasso_horseshoe}.
The rejection sampler, as we will show, has uniformly high acceptance probability for all $\beta_j$ and $\gshrink$ with the minimum acceptance probability $\approx 0.6975$ (Figure~\ref{fig:accept_prob_of_rejection_sampler}).
On the precision scale $\eta_j = \lshrink_j^{-2}$, the prior is given by
\begin{equation*}
\localPrior(\eta_j)
	= \localPrior(\lshrink_j) | \diff \lshrink / \diff \eta_j |
	\propto \frac{1}{1 + \eta_j^{-1}} \eta_j^{- 3 / 2}
	= \frac{1}{\eta_j^{1/2} (1 + \eta_j)}.
\end{equation*}
The full conditional $\eta_j \given \beta_j, \gshrink$ has the density
\begin{equation*}
\pi(\eta_j \given \beta_j, \gshrink)
	\propto \localPrior(\eta_j) \, \pi(\beta_j \given \gshrink, \eta_j)
	\propto \frac{1}{1 + \eta_j} \exp\left( - \eta_j \frac{\beta_j^2}{2 \gshrink^2} \right).
\end{equation*}
The task of sampling from the local scale posterior, therefore, boils down to that of sampling from the family of univariate densities
\begin{equation}
\label{eq:local_scale_update_target}
\pi(\eta)
	\propto \frac{1}{1 + \eta} \exp\!\left( - b \eta \right)
	\ \text{ for } \, b > 0.
\end{equation}

To sample from \eqref{eq:local_scale_update_target}, the online supplement of \cite{polson2014bayes_bridge} describes a slice sampling approach and \cite{makalic2015horseshoe_data_augmentation} a data augmentation method.
However, we find that both approaches suffer from slow-mixing as $b \to 0$ and the slow-decaying term $(1 + \eta)^{-1}$ becomes significant (Figure~\ref{fig:mixing_of_lscale_slice_sampler} and \ref{fig:mixing_of_lscale_DA_sampler}).

\begin{figure}
\centering
	\begin{minipage}{.9\linewidth}
	\includegraphics[width=\linewidth]{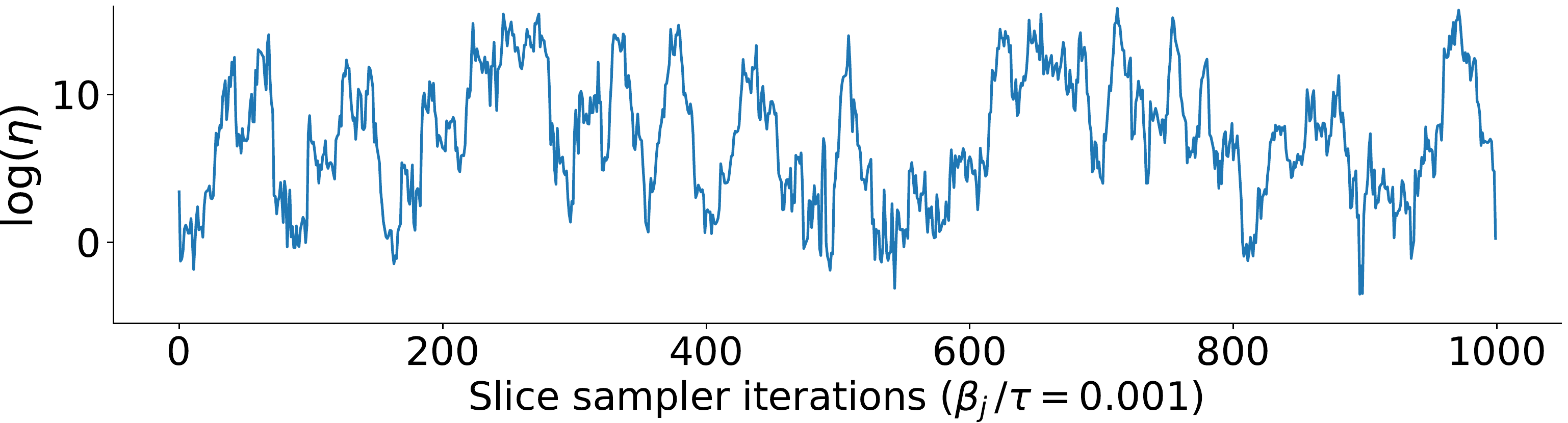}
	\end{minipage}

	\vspace*{\baselineskip}
	\begin{minipage}{.9\linewidth}
	\includegraphics[width=\linewidth]{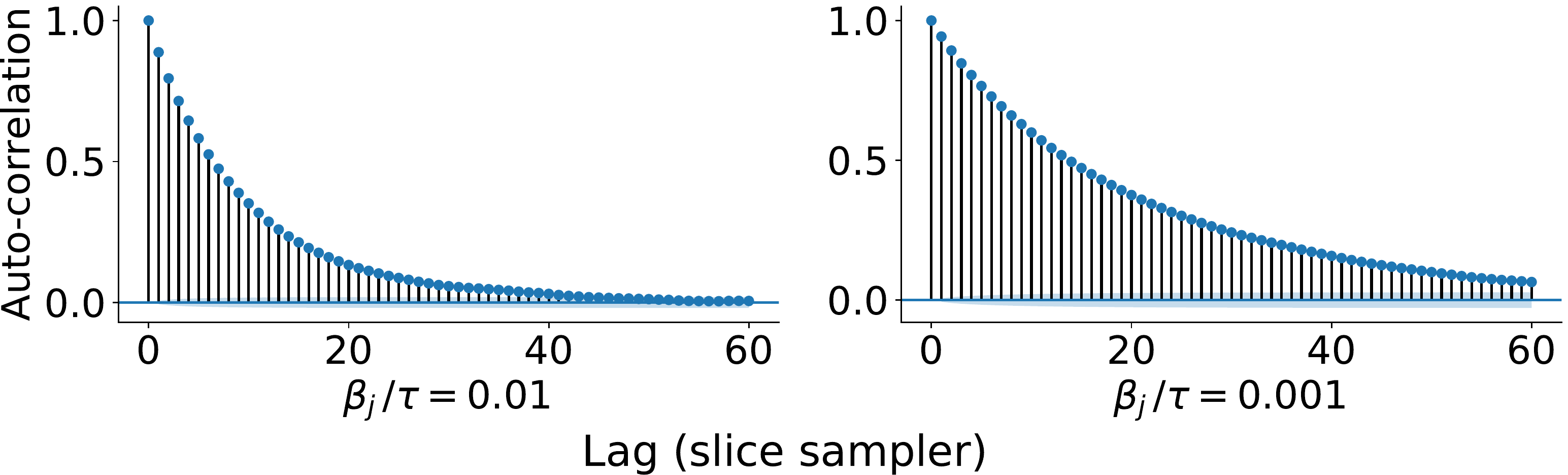}
	\end{minipage}
	\caption{%
		Trace and auto-correlation plots when slice sampling $\eta$ from \eqref{eq:local_scale_update_target} as proposed in \cite{polson2014bayes_bridge}.
		For the two different values of $b = \beta_j^2 / 2 \gshrink^2$, the auto-correlations at stationarity are computed from $10{,}000$ iterations of the sampler to demonstrate how the mixing rate degrades as $b \to 0$. 
	}
	\label{fig:mixing_of_lscale_slice_sampler}
\end{figure}

\begin{figure}
\centering
	\begin{minipage}{.9\linewidth}
	\includegraphics[width=\linewidth]{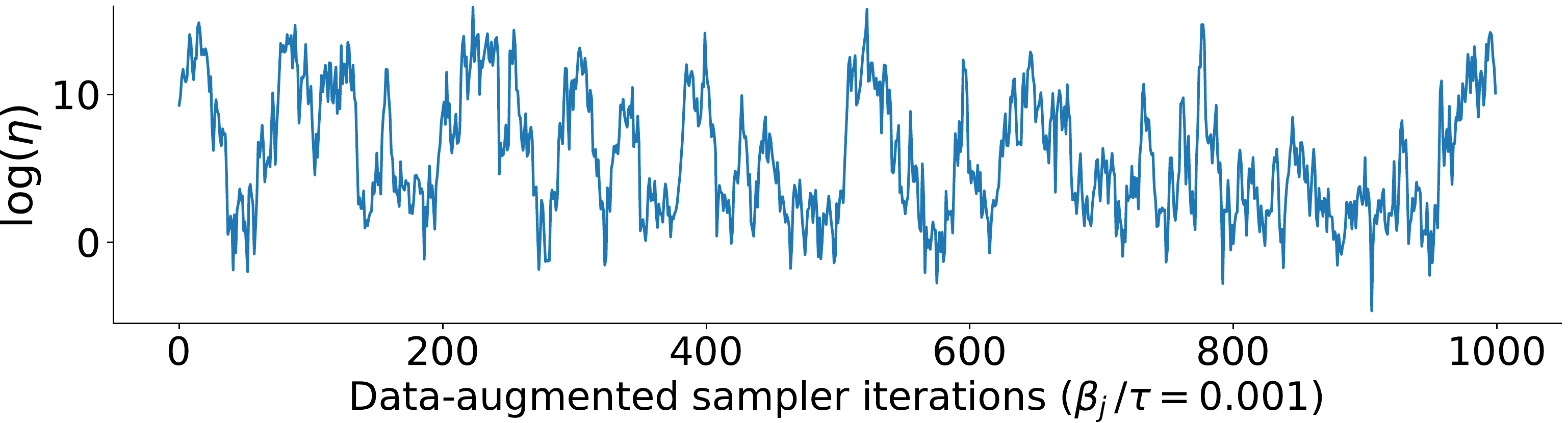}
	\end{minipage}

	\vspace*{\baselineskip}
	\begin{minipage}{.9\linewidth}
	\includegraphics[width=\linewidth]{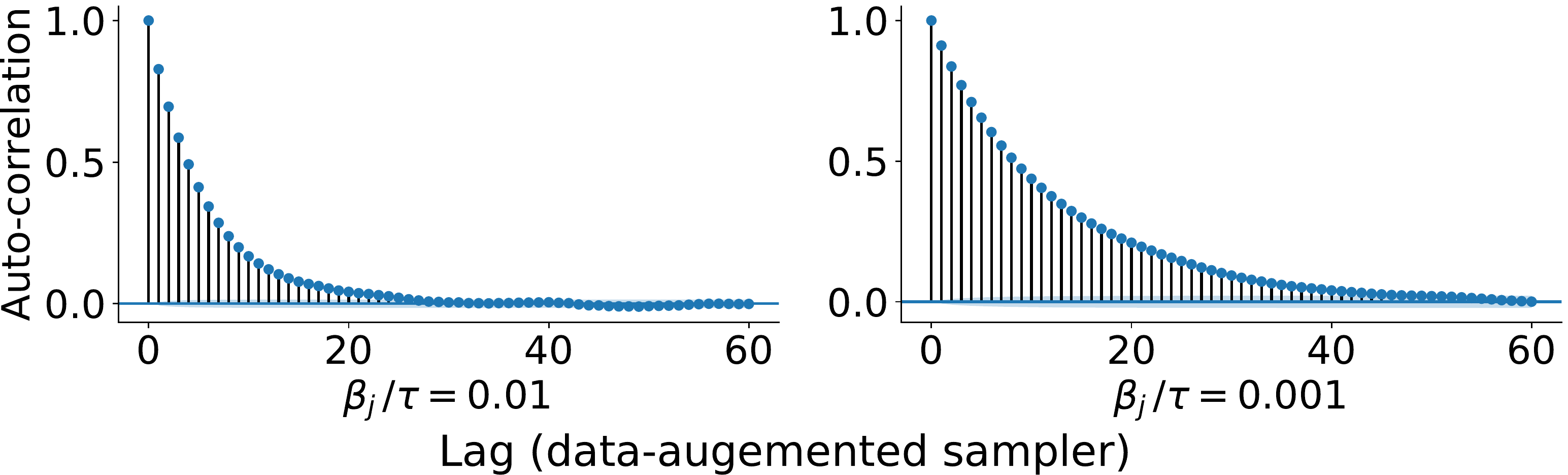}
	\end{minipage}
	\caption{%
		Trace and auto-correlation plots when sampling $\eta$ from \eqref{eq:local_scale_update_target} with the data-augmentation scheme of \cite{makalic2015horseshoe_data_augmentation}.
		The auto-correlations at stationarity are computed from $10{,}000$ iterations of the sampler.
	}
	\label{fig:mixing_of_lscale_DA_sampler}
\end{figure}

\subsection{Rejection sampler algorithm}
\label{sec:local_scale_rejection_sampler_description}
Our rejection sampler acts on a transformed parameter $\psi = \log(1 + \eta)$ that maps back as $\eta = e^\psi - 1$.
The density of $\psi$ is given by
\begin{equation*}
\pi(\psi)
	\propto \pi(\eta) | \diff \eta / \diff \psi |
	= \frac{1}{e^\psi} \exp(- b e^\psi) e^{\psi}
	= \exp(- b e^\psi) \ \text{ on } \psi \geq 0.
\end{equation*}
We now define a function $\boundingDensity_b$ that upper bounds the unnormalized target density
\begin{equation*}
\unnormDensity_b(\psi) := \exp(- b e^\psi).
\end{equation*}
For $b \geq 1$, we set
\begin{equation*}
\boundingDensity_b(\psi) = \exp\{-b (1 + \psi)\},
\end{equation*}
which coincides with an unnormalized density of the distribution $\textrm{Exp}(\text{rate} = b)$.
For $b < 1$, we set
\begin{equation*}
\boundingDensity_b(\psi) = \left\{
	\begin{array}{lr}
	\exp(-b) & \text{ for } \psi \leq \log (1 / b) \\
	\exp\!\left\{
		- 1 - \left( \psi - \log (1 / b) \right)
		\right\} & \text{ for } \psi \geq \log (1 / b)
	\end{array} \right.,
\end{equation*}
which coincides with an unnormalized density of a mixture of $\textrm{Uniform}(0, \log(1/b))$ and $\textrm{Exp}(1)$ shifted by $\log(1/b)$.
To draw a random variable $X$ from this mixture, we set $X \sim \textrm{Uniform}(0, \log(1/b))$ with probability $\log(1/b) \, / \left(\log(1 / b) + e^{b-1} \right)$ and $X - \log(1 / b) \sim \textrm{Exp}(1)$ otherwise.
R and Python code of the rejection sampler are available at \url{https://github.com/aki-nishimura/horseshoe-scale-sampler}.

\subsection{Analysis of acceptance probability}

The acceptance probability of a rejection sampler is given by the ratio of the integrals of the target to the bounding density \citep{ripley2009stochastic_simulation}.
In particular, the rejection sampler described in Section~\ref{sec:local_scale_rejection_sampler_description} has the acceptance probability
\begin{equation}
\label{eq:acceptance_rate_as_ratio}
\acceptRate(b)
	= \frac{\int_0^\infty \unnormDensity_b(\eta) \, \diff \eta}{\int_0^\infty \boundingDensity_b(\eta) \, \diff \eta}.
\end{equation}

\begin{figure}
	\begin{minipage}{.5\linewidth}
	\includegraphics[width=\linewidth]{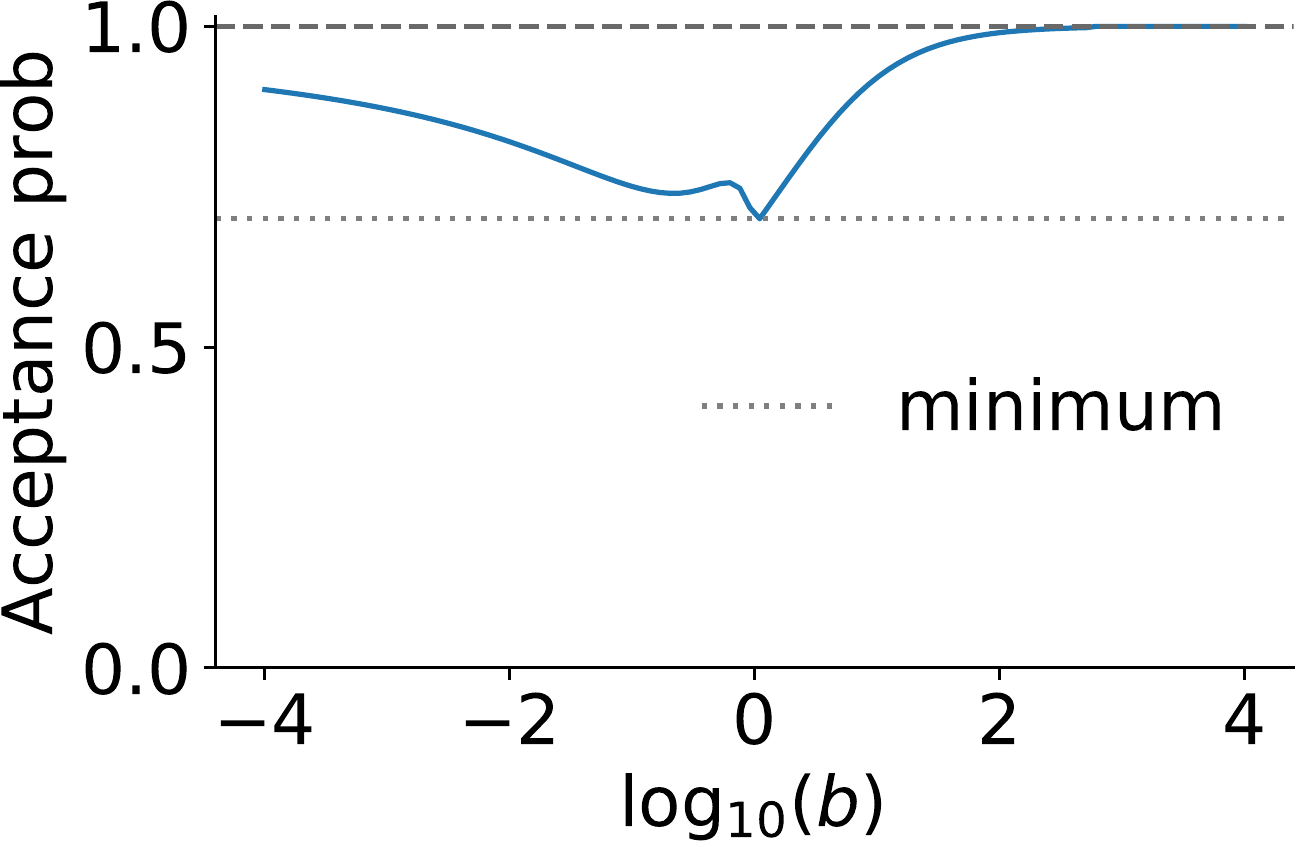}
	\end{minipage}
	~~
	\begin{minipage}{.45\linewidth}
	\caption{%
		Acceptance probability of the proposed rejection sampler as a function of $b = \beta_j^2 / 2 \gshrink^2$.
		The probability is uniformly lower-bounded and increases to 1 as $b \to 0$ and $b \to \infty$ (see Theorem~\ref{thm:rejection_sampler_accept_prob}).
		The minimum probability is $\approx 0.6975$.
	}
	\label{fig:accept_prob_of_rejection_sampler}
	\end{minipage}
\end{figure}

\noindent Figure~\ref{eq:acceptance_rate_as_ratio} plots the acceptance probability $\acceptRate(b)$, evaluated to high accuracy via numerical integration of the integrals in \eqref{eq:acceptance_rate_as_ratio}, and supports the theoretical results below.

\begin{theorem}
\label{thm:rejection_sampler_accept_prob}
The acceptance probability $\acceptRate(b)$ is uniformly lower bounded over $b > 0$ by a positive constant.
Moreover, $\acceptRate(b)$ converges to $1$ as $b \to 0$ and $b \to \infty$.
\end{theorem}

\begin{proof}
We can show that both the denominator and numerator of \eqref{eq:acceptance_rate_as_ratio} depend continuously on $b$, and so does $\acceptRate(b)$, by a simple application of the dominated convergence theorem.
The continuity of $\acceptRate(b)$ implies a uniform lower bound on $b \in (0, \infty)$ as soon as we establish $\acceptRate(b) \to 1$ towards the boundary $b \to 0$ and $b \to \infty$.

We establish a lower bound on the acceptance probability \eqref{eq:acceptance_rate_as_ratio} by explicitly computing the denominator and then lower bounding the numerator.
We first consider the case $b \geq 1$, when the denominator is given by
\begin{equation}
\label{eq:bounding_density_integral}
\int_0^\infty \boundingDensity_b(\eta) \, \diff \psi
	= \int_0^\infty \exp\{-b (1 + \psi)\} \, \diff \psi
	= b^{-1} e^{-b}.
\end{equation}
Then, using Taylor's theorem and the fact $\frac{\diff^2}{\diff \psi^2} e^\psi = e^\psi$, we have
\begin{equation*}
0 \leq e^\psi - (1 + \psi)
	\leq \psi^2 \max_{\psi' \in [0, \psi]} e^{\psi'}
	= \psi^2 e^\psi.
\end{equation*}
The above inequality in particular implies that
\begin{equation}
\label{eq:target_density_lower_bd}
\unnormDensity_b(\psi)
	= \exp(- b e^\psi)
	\geq \exp\{-b (1 + \psi)\} \exp(- b \psi^2 e^\psi).
\end{equation}
We now apply \eqref{eq:target_density_lower_bd} to lower bound the numerator of \eqref{eq:acceptance_rate_as_ratio}; for any $L > 0$,
\begin{equation}
\label{eq:target_density_integral_lower_bd}
\begin{aligned}
\int_0^\infty \exp(- b e^\psi) \, \diff \psi
	&\geq \int_0^{L} \exp\{-b (1 + \psi)\} \exp(- b \psi^2 e^\psi) \, \diff \psi \\
	&\geq \exp(- b L^2 e^L) \int_0^{L} \exp\{-b (1 + \psi)\} \, \diff \psi \\
	&= b^{-1} e^{-b} \exp(- b L^2 e^L) \left( 1 - e^{-bL} \right).
\end{aligned}
\end{equation}
From \eqref{eq:bounding_density_integral} and \eqref{eq:target_density_integral_lower_bd}, we obtain the following lower bound on the acceptance probability, which holds for any $L > 0$:
\begin{equation*}
\acceptRate(b) \geq \exp(- b L^2 e^L) \left( 1 - e^{-bL} \right).
\end{equation*}
Choosing $L = \log(\multiplicativeConst b) / b$ with $\multiplicativeConst > 1$, for example, we obtain the lower bound
\begin{equation}
\label{eq:acceptance_rate_intermediate_lower_bound}
\acceptRate(b) \geq
	\exp\!\left(
		- \frac{(\log \multiplicativeConst b)^2}{b} \multiplicativeConst^{1 / b} b^{1 / b}
	\right) \left( 1 - \frac{1}{\multiplicativeConst b}\right).
\end{equation}
It is straightforward to show that, for example by the derivative test, the function $b \to b^{1 / b}$ has the global maximum $\exp(e^{-1})$ on $b > 0$.
We can therefore simplify the lower bound \eqref{eq:acceptance_rate_intermediate_lower_bound} to
\begin{equation}
\label{eq:accept_rate_lower_bound_for_large_b}
\acceptRate(b) \geq
	\exp\!\left(
		- \exp(e^{-1}) \multiplicativeConst^{1 / b} \frac{(\log \multiplicativeConst b)^2}{b}
	\right) \left( 1 - \frac{1}{\multiplicativeConst b}\right).
\end{equation}
The lower bound in \eqref{eq:accept_rate_lower_bound_for_large_b}, and hence $\acceptRate(b)$, converges to 1 as $b \to \infty$.

We now turn to establishing a lower bound on the acceptance probability in the case $b < 1$.
We have
\begin{equation}
\label{eq:denominator_for_small_b}
\begin{aligned}
\int_0^\infty \boundingDensity_b(\psi) \, \diff \psi
	&= \int_0^{\log(1/b)} e^{-b} \, \diff \psi
	+ \int_{\log(1/b)}^\infty \exp\!\left\{ - 1 - (\psi + \log b) \right\} \diff \psi \\
	&= e^{-b} \log(1/b) + e^{-1}.
\end{aligned}
\end{equation}
To lower bound $\int \unnormDensity_b(\psi) \, \diff \psi$, we first observe that, by the change of variable $\psi' = \psi / \log(1/b)$,
\begin{equation}
\label{eq:lower_bound_for_left_integrand}
\begin{aligned}
\int_0^{\log(1/b)} \exp(- b e^\psi) \, \diff \psi
	&= \log(1/b) \, \const(b)
	\ \text{ where } \ \const(b) = \int_0^1 \exp\!\left( - b^{1 - \psi'} \right) \diff \psi'.
\end{aligned}
\end{equation}
On the interval $\psi' \in [0, 1)$, the integrand converges to 1 as $b \to 0$ and hence the dominated convergence theorem implies $\const(b) \to 1$ as $b \to 0$.
On the interval $\psi \in [\log(1 / b), \infty)$, we have
\begin{equation}
\label{eq:lower_bound_for_right_integrand}
\begin{aligned}
&\int_{\log(1/b)}^\infty \exp(- b e^\psi) \, \diff \psi \\
	&\hspace{4em}= \int_{\log(1/b)}^\infty \exp\!\left\{
		- b e^{\log(1/b)} e^{\psi - \log(1/b)} \right\} \, \diff \psi \\
	&\hspace{4em}= \int_{0}^\infty \exp(- e^\psi)  \, \diff \psi \\
	&\hspace{4em}\geq e^{-1} C'(\multiplicativeConst)
		\ \text{ for } \
			C'(\multiplicativeConst)
				= \exp\!\left(
					- (\log \multiplicativeConst)^2 \multiplicativeConst
				\right) \left( 1 - \frac{1}{\multiplicativeConst}\right),
\end{aligned}
\end{equation}
where the last inequality follows from \eqref{eq:target_density_integral_lower_bd} with $b = 1$ and $L = \log(\multiplicativeConst)$ for $\multiplicativeConst > 1$.
It follows from \eqref{eq:denominator_for_small_b}, \eqref{eq:lower_bound_for_left_integrand}, and \eqref{eq:lower_bound_for_right_integrand} that for $b < 1$
\begin{equation}
\label{eq:acceptance_rate_lower_bound_for_small_b}
\acceptRate(b)
	\geq \frac{
		\log(1/b) \, \const(b) + e^{-1} C'(\multiplicativeConst)
	}{
		e^{-b} \log(1/b) + e^{-1}
	},
\end{equation}
where $\lim_{b \to 0} C(b) = 1$ and $C'(\multiplicativeConst) \approx 0.264$ for $\kappa = 1.57$.
The lower bound in \eqref{eq:acceptance_rate_lower_bound_for_small_b}, and hence $\acceptRate(b)$, converges to 1 as $b \to 0$.
\end{proof}

\end{document}